\documentclass[10pt,journal]{IEEEtran}
\makeatletter
\def\ps@headings{%
\def\@oddhead{\mbox{}\scriptsize\rightmark \hfil \thepage}%
\def\@evenhead{\scriptsize\thepage \hfil \leftmark\mbox{}}%
\def\@oddfoot{}%
\def\@evenfoot{}}
\makeatother 

\usepackage{epsfig}
\usepackage{url}
\usepackage{hyperref}
\hypersetup{hidelinks}

\usepackage{algorithmic}
\usepackage{algorithm}
\usepackage{color}
\usepackage{soul}
\sethlcolor{yellow}
\usepackage{amsmath}
\usepackage{ntheorem}
\usepackage{amsfonts}
\usepackage{amssymb}
\usepackage{graphicx}
\usepackage{comment}
\usepackage{subfig}
\usepackage{multirow}
\usepackage{bm}
\usepackage{graphicx}
\usepackage{epsfig}
\usepackage{epstopdf}
\usepackage{threeparttable}
\graphicspath{{Figs/}}

\newtheorem{definition}{Definition}

\newtheorem{theorem}{Theorem}
\newtheorem*{proof}{Proof}

\newcommand{\warn}[1]{}

\def\IEEEkeywordsname{Keywords}
\def\baselinestretch{1}
\begin{document}

\title{DPCrowd: Privacy-preserving and Communication-efficient Decentralized Statistical Estimation for Real-time Crowd-sourced Data}

\author{\IEEEauthorblockN{Xuebin Ren, Chia-Mu Yu, Wei Yu, Xinyu Yang, Jun Zhao, and Shusen Yang}
\IEEEcompsocitemizethanks{
\IEEEcompsocthanksitem X. Ren, X. Yang, S. Yang are with Xi'an Jiaotong University. (\{xuebinren, yxyphd, shusenyang\}@mail.xjtu.edu.cn)
\IEEEcompsocthanksitem C. Yu is with National Chiao Tung University. (chiamuyu@gmail.com).
\IEEEcompsocthanksitem W. Yu is with Towson University. (wyu@towson.edu)
\IEEEcompsocthanksitem J. Zhao is with Nanyang Technological University. (junzhao@ntu.edu.sg)
}}

\maketitle

\begin{abstract}

In Internet of Things (IoT) driven smart-world systems, real-time crowd-sourced databases from multiple distributed servers can be aggregated to extract dynamic statistics from a larger population, thus providing more reliable knowledge for our society. Particularly, multiple distributed servers in a decentralized network can realize real-time collaborative statistical estimation by disseminating statistics from their separate databases. Despite no raw data sharing, the real-time statistics could still expose the data privacy of crowd-sourcing participants. For mitigating the privacy concern, while traditional differential privacy (DP) mechanism can be simply implemented to perturb the statistics in each timestamp and independently for each dimension, this may suffer a great utility loss from the real-time and multi-dimensional crowd-sourced data. Also, the real-time broadcasting would bring significant overheads in the whole network. To tackle the issues, we propose a novel privacy-preserving and communication-efficient decentralized statistical estimation algorithm (\textsf{DPCrowd}), which only requires intermittently sharing the DP protected parameters with one-hop neighbors by exploiting the temporal correlations in real-time crowd-sourced data. Then, with further consideration of spatial correlations, we develop an enhanced algorithm, \textsf{DPCrowd}+, to deal with multi-dimensional infinite crowd-data streams. Extensive experiments on several datasets demonstrate that our proposed schemes \textsf{DPCrowd} and \textsf{DPCrowd}+ can significantly outperform existing schemes in providing accurate and consensus estimation with rigorous privacy protection and great communication efficiency.



\end{abstract}

\IEEEpeerreviewmaketitle

\begin{IEEEkeywords}
Differential privacy, decentralized statistical estimation, real time, communication efficiency, crowd-sourced data
\end{IEEEkeywords}

\section{Introduction}\label{sec: Introduction}


With the proliferation of smart devices and communication technologies, massive crowdsensing data can be acquired in real time, including industrial data in Industrial Internet-of-Things (IIoTs)~\cite{zheng2020privacy,8730298}, proximity sensing data in Internet-of-Vehicles (IoVs)~\cite{7588230,lin2017survey}, and IoT-based health data~\cite{huang2019practical,WH16}. The aggregate statistics~\cite{braverman2016communication,jordan2019communication} of these real-time data can provide valuable knowledge (e.g., popular business sites, disease outbreaks, and traffic dynamics~\cite{rescuedp2016,wang2018privacy}) and facilitate intelligence for numerous smart-world systems, including smart industry or smart cities~\cite{zanella2014internet,lin2017survey}. Nonetheless, these data could be crowd-sourced and stored at peer organizations (e.g., companies and hospitals) or edge servers (e.g., smart vehicles)~\cite{li2018privacy} as isolated data silos, which are difficult to be thoroughly aggregated and fully utilized~\cite{sonehara2011isolation}. Therefore, it is essential to help multiple distributed parties to achieve real-time statistical analysis (or statistical parameter estimation) from their separately crowd-sourced databases.

Different from conventional statistical parameter estimation that relies on a central server to process all crowd-sourced data~\cite{zhang2018reap,wang2018privacy}, these peer servers often belong to no central entity but equal to each other, thus can only form a decentralized network with no mutual trust~\cite{asadi2014survey,7588230}, as shown in Fig.~\ref{systemmodel}. Distributed or decentralized statistical parameter estimation~\cite{stankovic2010decentralized,kar2012distributed,braverman2016communication,jordan2019communication} has been studied in wireless sensor networks to infer the environment parameters by sharing intermediate statistics, which, however, may still expose the sensitive information. Particularly, for the distributed crowd-sourcing servers, statistics sharing among multiple servers may disclose the sensitive information of crowd-source users or provide extra information for malicious or adversary compromised servers~\cite{zheng2020privacy}.

Differential privacy (DP), as the de-facto paradigm for privacy preservation with rigorous guarantee~\cite{DworkRoth-77,hassan2019differential}, has received considerable attention in the privacy protection of monitoring or crowd-sourced data, focusing on either data publication~\cite{su2016differentially,yang2017survey,ren2018textsf,wang2019locally} or statistical aggregation~\cite{fan2014adaptive,rescuedp2016,wang2018privacy,WangZYLYRS19}. Nonetheless, most of the existing works are considered in the context of single-server application~\cite{fan2014adaptive,rescuedp2016}, or rely on a central coordinator~\cite{wang2018privacy,su2016differentially}, or only achieve one-time data publication~\cite{su2016differentially,ren2018textsf,wang2019locally}, or conduct multiple rounds of computation while suffering from severe privacy degradation~\cite{HuangMitra-1130,huang2015differentially,8260919}. None of them can be directly adopted for our application scenario, in which fully decentralized servers conduct real-time statistical estimation without any central entity. Thus, this motivates us to design a novel differential private and communication efficient framework of real-time statistical estimation from crowd-sourced data stored at multiple distributed servers in a fully decentralized network.

\textbf{Design Challenges.} The main challenges in developing such a framework with DP can be summarized as follows.

\begin{itemize}

\item \textit{Huge communication cost}. To achieve consensus estimation for distributed servers in a decentralized network, a straightforward method is to let each server release its own aggregate statistics to all other servers hop by hop at each timestamp. Nonetheless, besides privacy concerns, continuous multi-hop broadcast incurs both tremendous communication overhead and high delay.

\item \textit{Real-time data release}. The global information often needs to be derived in a real-time fashion (e.g., the traffic condition or epidemic disease outbreak). Nonetheless, according to the sequential composition theorem of DP~\cite{Mcsherry-39}, naive DP protection on continuous data stream causes severe utility loss or extravagant privacy budget consumption~\cite{dwork2010differentialco}.

\item \textit{Multi-dimensional Data}. The aggregate statistics may be multi-dimensional, reflecting different aspects of the environment. Nonetheless, with the increase of data dimensions, the data stream would be sparse and lead to both high computational complexity and low data utility for many existing privacy-preserving algorithms~\cite{rescuedp2016}. 

\end{itemize}

\textbf{Contributions.} Our contributions are summarized below.

\begin{itemize}

    \item We propose \textsf{DPCrowd}, an efficient framework of real-time differentially private decentralized statistical estimation for multiple distributed servers with separately crowd-sourced datasets. To the best of our knowledge, this is the first work realizing real-time decentralized statistical estimation with both privacy protection and communication efficiency.

    \item We leverage the Laplace mechanism and Kalman-consensus information filter to realized privacy protection and communication reduction for real-time decentralized statistical estimation with fast convergence and consensus estimation. By further adopting adaptive sampling based intermittent communication strategy, \textsf{DPCrowd} can achieve statistical estimation with much higher utility privacy tradeoff and lower communication cost.

    \item Based on \textsf{DPCrowd}, we further present \textsf{DPCrowd}+ to deal with multi-dimensional infinite data streams. \textsf{DPCrowd}+ satisfies $w$-event DP for infinite streams and mitigates the sparsity issue in multi-dimensional data, thus further enhancing the utility for statistical estimation on multi-dimensional streams.

    \item We conduct extensive experiments on both synthetic and real-world datasets. The experimental results demonstrate that \textsf{DPCrowd} and \textsf{DPCrowd}+ can not only achieve superior estimation accuracy under the given privacy guarantees, but also offer desirable estimation consensus with low communication cost.

\end{itemize}

The remainder of this paper is organized as follows: In Section~\ref{sec:related}, we conduct a brief literature review of related works. In Section~\ref{sec: System Model}, we introduce models and formalize the problem. In Section~\ref{sec: Preliminaries}, we provide some preliminaries. In Section~\ref{sec: approach}, we introduce our baseline and enhanced schemes. In Section~\ref{sec:analysis2}, we conduct the privacy, utility, communication latency and cost analysis of our schemes. In Section~\ref{sec:evaluation}, we describe the performance evaluation results. Finally, we conclude the paper in Section~\ref{sec:conclusion}.

\section{Related Work}\label{sec:related}

In the following, we review some works that are relevant to our study.

\textbf{DP for Data Stream Publication}. Dwork {\em et al.} initiated the theoretical study of DP~\cite{Dwork-405} on streaming data release (or publication)~\cite{dwork2010differential,dwork2010differentialco}. They proposed two DP notions, namely event-level and user-level DP. The former hides a single event and the latter hides all the events of any user. Mir {\em et al.}~\cite{mir2011pan} studied estimating distinct counts, moments, and heavy hitters, which is also studied by Chan {\em et al.}~\cite{chan2012differentially}. In addition, Fan {\em et al.}~\cite{fan2014adaptive} presented FAST to achieve DP aggregate monitoring in the sampling-and-filtering framework. Chen {\em et al.}~\cite{chen2017pegasus} presented PeGaSus to achieve event-level DP in the framework of perturb-group-smooth. Likewise, Kellaris {\em et al.}~\cite{kellaris2014differentially} addressed the shortcoming of event-level DP and user-level DP, and proposed a new notion of $w$-event DP, which can be thought of as a sliding window version of DP on the infinite data stream. Wang {\em et al.}~\cite{rescuedp2016} proposed RescueDP by applying the idea of $w$-event DP to FAST. Beyond that, the authors further enhanced RescueDP with advanced techniques, such as recurrent neural network in time-series analysis and dynamic programming for dynamic grouping, which demonstrate much better performance~\cite{wang2016tdsc}. All these studies are considered in the context of a single-server application.

\textbf{DP for Distributed Data Publication}. Most above data publication studies focus on streaming data in a centralized setting and are not practical for the distributed scenarios. Early studies attempt to achieving DP via adding partial noise at distributed servers~\cite{acs2011have}. For example, Goryczka {\em et al.}~\cite{7286780} conducted a comparative study on secure data aggregation with DP in a distributed setting. Alhadidi {\em et al.}~\cite{AlhadidiMohammed-1153} proposed to privately publish horizontally partitioned data with integration of DP and secure multi-party computation. Hong {\em et. al.}~\cite{hong2015collaborative} proposed collaborative generation algorithms for search logs at different parties with $(\varepsilon, \delta)$-DP. Su {\em et al.}~\cite{su2016differentially} presented a DP solution to publishing high-dimensional, but vertically split data in a distributed setting. Nonetheless, these schemes mainly deal with static data. Further, Wang {\em et al.}~\cite{wang2018privacy} rebuilt RescueDP~\cite{rescuedp2016}\cite{wang2016tdsc} and proposed a distributed framework of DADP by introducing multiple agents between the crowd-sourcing users and the central server. Nonetheless, it still relies on the coordination of a central server and is not fully decentralized. Beyond the above studies, local differential privacy (LDP) has also been a promising paradigm for large-scale crowd-sourcing systems for various applications~\cite{Erlingsson-2014,ren2018textsf}.

\textbf{DP for Distributed Parameter Estimation}. There have been a few studies on private distributed or decentralized parameter estimation recently. For example, Belmega {\em et al.}~\cite{BelmegaSankar-1182} explored an information-theoretic approach to obtain the state estimation between two parties with privacy. Huang {\em et al.}~\cite{HuangMitra-1130,huang2015differentially} proposed a class of iterative algorithms for solving the private distributed optimization problem.
Recently, a variety of privacy-preserving distributed (collaborative) learning or federated learning (FL)~\cite{truex2019hybrid,zhao2019privacy,geyer2017differentially,ZhangZhu-1222,wang2019beyond} approaches have emerged as new solutions to privately learn from distributed datasets. For example, 
Geyer {\em et al.}~\cite{geyer2017differentially} proposed to achieve client level DP for distributed FL clients by injecting noise to the aggregated update models of distributed clients, where moment accountant mechanism is also used for tightly tracking the privacy loss. Truex {\em et al.}~\cite{truex2019hybrid} combined both techniques of DP and secure multiparty computation to reduce the noise growth while maintaining effective privacy guarantee. Likewise, Zhao {\em et al.}~\cite{zhao2019privacy} achieved privacy-preserving distributed collaborative deep learning via not only running privacy-preserving stochastic descent gradient independently on distributed datasets using object perturbation on loss function, but also privately selecting reliable participants via the exponential mechanism. These methods can allow massive distributed data utilization with privacy preservation, which, however, are mainly considered in a batch learning scenario instead of streaming setting. To address this issue, Li {\em et al.}~\cite{8260919} presented a distributed online learning framework with DP. Nonetheless, temporal correlations in the dynamic estimation were hardly considered in these studies. 


Unlike the above studies, we aim to design a privacy-preserving and communication efficient framework of real-time decentralized statistical estimation for multiple distributed servers with crowd-sourced data streams, which can be widely used in IoT-driven smart-world systems. There is a controversy~\cite{li2016differential,cao2018quantifying,mcsherrypost1} over what does DP guarantee for correlated data streams due to different understanding of privacy definition. Some recent works~\cite{cao2018quantifying,song2017pufferfish} suggested that DP offers a weaker bound on privacy loss when data records are correlated. Nonetheless, similar to~\cite{chen2017pegasus}, \textit{in this paper, we emphasize to design privacy-preserving mechanisms based on general DP definitions~\cite{dwork2010differentialco} with privacy parameter $\varepsilon$ while minimizing the statistical estimation error.}

\section{Models and Problem Definition}\label{sec: System Model}

In this section, we first introduce system model, communication model, data model, data model, adversary model and then present the problem definition.

\begin{figure}[htbp]\vspace{-0.2cm}
	\centering\epsfig{file=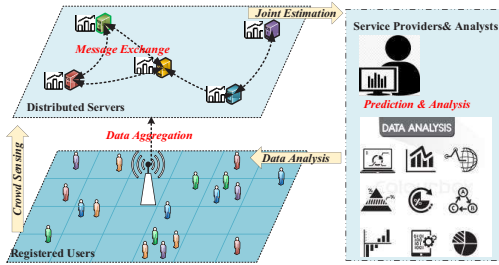, width=0.4\textwidth}
	\caption{Decentralized Statistical Estimation from Crowd-sourced Data}\label{systemmodel}
\vspace{-0.2cm}
\end{figure}

\textbf{System Model}. As shown in Fig.~\ref{systemmodel}, we consider there are $m$ distributed servers that provide geo-location services to a population of $n$ crowd-sourcing users $\{1, \dots, n\}$ scattered in an area, which is divided into $d$ disjoint regions. Each time $t$, each user randomly registers at one of the distributed servers, and uploads the check-in information with the secure connection technology. Each server SP$_i~(i=1, \dots, m)$ then collects the crowd-sourced data from its corresponding user group $G_i (t) \subseteq  \{1, \dots, n\}$ with the population of $|G_i (t)|$. Assume that all the users generally follow the same mobility model (i.e., the same transition probability from one region to another) when the regions are coarsely divided. In this paper, these servers are considered to be connected in a decentralized network and interested in the real-time population distribution among regions.


\textbf{Communication Model}. Though mostly static, we assume a general scenario, in which the communication network among distributed servers is dynamic and evolves with time (e.g., vehicular networks). We assume all servers communicate with each other based on a $m \times m$ time-variant adjacent matrix $\mathcal{E}(t)$. We abstract the communication graph for $i^{th}$ server at timestamp $t$ as
\begin{align}\label{eq:graph}
\mathcal{G}_i (t)=\{(i,j):e_{ij}(t) \in \mathcal{E}(t)\}.
\end{align}
Here, the element $e_{ij}(t)=1$ means that there exists a communication between server $i$ and $j$ at the timestamp $t$ while $e_{ij} (t)=0$ means that no communication between them. We assume that at every timestamp, the graph has no isolated server; i.e., for each $i=1, \ldots, m$, there exists $j\neq i$ such that $e_{ij}(t)=e_{ji}(t)=1$.


\textbf{Data Model}. Let $D_t$ be a two-dimensional matrix with the size of $n\times d$ at timestamp $t$.
Denote $D_{i,t}=[D_{i,t}[u_{1}^i]^T \dots D_{i,t}[u_{|G_i (t)|}^i]^T]^T$ as the two-dimensional database at the $i^{th}$ server SP$_i$ at time $t$, with the size of $|G_i (t)| \times d$, where $D_{i,t}[u_{j}^i]~(1 \leq j \leq |G_i (t)|)$ denotes the $j^{th}$ row of $D_{i,t}$, which corresponds to SP$_i$'s registered users $u^i_j \in G_i (t)$. In $D_{i,t}$, each row corresponds to a registered user in $G_i (t)$ and each column corresponds to a region. The value of $D_{i,t}(p,q)$ is $1$ refers to the case that the $p^{th}$ user in $G_i (t)$ appears at the $q^{th}$ region at time $t$, and $0$ otherwise. Since any user can appear at exactly one region at the same time, each row in $D_{i,t}$ also contains at most one $1$.

\textbf{Adversary Model}. We focus on the data privacy of crowd-sourcing users in decentralized statistical estimation. We assume the crowd-sourcing users trust on the distributed servers, at which they registered. This assumption is common as the users may have to contribute their data to the servers for certain personalized services, i.e., online recommendation or real-time navigation. Nonetheless, they wish better service quality while minimizing their privacy risks. Thus, each user would consider his/her unregistered servers or any third-party analysts are potential \textit{honest-but-curious} adversaries, which honestly follow the mechanism, but try to infer the private information from his/her register server. This adversary model is practical in a decentralized network where distributed servers belong to different individuals or organizations, which do not have mutual trust. 

\textbf{Problem Definition}. Let $\mathbf{r}(t)=f(D_t)=(r^1 (t), r^2 (t), \ldots, r^d (t))$ denote the true statistics (e.g., $r^k (t)$ denote the total number of users) over $D_t$ at time $t$ in the $k^{th}$ region $(k=1,2, \ldots, d)$, where $f$ is an aggregate function (e.g., sum) applied to all $d$ dimensions. However, $\mathbf{r}(t)$ cannot be accurately obtained by any distributed server with consensus, which only has partial knowledge of all crowd-sourcing users. In particular, the $i^{th}$ server can only aggregate its registered users' data to produce its own aggregate statistics $\mathbf{x}_i (t)=f(D_{i,t})$ and may directly estimate $\mathbf{r}(t)$ from $\mathbf{x}_i (t)$. Nonetheless, due to partial samples and no coordination, the estimation would be rather rough and vary extravagantly among the distributed servers. Thus, collaborative estimation from multiple distributed servers seems to be promising. Nonetheless, severe privacy risks under the above adversary model and communication overheads in a time-varying decentralized network may still prevent the collaborations among servers. To further encourage active collaboration, the distributed servers can also be rewarded with incentives mechanisms based on their contributions recorded by some distributed ledgers (such as Blockchain)~\cite{kang2019incentive}. However, this is beyond our focus of privacy preservation in this paper.

Therefore, based on aforementioned system models and assumptions, our problem can be formalized as:
\textit{with the partial stream datasets $D_{1,t}, D_{2,t}, \ldots, D_{m,t}$ at $m$ distributed servers in a time-varying decentralized network $\mathcal{G}_i (t)$, we focus on helping the mutually untrusted distributed servers to communication-efficiently estimate the accurate overall statistics $\mathbf{r}(t)$ in real-time with consensus while guaranteeing differential privacy for crowd-sourcing users registered at each distributed server.}

\section{Preliminaries}\label{sec: Preliminaries}

In this section, we provide some background about the notion of differential privacy (DP), DP on data streams, as well as Kalman-Consensus information filter.

\subsection{Differential Privacy and Laplace Mechanism}\label{sec: Differential Privacy for Static Datasets}

Differential privacy (DP) is a de-facto standard for data privacy. The rationale behind DP is that adding or removing any single data record will not have much influence on query results on the dataset. A formal definition of DP~\cite{DworkRoth-77} is given below.
\begin{definition}
\textbf{$\varepsilon$-DP~\cite{DworkRoth-77}:} A randomized mechanism $\mathcal{M}$ satisfies $\varepsilon$-DP if for any two neighboring datasets $D$ and $D'$ that differ at most one data record, and for any possible outputs $O \subseteq Range(\mathcal{M})$,
\begin{align}\label{eq: 1}
Pr\left[\mathcal{M}(D) \in O \right]\leq e^{\varepsilon} \cdot Pr\left[\mathcal{M}(D') \in O \right],
\end{align}
where the probability is taken over $\mathcal{M}'s$ randomness. \emph{Privacy budget} $\varepsilon$ is a parameter for the tradeoff between privacy and utility. From Eq. (\ref{eq: 1}), we see that smaller $\varepsilon$ means better privacy but lower utility.
\end{definition}
\begin{definition}
\textbf{Sensitivity~\cite{DworkRoth-77}:} For any function $f: \mathcal{D} \rightarrow \mathcal{R}^d$, the sensitivity of $f$ w.r.t $\mathcal{D}$ is defined as
\begin{align}
\Delta f=\max_{D, D' \in \mathcal{D}}\| f(D)-f(D')\|
\end{align}
for all $D$ and $D'$ that differs on at most one record.
\end{definition}
\textbf{Laplace mechanism} is the most popular scheme for DP, which adds carefully calibrated noise to query results~\cite{DworkRoth-77}. In particular, the noise $\textcolor{blue}{\upsilon}$ follows a zero-mean Laplace distribution $\mathcal{L}(b)$ with scale parameter $b$, which has the probability density function
\begin{align}
P(\textcolor{blue}{\upsilon}|b)=\frac{1}{2b}\mathrm{exp}({-\frac{|\textcolor{blue}{\upsilon}|}{b}}).
\end{align}
\begin{theorem}\label{theorem: LM}
(Laplace Mechanism~\cite{DworkRoth-77}) For any function $f: \mathcal{D} \rightarrow \mathcal{R}^d$ on any dataset $D \in \mathcal{D}$, the Laplace Mechanism $\mathcal{M}$ that adds Laplace noise $\left \langle \upsilon_1,\ldots,\upsilon_d\right \rangle$ to the function output, i.e.,
\begin{align}
\mathcal{M}(D)=f(D)+\left \langle \upsilon_1,\ldots,\upsilon_d\right \rangle
\end{align}
satisfies $\varepsilon$-DP, where $\upsilon_k$ for $k=1,\ldots,d$ is drawn from Laplace distribution $\mathcal{L}(\Delta f/\varepsilon)$ with $\Delta f$ as the sensitivity of $f(\cdot)$ and $\varepsilon$ as the privacy budget.
\end{theorem}


DP enjoys the following two useful properties~\cite{DworkRoth-77}.

\begin{theorem}\label{theorem: Sequential Composition}
(Sequential Composition~\cite{DworkRoth-77}). Let $\mathcal{M}_1, \ldots, \mathcal{M}_T$ be $T$ randomized mechanisms, each of which satisfies $\varepsilon_t$-DP. A sequence of mechanisms $\mathcal{M}_t$ over a database $D$ will guarantee $\sum \varepsilon_t$-DP.
\end{theorem}

\begin{theorem}\label{theorem: Post-Processing}
(Post-Processing~\cite{DworkRoth-77}) Let $\mathcal{M}$ be a randomized mechanism satisfying $\varepsilon$-DP and $f$ be an arbitrary function. Then, $f(\mathcal{M}(D))$ will still guarantee $\varepsilon$-DP.
\end{theorem}

\subsection{DP on Data Streams}\label{sec: Differential Privacy for Data Stream}

The most straightforward DP notion for data streams is event-level DP for infinite streams~\cite{dwork2010differential,dwork2010differentialco}, which aims to protect the presence of a particular event at the time $i$ in a stream with unlimited length. Another is user-level DP for finite streams, guaranteeing that the presence of any user is indistinguishable in an entire data stream during certain period. Event-level DP is weaker than user-level DP as it does not consider the correlation among events in consecutive timestamps. Nonetheless, user-level DP for finite streams may restrict many interruptible real-time applications that generate infinite streams, whereas event-level DP is sometimes insufficient. Thus, $w$-event privacy~\cite{kellaris2014differentially}, an approximate user-level in a sliding window of $w$ continuous timestamps, is proposed as an alternative DP definition for data streams. In this paper, we try to cover both user-level DP for finite streams and $w$-event privacy for infinite streams.

Before giving the definition of $w$-event DP, we first introduce the definition of $w$-neighboring, which describes two streams differs in a window of $w$ timestamps. For an infinite data stream $S=[D_1, D_2, \ldots]$, we define its stream prefix at timestamp $t$ as $S_t=[D_1, D_2, \ldots, D_t]$.

\begin{definition}
($w$-neighboring). For any positive integer $w$, two stream prefixes $S_t$, $S_t '$ are defined as $w$-neighboring, if
\begin{enumerate}
  \item for each $S_t [i]$, $S_t ' [i]$ such that $i \in [t]$ and $S_t [i] \neq S_t '[i]$, it holds that $S_t [i]$, $S_t '[i]$ are neighboring;
  \item for each $S_t [i_1]$, $S_t  [i_2]$, $S_t ' [i_1]$, $S_t ' [i_2]$ with $i_1<i_2$, $S_t [i_1]\neq S_t '[i_1]$ and $S_t [i_2]\neq S_t '[i_2]$, it holds that $i_2-i_1 +1 \leq w$.
\end{enumerate}
\end{definition}

\begin{definition}\label{def:w-event}
($w$-event $\varepsilon$-DP). A mechanism $\mathcal{M}$ is $w$-event $\varepsilon$-DP, if for the given integer $w$, all output sets $O \subseteq Range(\mathcal{M})$ and all $w$-neighboring stream prefixes $S_t$, $S_t '$ with all $t$, it satisfies that
\begin{align}
Pr[\mathcal{M}(S_t)\in O] \leq e^\varepsilon \cdot Pr[\mathcal{M}(S_t ')\in O].
\end{align}
\end{definition}

\subsection{Kalman-Consensus Information Filter}

Kalman filter is an effective algorithm for estimating dynamic processes that contain statistical noise. In particular, an underlying dynamic process with noise can be formulated by a linear time-varying model (aka. process model)
\begin{equation}\label{eq:state}
r(t+1)=A(t) \cdot r(t)+\omega(t),
\end{equation}
where $r(t)$ is the process state at time $t$ ($r(0)$ is an initial state with a normal distribution $N(\overline{r}(0), P_0)$), $\omega(t)$ is the noise sampled from a normal distribution $N(0, Q_t)$, and $A(t)$ is the transition matrix that describes the transitions of the process.

In a distributed network, each node $i$ can have an observation $x_i$ of the dynamic process with the following linear sensing model (aka. measurement model)
\begin{equation}\label{eq:sensing}
x_i (t)=H_i (t) \cdot r(t)+v_i (t),
\end{equation}
where $H_i (t)$ is the observation matrix and $v_i (t)$ is the measurement noise assumed to follow a normal distribution $N(0, R_t)$.

The Kalman filter can be used for each node to estimate the true $r(t)$ independently. We denote $\hat{x}_i (t)$ and $\overline{x}_i (t)$ as estimate and prior estimate of $r(t)$, respectively, for node $i$. Then, the estimate $\hat{x}_i (t)$ of $r(t)$ can be given as a linear combination of the prior estimate $\overline{x}_i(t)$ and the measurement $x_i (t)$
\begin{equation}\label{eq:stdKF}
\hat{x}_i(t)=\overline{x}_i(t)+K_i (t)(x_i(t)-H_i(t)\overline{x}_i(t)),
\end{equation}
where $K_i (t)$ is called Kalman gain and adjusted to minimize the posterior error covariance at each timestamp. Particularly, the prior estimate $\overline{x}_i (t)$ can be predicted according to the process model (Eq.~(\ref{eq:state})) and the measurement model (Eq.~(\ref{eq:sensing})). 

The standard Kalman filter is only applicable to produce the estimation of true state $r(t)$ for each node individually. Nonetheless, all $m$ nodes measure the same dynamic process described in Eq.~(\ref{eq:state}) and their estimation can be better calibrated once their measurements are shared among the network.

The Kalman-consensus information filter (KCIF)~\cite{olfati2009kalman} is a decentralized form of Kalman filter to collaboratively estimate the targeted process $r(t)$ with better consensus. In particular, besides the standard Kalman estimator operations, each node will exchange messages among its neighboring nodes and enforce a consensus term on locally prior estimates to reach a consensus among all nodes. The Kalman-consensus information filter can be written as
\begin{equation}\label{eq:KCIF}
\hat{x}_i (t)=\overline{x}_i (t)+M_i (t)(y_i(t)-Y_i \overline{x}_i (t))+C_i(t) \sum_{j \in N_i} (\overline{x}_j (t)-\overline{x}_i (t)),
\end{equation}
where $y_i(t)$ and $Y_i$ are weighted measurement and information matrix of neighbouring nodes of $i$, respectively, $N_i$ refers to the set of one-hop neighbors of node $i$, $M_i(t)$ is the posterior estimation covariance, and $C_i(t)$ is the consensus gain, which keeps the balance between the consensus and the stability of distributed Kalman estimators.

\begin{table}
\centering \caption{Notations}\label{Notation}\vspace{-2mm}
\fontsize{8pt}{\baselineskip}\selectfont
\begin{tabular}{lp{0.35\textwidth}} \hline
$m, d, T$& Number of distributed servers, regions, timestamps\\
$i, t, j, k$& Index of servers, time, neighbouring servers, regions \\
SP$_i$& $i^{th}$ distributed server\\
$u_i^j$& $j^{th}$ registered user at SP$_i$\\
$\mathcal{E}(t)$& Dynamic adjacent matrix at time $t$ \\
$N_i$& Neighboring set of SP$_i$\\
$\overline{N}$& Average node degree of distributed servers \\
$r(t)$& Overall real-time statistics at time $t$\\
$Q$& Variance (covariance) of $r(t)$ \\
$H_i$& Observation coefficient of SP$_i$\\
$f(\cdot)$& Aggregate statistics function\\
$x_i(t), z_i(t)$& aggregate, and perturbed statistics of SP$_i$ at time $t$\\
$R_i, \hat{R}_i$& Variance (covariance) of observation statistics of SP$_i$\\
$\overline{x}_i (t), \widehat{x}_i (t)$& Prior/Posterior estimated statistics of SP$_i$ at time $t$\\
$M_i(t)$& Variance (covariance) of posterior estimation at $t$\\
$K_i(t), C_i(t)$& Kalman/Consens gain of SP$_i$ at time $t$\\
$P_i$& Variance (covariance) of prior estimation at SP$_i$\\
$u_i (t)$& Weighted measurement of SP$_i$\\
$y_i (t)$& Average measurement of SP$_i$'s neighbours \\
$U_i$& Information matrix of SP$_i$\\
$Y_i$& Fused information matrix of SP$_i$'s neighbours \\
\hline
\end{tabular}
\vspace{-3mm}
\end{table}

\section{Our Approaches}
\label{sec: approach}


In this section, we first give a non-private solution for real-time decentralized statistical estimation for crowd-sourced data. Then, we detail our baseline solution with DP and communication efficiency, which is called \textsf{DPCrowd} for one-dimensional and finite streams. Finally, we present the enhanced solution, called \textsf{DPCrowd+} for multi-dimensional and infinite streams. The main notations are listed in Table~\ref{Notation}.

\subsection{Non-private Solution}\label{sec: Non-private Solution}

One natural solution is that each distributed server independently estimates true statistics from its own database $D_{i,t}$.

\subsubsection{Basic Idea}\label{sec: alternatives}

Without loss of generality, we simply denote the one-dimensional true statistics $r^k(t)$ at the $k^{th}, ~(\text{where}~k=1, \ldots, d)$ dimension as $r(t)$. Then, we model $r(t)=f(D_t)$ as a dynamic process defined as Eq.~(\ref{eq:state}),
where $A(t)$ is the transition coefficient and can be simplified as a constant $A(t)=A=1$ when the timestamp is short. $\omega(t)$ is the process noise and follows a normal distribution, i.e., $\omega(t) \sim N(0, Q)$. Here, $Q$ can be learned from history data.

Since the user group $G_i(t)$ of each distributed server at slot $t$ can be regarded as a uniform sample of the whole population, it can be naturally assumed that the aggregate statistics $x_i (t)=f(D_{i,t})$ at the server SP$_i$ is an observation of the true time-series statistics $r(t)=f(D_t)$ and follows the linear equation as Eq.~(\ref{eq:sensing}).
The linear observation coefficients $H_i(t)$ corresponds to the ratio of the registered users $G_i (t)$, which represents the estimation weight of each distributed server in the crowd-sensing scenarios. 
\begin{equation}\label{eq:observ coefficient}
H_i (t)=\frac{|G_i (t)|}{n},
\end{equation}
and $v_i (t)$ is the observation (measurement) noise and follows a normal distribution, i.e., $v_i (t) \sim N(0, R_i (t))$. Since the uniform sample, there is generally, $R_i (t)=(H_i (t))^2Q$. 

To have an estimation for the real-time true statistics, the standard Kalman filter~\cite{fan2014adaptive} or other temporal correlation exploitation techniques~\cite{rescuedp2016,kellaris2014differentially,chen2017pegasus} can be adopted by each distributed server individually to exploit the temporal correlations in the aggregation data of crowd-sourced users, which can be formulated as Eq.~(\ref{eq:stdKF}). Nonetheless, due to independent estimation with partial knowledge, the estimations can be rather rough and no consensus can be achieved among distributed servers without mutual trust. An alternative solution is to multi-hop broadcast (i.e., blind flooding) each server's independent estimation to all others and then conduct weighted average estimation at all distributed servers. Nonetheless, the multi-hop broadcast would cause not only all-to-all communication complexity of $O(m^2)$, but also large estimation delay of at most $O(m)$ relays.

\subsubsection{KCIF Based Statistical Estimation}

We now propose a communication-efficient solution by utilizing Kalman-consensus information filter~\cite{olfati2009kalman} to collaboratively estimate the true statistics for distributed servers via only single-hop message exchange. The key idea is that each distributed server corrects its prior estimation with not only the standard Kalman process, but also the consensus information from its one-hop neighboring servers.

Algorithm~\ref{alg:nonprivate} presents the KCIF based statistical estimation. At each timestamp, given initialized estimation covariance $P_0$ and prior estimation $\bar{x}_i (t)$, each server begins by obtaining its aggregate statistics $z_i (t)=f(D_{i,t})$ from its partial crowd-sourced data $D_{i,t}$. Then, it computes and broadcasts the prior estimation, the weighted information vector $u_i (t)$ and matrix $U_i (t)$ to its one hop neighbors. Meanwhile, it receives the similar information from its direct neighbors $N_i$ and fuses the information as $y_i (t) =\sum_{j\in J_i} u_j (t) $, ~$Y_i (t) =\sum_{j\in J_i} U_j (t)~(\text{where}~ J_j=N_j \bigcup \{i\})$. After that, it computes the posterior estimation error covariance $M_i(t)$ and consensus gain $C_i(t)$ to derive the posterior estimation $\widehat{x}_i(t)$. Finally, it updates both the prior estimation and prior estimation covariance for the next iteration. With only one-hop communications, all servers can collaboratively estimate the true statistics from their own partial database $D_i$ in a non-private way.

\begin{algorithm}[htbp]\footnotesize
	\caption{Non-private Decentralized Statistical Estimation}
	\label{alg:nonprivate}
	\begin{algorithmic}[1]
		\REQUIRE Raw crowd-sourced data $D_{i,t}$, population ratio $H_i$, initial value $P_i(0)=P_0$, $\bar{x}_i(0)=f(D_0^i)$, messages $\text{msg}_j(t)=\{u_j(t), U_j (t), \overline{x}_j(t)\}$, neighbor set $N_j$,  $J_j=N_j \bigcup \{i\}$, stepsize parameter $\beta$.
		\ENSURE Posterior estimate aggregation $\hat{x}_i(t)$;
		
		\STATE Obtain aggregate statistics $z_i(t)=f(D_{i,t})$ with covariance $R_i$;
		\STATE Compute $u_i(t)=H_i (t) z_i(t)/R_i (t) $, $U_i=(H_i (t))^2/ R_i (t)$;
		\STATE Broadcast the message $\text{msg}_i(t)=(u_i(t), U_i, \bar{x}_i(t))$;
		\STATE Receive messages $\text{msg}_j (t)$ from all neighbors $j \in N_i$;
		\STATE Aggregate information $y_i(t) =\sum_{j\in J_i} u_j(t) $, ~$Y_i =\sum_{j\in J_i} U_j (t)$;
		\STATE Compute posterior estimation covariance $M_i(t)=1/(1/P_i(t)+Y_i (t))$;
		\STATE Compute consensus gain $C_i(t)$ as \\ $\gamma=\beta/(|P_i(t)|+1)$,  $C_i(t)=\gamma P_i(t)$;	
		\STATE Calculate posterior estimation as \\
		$\widehat{x}_i(t)=\overline{x}_i(t)+ M_i (y_i(t)-Y_i (t)\bar{x}_i(t))+C_i(t) \sum_{j \in N_i}(\bar{x}_j(t) -\bar{x}_i(t))$;
		\STATE Update the prior estimation $\bar{x}_i(t)=A \cdot \hat{x}_i(t) $;
		\STATE Update the prior estimation covariance $P_i(t)=A^2 M_i(t) +Q$;
	\end{algorithmic}
\end{algorithm}

\subsubsection{Challenges for DP and Communication-efficiency}\label{subsubsec: challenges}

The messages exchanged among servers in Algorithm~\ref{alg:nonprivate} contains the information derived from the raw aggregate statistics $z_i(t)$, which may lead to the privacy exposure of individual users in $G_i$. Especially, servers may be geographically far away and not privacy reliable to each other. Besides, despite only single-hop communication, the continuous message exchange in Algorithm~\ref{alg:nonprivate} at each timestamp would still incur great communication cost on a long timescale. To address both the concerns of privacy and communication cost, we aim to propose a real-time decentralized statistical parameter estimation framework with both DP and communication efficiency for multiple distributed servers on crowd-sourced data. A naive solution for DP suggests adding Laplace noise to the raw aggregate statistics $z_i$. However, we are still facing the following challenges:

\begin{itemize}

	\item \emph{How to make the decentralized statistical estimation work in both communication-efficient and DP way?}

	\item \emph{How to improve the estimation utility with given DP requirements (i.e., user-level $\varepsilon$-DP for finite streams or $w$-event level DP for infinite streams) considering dynamic aggregation consumes DP budget quickly?}

	\item \emph{How to reduce the estimation error for sparse regions in multi-dimensional data considering the DP noise may overwhelm the statistics over sparse regions?}

\end{itemize}

\begin{figure}[htbp]\vspace{-0.2cm}
	\centering\epsfig{file=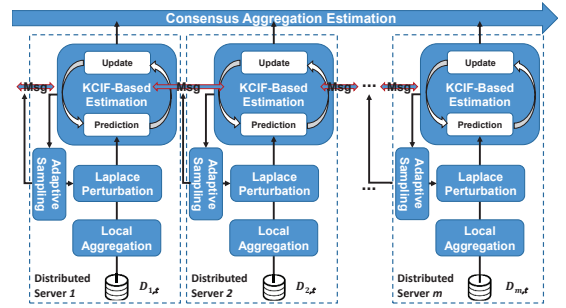, width=0.4\textwidth}
	\caption{A High-level Overview of \textsf{DPCrowd} \label{fig:ProperDP}}
\vspace{-0.5cm}
\end{figure}

\subsection{\textsf{DPCrowd}: Real-time Decentralized Statistical Estimation for One-dimensional and Finite Data Streams}\label{sec: ProperDP}
To mitigate the privacy and communication challenges in the non-private solution of Section~\ref{sec: Non-private Solution}, we first propose a baseline scheme \textsf{DPCrowd} with user-level $\varepsilon$-DP and communication-efficiency for real-time decentralized statistical estimation on one-dimensional and finite data streams.
\subsubsection{Overview of \textsf{DPCrowd}} \label{sec: Basic Idea}

Fig.~\ref{fig:ProperDP} presents a high-level overview of \textsf{DPCrowd} on distributed servers. \textsf{DPCrowd} mainly consists of three mechanisms: Laplace Perturbation, KCIF-based Estimation, and Adaptive Sampling based Intermittent Communication.

\begin{itemize}

  \item \textbf{Laplace Perturbation.} After confirming $D_{i,t}$, each server computes the raw aggregate statistics $x_i (t)=f(D_{i,t})$ at time $t$. Here, we focus on estimating the one-dimensional statistics, e.g., the population of users appeared in a particular region. To guarantee user-level $\varepsilon$-DP for the finite stream, each server perturbs its raw statistics $x_i$ by Laplace mechanism with certain portion of allocated privacy budget, performs the post-process on perturbed statistics, and forwards the results to its neighbors.

  \item \textbf{KCIF-based Estimation.} KCIF-based estimation mechanism over each distributed server fuses the information exchanged from other servers and correct its own prior prediction according to both Kalman gain and consensus gain. Kalman gain can reduce both the observation noise and perturbation noise by exploiting the temporal correlations in real-time statistics. Consensus gain can integrate partial statistics from distributed servers to further correct overall estimation with consensus.

  \item \textbf{Adaptive Sampling based Intermittent Communication.} To reduce the communication cost and better utilize the privacy budget for a finite stream, we propose an adaptive sampling based intermittent communication strategy via leveraging the temporal correlations in crowd-sourced data. In particular, based on the dynamic changes between the prior estimation and posterior estimation after KCIF-based estimation, the server adaptively decides whether to perturb the aggregate statistics with certain privacy budget or approximate it with the previous estimation. Thus, the limited privacy budget can be allocated more to the necessary timestamps. Once the approximation strategy is chosen at the current timestamp, the server does not need to broadcast its estimations, thus further reducing the communication overheads. 

\end{itemize}

Based on the above design rationales, Algorithm~\ref{alg:ProperDP} presents the main procedures of \textsf{DPCrowd} at a distributed server SP$_i$. In the following, we describe the main components with detailed procedures.

\begin{algorithm}[htb]\footnotesize
\caption{\textsf{DPCrowd}}
\label{alg:ProperDP}
\begin{algorithmic}[1]
\REQUIRE
\begin{tabular}[t]{p{0mm}l}
&$D_{i,t}$: Partial dataset crowd-sourced at SP$_i$ at timestamp $t$,\\
&$\varepsilon$: privacy budget,\\
&$T_s$: maximum number of sampling timestamps.
\end{tabular}
\ENSURE $r_i (t)$: Released statistics of SP$_i$ at timestamp $t$;

\FOR {each timestamp $t=1, \ldots, T$}
\STATE Obtain raw aggregate statistics $x_i (t)=f(D_{i,t})$;
\IF {$t$ is a sampling point \&\& $\text{numSamples}_i<T_s$}
\STATE $z_i (t) \leftarrow$ perturb $x_i (t)$ by \textbf{Laplace Perturbation};
\STATE $\text{numSamples}_i++$; ~//\textbf{\textit{Number of sampling timestamps}}
\STATE Estimate prior $\overline{x}_i (t)$ and message $\text{msg}_i (t)$ from \textbf{KCIF-Prediction};
\STATE Broadcast the message $\text{msg}_i (t)$;
\STATE Receive messages $\text{msg}_j (t)$ from one-hop neighbors $j \in N_i$;
\STATE Estimate posterior $\hat{x}_i (t)$ from \textbf{KCIF-Update};
\STATE Adjust sampling rate by \textbf{Adaptive Sampling};
\STATE Release posterior estimation as $r_i (t)$ $\leftarrow$ $\hat{x}_i (t)$;
\ELSE
\STATE $z_i (t)$ $\leftarrow$ $r_i (t-1)$;
\STATE Estimate prior $\overline{x}_i (t)$; ~//\textbf{\textit{No message broadcast}} 
\STATE Receive messages $\text{msg}_j (t)$ from one-hop neighbors $j \in N_i$;
\STATE Estimate posterior $\hat{x}_i (t)$ from \textbf{KCIF-Update};
\STATE Release posterior estimation as $r_i (t)$ $\leftarrow$ $\hat{x}_i (t)$;
\ENDIF
\ENDFOR
\end{algorithmic}
\end{algorithm}

\subsubsection{Laplace Perturbation}\label{sec: Laplace}

To realize user-level $\varepsilon$-DP at each server, the basic idea is to apply Laplace mechanism with different budget $\varepsilon(t)$ to inject Laplace noise to aggregate statistics at each time $t$, while keeping the total privacy budget consumption $\sum\varepsilon(t)$ for the finite stream no more than $\varepsilon$.

(1) Local Data Aggregation: 
At each timestamp $t$, the server SP$_i$ obtains its aggregate statistic from its local crowd-sourced data (i.e., $x_i (t)=f(D_{i,t})$) and calculates its current observation coefficient $H_i (t)=\frac{|G_i (t)|}{n}$.

(2) Aggregate Data Perturbation: 
We adopt the Laplace mechanism to perturb the aggregate statistic $x_i (t)$ with a noise $\upsilon_i (t)$ drawn from the Laplace distribution $\mathcal{L}(\Delta_f/\varepsilon(t))$, where $\Delta f$ is the sensitivity of the aggregate function $f(\cdot)$ and $\varepsilon(t)$ is the DP budget allocated at current timestamp. Then, we can obtain a noisy statistical value $z_i (t)=x_i (t)+\upsilon_i (t)$, which satisfies $\varepsilon(t)$-DP.

Particularly, taking the population sum of a region as the statistic function, since each crowd-sourcing user is associated with one distributed server at the same time and whether an individual user appears at a certain region can change $x_i (t)$ by at most $1$, the sensitivity of the statistic function is then $\Delta_f=1$. The Laplace noise $\upsilon_i (t)$ is drawn from Laplace distribution $\mathcal{L}(1/\varepsilon(t))$, where $\varepsilon(t)$ is the DP budget at time $t$. For a time-series $x_i (t)$ with the time length of $T$, according to the sequential composition theorem, the privacy budget can be simply allocated as $\varepsilon(t)=\varepsilon/T$ at each time $t$ to meet the requirement of user-level $\varepsilon$-DP. However, much smaller $\varepsilon(t)$ would lead to larger amplitude of noise and worsened utility. Therefore, instead of uniform allocation of privacy budget $\varepsilon(t)=\varepsilon/T$ in the finite stream, we adaptively perturb the statistics at different timestamps by allocating different privacy budget according to the dynamic changes of $x_i (t)$ as described in Section \ref{sec: Basic Idea}. The detailed adaptive privacy budget allocation scheme will be introduced later in Section~\ref{sec: Adaptive Sampling}.

\textbf{Remark 1.} We make an assumption that each crowd-sourcing user can be associated with only one distributed server at the same time in the system model of Section III. Without loss of generality, this assumption can be relaxed as each user can be associated with at most $c$ distributed servers at the same time. In such a case, the sensitivity can be set as $\Delta_f=c$ to increase the amount of perturbation noise in our algorithms to provide sufficient privacy preservation for any crowd-sourcing user.

\textbf{Remark 2.} Our work emphasizes designing a privacy-preserving mechanism with a given parameter $\varepsilon$ while improving data utility. However, some work~\cite{cao2018quantifying,song2017pufferfish} argued that DP on correlated data could offer $\varepsilon'$-DP (where $\varepsilon<\varepsilon' \ll T\varepsilon$ for general correlations), a weaker privacy guarantee when adopting the personal data principle~\cite{li2016differential,cao2018quantifying} in the different understanding of privacy~\cite{mcsherrypost1}. Then, the privacy parameter $\varepsilon$ can be simply scaled down (by no more than $T$ times) according to the temporal correlation degree in the statistical results to satisfy stronger privacy protection under the personal data principle.

Combining with the original observation process in Eq.~(\ref{eq:sensing}), 
the noisy local statistics $z_i (t)$ can be further represented as
\begin{align}\label{eq:private_observation}
z_i (t)=H_i (t)  r(t)+v_i (t)+\upsilon_i (t)=H_i (t)  r(t)+o_i (t).
\end{align}
Here, $o_i (t)$ denotes the overall observation noise at $t$, which equals to the sum of two independent noise: the original observation noise $v_i (t)$ with the variance $\mathrm{Var}(v_i (t))=2(\frac{\Delta_f}{\varepsilon(t)})^2$, and the privacy-preserving noise $\upsilon_i (t)$ with the variance $\mathrm{Var}(\upsilon_i (t))=(H_i (t))^2 Q$. 
Then, $z_i (t)$ can be further post-processed and shared with other servers to jointly estimate the true statistics $r(t)$ later.

\subsubsection{KCIF-Based Estimation}\label{sec:KCF}

To collaboratively estimate the true statistics with consensus, each server not only needs to conduct prior estimation according to its own knowledge, but also corrects the prior estimation via messages exchange. Based on the non-private solution in Section~\ref{sec: Non-private Solution}, we adopt a KCIF-based estimation mechanism to fuse the information from distributed servers, thus achieving both high utility and consensus.

(1) Noise Model of Kalman Filter: Generally speaking, Kalman filter achieves the optimal posterior estimation when the measurement noise follows the Gaussian distribution. Fortunately, as proved in~\cite{fan2014adaptive}, Kalman filter works effectively on the noise with Laplace distribution $\mathcal{L}(\Delta_f/\varepsilon(t))$ when the variance parameter $R$ in Kalman filter satisfies $R \propto 2(\frac{\Delta_f}{\varepsilon(t)})^2$. That is to say, we can use a Gaussian distribution $\mathcal{N}(0, R)$ to approximate the Laplace distribution $\mathcal{L}(\Delta_f/\varepsilon(t))$ for privacy preservation. Thus, according to Eq.~(\ref{eq:private_observation}), to achieve minimum variance posterior estimate under both the observation noise and the privacy-preserving noise, the optimal value $\hat{R}_i (t)$ for Kalman filter can be set as
\begin{align}\label{eq:approx}
\hat{R}_i (t) \propto \alpha\cdot{2(\frac{\Delta_f}{\varepsilon(t)})^2+(H_i (t))^2 \cdot Q},
\end{align}
where $\alpha$ is an adjustable proportional coefficient. This approximation has also been verified in our experiments in Section~\ref{sec:evaluation}. 

(2) KCIF-Prediction: 
KCIF-Prediction maintains a prior estimation $\overline{x}_i (t)$ for each server SP$_i$. It can be initialized as
\begin{align}
\overline{x}_i (0)={x_i (0)}/{H_i (0) }.
\end{align}
After that, according to Eq.~(\ref{eq:state}), the prior estimation can be predicted as its previous estimation.
\begin{align}
\overline{x}_i (t)=A \hat{x}_i (t-1).
\end{align}
In addition, according to the standard Kalman filter, the prior estimation error covariance $P_i (t)$ of SP$_i$ can be predicted as
\begin{align}
P_i (t)=A^2 M_i(t-1) +Q,
\end{align}
where $M_i (t-1)$ is the posterior error covariance at time $t-1$ and $Q$ is the variance of process noise in Eq.~(\ref{eq:state}). 
The posterior error covariance can be initialized as $M_i (0)=(H_i (t))^2/\hat{R}_i (t)$, where $\hat{R}_i (t)$ is set according to Eq.~(\ref{eq:approx}). 


(3) Message Exchange: 
After perturbation and prediction, each server exchanges messages with their neighbors in one hop for collaborative estimation. The message $\text{msg}_i (t)$ encapsulated from SP$_i$ consists of three parts: the prior estimation $\overline{x}_i (t)$, the weighted statistics $u_i (t)$, and the information matrix $U_i (t)$. In particular, $u_i(t)$ and $U_i (t)$ can be computed as
\begin{align}\label{eq:weight}
u_i (t)&={(H_i (t) \cdot z_i (t))}/{\hat{R}_i (t)},\mbox{ and}\\
U_i (t)  &=(H_i (t))^2 / \hat{R}_i (t).
\end{align}
After that, $\text{msg}_i (t)=(\bar{x}_i (t),u_i (t), U_i (t))$ is broadcasted to the directed neighbors. $\text{msg}_i (t)$ only contains sanitized information of SP$_i$'s aggregate statistics over $D_{i,t}$, which do not leak the privacy.

(4) KCIF-Update: 
Receiving the messages from direct neighbors $j\in N_i$, SP$_i$ first sums up the weighted aggregation $u_j (t)$ and the weighted information matrices $U_j (t)$ as follows.
\begin{align}
y_i (t)=\sum_{j\in N_i \bigcup \{i\}} u_j (t),\\
Y_i (t)=\sum_{j\in N_i \bigcup \{i\}} U_j (t).
\end{align}
Then, combining with the prior estimation error covariance $P_i$, SP$_i$ will compute both the posterior estimation error covariance $M_i(t)$ and consensus gain $C_i(t)=\gamma_i (t) P_i(t)$ in Kalman-consensus information filter, respectively.
\begin{align}\label{eq:kgain}
M_i (t)&=1/(1/{P_i(t)}+Y_i(t)),\\
C_i(t)&=\gamma_i (t) P_i(t)=\beta P_i (t) /(|{P_i(t)}|+1),
\end{align}
where $\beta>0$ is a relative small constant with the order of the time step size in discretization of the continuous time process. 

Finally, according to the Kalman-consensus information filter, the posterior estimation $\hat{x}_i (t)$ at SP$_i$ can be computed as
\begin{align}\label{eq:posterior}
&\hat{x}_i (t)=\bar{x}_i (t)+ \\ \nonumber
 & M_i (t) (y_i(t)-Y_i(t)\bar{x}_i(t))+C_i(t) \sum_{j \in N_i}(\bar{x}_j (t)-\bar{x}_i (t)).
\end{align}
With the correction of standard Kalman estimation term controlled by the posterior estimation error covariance $M_i (t)$ and the consensus term controlled by consensus gain $C_i(t)=\gamma_i P_i(t)$ in Eq.~(\ref{eq:posterior}), the posterior estimations of true statistics at each distributed server will gradually reach both accuracy and consensus.

\subsubsection{Adaptive Sampling based Intermittent Communication}\label{sec: Adaptive Sampling}

According to the design rationale in Section~\ref{sec: Basic Idea}, a sampling based intermittent communication strategy can provide the following benefits for \textsf{DPCrowd}:

\begin{itemize}

    \item \textbf{Communication Efficiency.} Considering the signal sparsity in streams, despite only one-hop communication between neighboring servers, the continuous message exchanges at all timestamps of KCIF-based estimation in Section~\ref{sec:KCF} seems to be communication expensive, in terms of the length $T$ of the finite data stream. One common method of communication reduction is to reduce the communication frequency via sampling.
	
    \item \textbf{Privacy Budget Allocation.} As described in Section~\ref{sec: Laplace},to achieve user-level $\varepsilon$-DP for a finite stream with the time length of $T$ timestamps, one simple idea is to uniformly allocate the total privacy budget $\varepsilon$ for all $T$ timestamps. Then, if $T$ is large, the average privacy budget $\varepsilon(t)$ used for each timestamp will be small and lead to the low utility of estimation at each server. To enhance the estimation accuracy, one key idea is to reduce the noise addition by selectively allocating more privacy budget at some sampling timestamps and approximating aggregation results at non-sampling timestamps with previous estimations without privacy budget consumption.

\end{itemize}

Combining the above ideas, we propose to apply the sampling based intermittent communication strategy to both reduce the communication frequency in KCIF-based estimation (Section~\ref{sec:KCF}) and save up the privacy budget in Laplace perturbation (Section~\ref{sec: Laplace}), without significantly affecting the estimation utility. The basic idea is that, only at the sampling timestamps, each distributed server allocates privacy budget and sends DP protected messages to its one-hop neighbors (Lines 3$\sim$11 in Algorithm~\ref{alg:ProperDP}); while at the non-sampling timestamps, each distributed server approximates the prior estimation with previous posterior estimation without privacy budget and does not send out messages (Lines 12$\sim$17 in Algorithm~\ref{alg:ProperDP}).

One straightforward solution is the fixed-rate sampling strategy. Given a predefined sampling interval $I$, each server SP$_i$ will periodically sample and perturb its aggregate statistics $x_i(t)$ by Laplace mechanism. The total number of sampling and communication timestamps for each server is $T_s=T/I$, and the privacy budget for each sampling timestamp $t$ is $\varepsilon(t)=\varepsilon/T_s$. The choice of sampling rate (or sampling interval $I$) has the following impacts:

\begin{itemize}

    \item When $I$ is small, the communication frequency is high with less consensus error, but $\varepsilon(t)$ is small and too many sampling timestamps will lead to much perturbation error.

    \item When $I$ is large, communication frequency and perturbation error can be reduced, but a large sampling gap will cause larger approximation and concensus error.

\end{itemize}

Thus, both sampling and non-sampling timestamps will cause errors and may have an impact on the overall accuracy. To achieve higher accuracy, it requires to seek the optimal sampling rate according to some prior knowledge about the data, which is, however, not applicable in real-time crowd-sourced data. A good sampling strategy should adjust the sampling rate to minimize the two errors with given privacy budget $\varepsilon$. 
We apply the adaptive-rate sampling strategy based on PID control in FAST~\cite{fan2014adaptive} to adjust the sampling rate based on the dynamics of the statistics. In particular, each server SP$_i$ dynamically adjusts it own sampling intervals $I_i$ according to the real-time error between the prior and posterior estimations. The details can be referred to~\cite{fan2014adaptive}.



\subsection{\textsf{DPCrowd}+: Real-time Decentralized Statistical Estimation for Multi-dimensional and Infinite Data Streams} \label{sec:extension}

So far, \textsf{DPCrowd} focuses on the crowd-sourced data with one-dimension and limited length, e.g., distribute servers only care about the true statistic of a particular region in a particular time period. Nonetheless, in reality, typical crowd-sourced data are multi-dimensional (even high-dimensional) and infinitely generated. For example, the servers need to estimate the true statistics over all regions uninterruptedly. In such a case, there are two challenges for \textsf{DPCrowd}:
\begin{itemize}
\item \emph{Without consideration of the sparsity, the same amount of noise would be added to all regions and ruin the utility of those regions with a small value.}
\item \emph{Simple event-level DP or user-level DP on finite streams will be not applicable to infinite data streams as the total privacy budget accumulates with the time.}
\end{itemize}

To address the aforementioned challenges, we further propose \textsf{DPCrowd}+, a more applicable privacy-preserving decentralized statistical estimation mechanism for multi-dimensional infinite crowd-sourced data streams.

\subsubsection{Data Modeling} \label{sub:data modeling}

Before introducing the details of \textsf{DPCrowd}+, we first extend the data model in Section~\ref{sec: ProperDP} to a multi-dimensional scenario. Similar to \textsf{DPCrowd}, multi-dimensional true statistics $\mathbf{r} (t)$ can be modeled and formulated by vectors as follows
\begin{align}\label{eq:multistate}
\mathbf{r} (t+1)=\mathbf{A(t)} \cdot \mathbf{r} (t)+\bm{\omega} (t),
\end{align}
where $\mathbf{r} (t)$ are $d$-dimensional vector and each element $r^k(t)$ represents the true statistics of region $k$ at timestamp $t$. $\mathbf{A(t)}=[a_{i,j}(t)]_{d\times d}$ can be a $d \times d$ time-varying transition matrix, which models the correlations among dimensions (e.g., regions). Particular, matrix ${A(t)}$ at time $t$ may be a general linear transformation matrix or a Markov matrix (stochastic matrix). For a Markov matrix, each element $a_{i,j}(t)$ may represent the probability that a user may transit from region $i$ to $j$ in a city or from website $i$ to $j$ during Internet surfing at different time $t$~\cite{fan2014monitoring}. For simplicity, we consider $\mathbf{A(t)}=\mathbf{A}$ is a constant linear transition matrix, which can be trained from the history data. 

Also, $\bm{\omega} (t)=(\omega^1 (t), \omega^2 (t),\ldots,\omega^d (t))$ is the $d$-dimensional process noise that follows the $d$-dimensional Gaussian distribution, i.e., $\bm{\omega} (t)\sim \bm{N}(0,\mathbf{Q})$, where $\mathbf{Q}=[Q_{i,j}]_{d \times d}$ is the covariance matrix and each element $Q_{i,j}$ is a scalar value and represented as the covariance of $\omega_d (t)$. Although the constant matrix $\mathbf{A(t)}$ represents the general steady correlations among dimensions, the time-varying process noise $\bm{\omega} (t)$ can reflect the dynamic changes of dimensional correlations and sparsity. For example, the unusual social events may lead to the changes of traffic patterns or webpage views at a certain period. For simplicity, we assume that the process noise of each dimension is independent of each other, i.e., $Q_{i,j}=0,~i \neq j$, then its covariance matrix $\bm{Q}$ can be simplified as
\begin{align}
\mathbf{Q}=\mathbf{diag}(Q_{1,1}, Q_{2,2},\ldots,Q_{d,d}).
\end{align}

Meanwhile, we assume the aggregate $d$-dimensional statistical vector $\mathbf{x}_i (t)$ at each distributed server also follows a linear equation as
\begin{align}\label{eq:aggregation}
\mathbf{x}_i (t)=H_i (t) \cdot \mathbf{r} (t)+\mathbf{v}_i (t),
\end{align}
where $H_i (t)$ at slot $t$ is a scalar value represents the linear observation coefficients as Eq.~(\ref{eq:observ coefficient}) and $\mathbf{v}_i (t)=(v_i^1 (t), v_i^2 (t),\ldots,v_i^d(t))$ is a $d$-dimensional observation noise. We assume that each element of $\mathbf{v}_i (t)$ is independent from each other and follows the zero-mean Gaussian distribution with the variance $\mathbf{R}_i (t)=(H_i (t))^2 \mathbf{Q}$. Then, there is $v_i^k (t)\sim N(0, R_i^k (t))$, where $R_i^k (t)=(H_i (t))^2 Q_{k,k} ~(k=1,2,\ldots,d)$.

\subsubsection{\textsf{DPCrowd}+ based on Dynamic Grouping}\label{sub:promiseDP frame}

The workflow of \textsf{DPCrowd}+ on each distributed server is shown in Fig.~\ref{fig:PromiseDP}. Compared with \textsf{DPCrowd} shown in Fig.~\ref{fig:ProperDP}, \textsf{DPCrowd}+ includes two more components: (i) dynamic grouping and (ii) adaptive budget allocation, inspired by~\cite{rescuedp2016}. In the dynamic grouping mechanism, similar regions with small values will be grouped together to avoid the overdose of noise. In particular, the correlations of different regions are calculated based on the previously published results to guarantee privacy. Thus, high-dimensional aggregate statistics may be grouped into several groups and different Laplace noise is then added to each group to strike a good balance between privacy and utility. The adaptive budget allocation mechanism is responsible for allocating the privacy budget to make sure $w$-event DP is satisfied in the infinite aggregation stream. Thus, besides adopting adaptive sampling to reduce the noise, the privacy budget for each timestamp should be carefully allocated to meet the requirement. The privacy budget will be allocated according to the dynamics of grouped regions to improve the utility. Since the dynamic grouping and adaptive allocation mechanism are exactly the same as those in RescueDP algorithm. Please refer to~\cite{rescuedp2016} for more detail.

	

\begin{figure}[htbp]\vspace{-0.3cm}
	\centering\epsfig{file=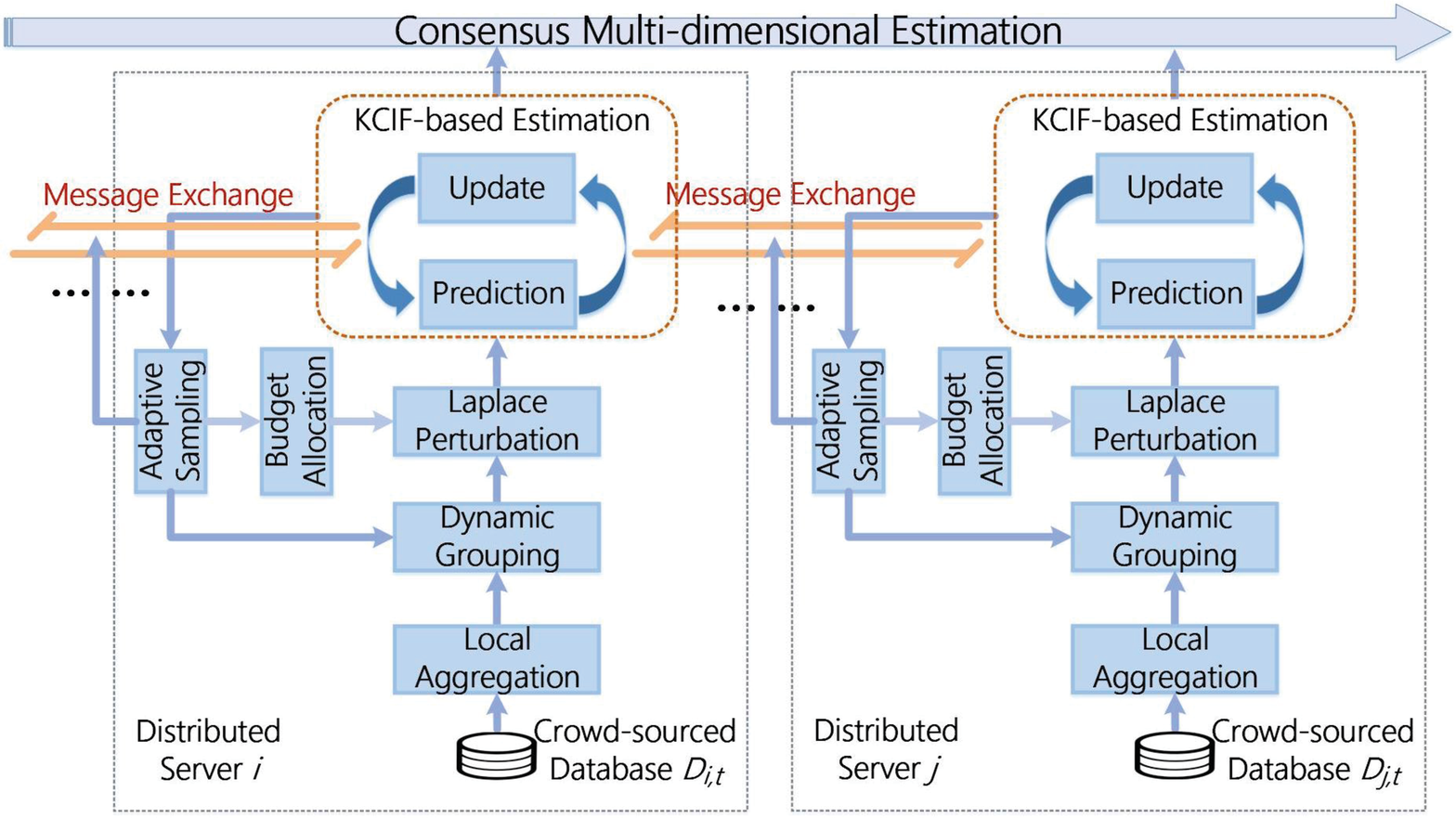, width=0.4\textwidth}
	\caption{A Framework of \textsf{DPCrowd}+ \label{fig:PromiseDP}}
\vspace{-0.3cm}
\end{figure}%

Algorithm~\ref{alg:PromiseDP} presents the main procedures of \textsf{DPCrowd}+. It should be noted, other components, i.e., Laplace perturbation, KCIF-based estimation, adaptive sampling will also be adjusted to incorporate dynamic grouping and adaptive budget allocation.

\begin{algorithm}[htbp]\footnotesize
	\caption{\textsf{DPCrowd}+}
	\label{alg:PromiseDP}
	\begin{algorithmic}[1]
		\REQUIRE
		\begin{tabular}[t]{p{0mm}l}
			&$D_{i,t}$: partial databases for $d$ regions of SP$_i$ at timestamp $t$,\\
			&$\varepsilon$: privacy budget,\\
		\end{tabular}
		\ENSURE $\mathbf{r}_i$: Released $d$-dimensional statistics of SP$_i$ at timestamp $t$.
		
		\FOR {each timestamp $t=1,\ldots,T$}
		\IF {$t$ is a sampling point}
		\STATE Add sampling regions into set $W_i (t)$;
		\STATE Group regions in $W_i (t)$ by \textbf{Dynamic Grouping};
		\FOR {each region $k$ do}
		\IF {$k \in W_i (t)$}
		\STATE Obtain privacy budget from \textbf{Adaptive Budget Allocation};
		\STATE $z_i (t) \leftarrow$ perturb $x_i^k (t)$ by \textbf{Laplace Perturbation};
		\ELSE
		\STATE $x_i^k (t) \leftarrow r_i^k (t-1)$;
		\ENDIF
		\ENDFOR
		\STATE Obtain $prior$ and $message_i(t)$ from \textbf{KCIF-Prediction};
		\STATE Broadcast $message_i (t)$ to neighbors in $N_i$;
		\STATE Receive messages from all neighbors in $N_i$;
		\STATE Obtain $posterior$ form \textbf{KCIF-Update};
		\STATE $\mathbf{r}_i (t) \leftarrow posterior$;
		\FOR {$k \in W_i (t)$}
		\STATE Adjust sampling rate by \textbf{Adaptive Sampling};
		\ENDFOR
		\ELSE
		\STATE Obtain $prior$ and $message_i (t)$ from \textbf{KCIF-Prediction};
		\STATE Receive messages from $N_i$; ~\textbf{\textit{//No message broadcast}}
		\STATE Estimate posterior $\hat{x}_i (t)$ from \textbf{KCIF-Update};
		\STATE Release posterior estimation as $\mathbf{r}_i (t)$ $\leftarrow$ $\mathbf{\hat{x}}_i (t)$;
		\ENDIF
		\ENDFOR
	\end{algorithmic}
\end{algorithm}

\section{Algorithm Analysis}\label{sec:analysis2}

In this section, we conduct the theoretical analysis of our scheme in terms of privacy protection and utility, as well as communication latency and cost.

\subsection{Privacy Analysis}

\begin{theorem}\label{proof:DPCrowd}
	\textsf{DPCrowd} in Algorithm~\ref{alg:ProperDP} guarantees user-level $\varepsilon$-DP for the registered crowd-sourcing users for a finite stream with the length of time $T$ at each server SP$_i$.
\end{theorem}

\begin{proof}

Given the maximum number of sampling timestamps $T_s$ and the total privacy budget $\varepsilon$, each Laplace perturbation adds noise drawn from the Laplace distribution $\mathcal{L}(\Delta f\cdot T_s/\varepsilon)$ at each sampling timestamp, which satisfies $\varepsilon/T_s$-DP for each crowd-sourcing user according to the Theorem~\ref{theorem: LM}. Then, based on Theorem~\ref{theorem: Sequential Composition}, after $T_s$ sampling timestamps, Laplace perturbations on the aggregate statistics satisfies $\varepsilon$-DP for each user for the whole finite stream.

Among all processes in \textsf{DPCrowd}, only the process of Laplace perturbation can access to the true aggregate statistic at each timestamp at each server, and other processes are all conducted on the perturbed statistics. Thus, according to the post-processing property in Theorem \ref{theorem: Post-Processing}, \textsf{DPCrowd} satisfies user-level $\varepsilon$-DP for the finite stream with length $T$ at each server.
$\hfill\blacksquare$
\end{proof}

\begin{theorem}
\textsf{DPCrowd}+ in Algorithm~\ref{alg:PromiseDP} guarantees $w$-event $\varepsilon$-DP for the registered crowd-sourcing users for an infinite stream at each distributed server SP$_i$.
\end{theorem}

\begin{proof}

Similar to Proof of Theorem~\ref{proof:DPCrowd}, we prove \textsf{DPCrowd}+ satisfies $w$-event $\varepsilon$-DP if and only if the Laplace perturbation satisfies $w$-event $\varepsilon$-DP. In particular, according to~\cite{rescuedp2016}, the perturbed statistics satisfy $w$-event $\varepsilon$-DP, our adoption of multi-dimensional Laplace mechanism with both dynamic grouping and adaptive privacy budget allocation strategies would satisfy $w$-event $\varepsilon$-DP for each region $k$ at each server SP$_i$. Therefore, \textsf{DPCrowd}+ satisfied $w$-event $\varepsilon$-DP.
$\hfill\blacksquare$

\end{proof}
	
\subsection{Utility Analysis}

Without loss of generality, we mainly focus on the general error analysis of \textsf{DPCrowd} for one-dimensional data streams (\textsf{DPCrowd}+ can have similar conclusions) without considering the adaptive sampling mechanism. The mean square posterior estimation error of SP$_i$ at timestamp $t$ can be calculated as
\begin{align}
	\mathbb{E}[|\hat{x}_i(t)-x(t)|^2]=\mathbb{E}[(\hat{x}_i-x(t))(\hat{x}_i-x(t))],
\end{align}
which is also equal to the error variance matrix $M_i(t)$ for one-dimensional data (or the trace of the error covariance matrix for multi-dimensional data) of SP$_i$ according to Kalman consensus information filter~\cite{olfati2009kalman}.
Based on Algorithms~\ref{alg:nonprivate} and \ref{alg:ProperDP}, for one-dimensional case, we have
\begin{align}\label{eq:error analysis}
	M_i(t) &=1/(1/P_i(t)+Y_i (t))\\
    &=1/(1/P_i(t)+(H_i (t))^2/\hat{R}_i(t)),\\ \nonumber
	&=\frac{\hat{R}_i(t) P_i(t)}{\hat{R}_i(t) + (H_i(t))^2 P_i(t)}\\
	&=\frac{\hat{R}_i(t) (M_i(t-1)+Q)}{\hat{R}_i(t)+(H_i(t))^2 (M_i(t-1)+Q)},
\end{align}
where $M_i (t)$ is initialized as $M_i(0)={(H_i (0))^2}/{\hat{R}_i (0)}$. From Eq.~(\ref{eq:error analysis}), we can have the following observations about the estimation error.

\begin{itemize}

    \item The posterior estimation error is decided by the observation noise variance $\hat{R}_i (t)$ (including that caused by privacy-preserving noise), the process noise variance $Q$, and the coefficient $H_i (t)$.

    \item With the increase of $\varepsilon$, observation noise variance $\hat{R}_i (t)$ decreases, and so does the error variance. This also shows the general trade-off between utility and privacy in \textsf{DPCrowd}.
	
    \item As $\hat{R}_i (t)$ is the same order of $Q$, the posterior estimation error would increase with process noise variance $Q$. This implies the stream of statistics with more fluctuates (i.e., larger $Q$) would generally cause larger estimation error.

    \item Since $M_i(0)$ is iteratively substituted in the update of $M_i (t)$, it is not difficult to see that the posterior estimation error would become small when the coefficient $H_i (t)$ is large.

\end{itemize}

\subsection{Communication Latency and Cost}

In both \textsf{DPCrowd} and \textsf{DPCrowd}+, each distributed server SP$_i$ only exchanges messages with its one-hop neighbors at each timestamp. The communication latency is only $O(1)$ hop and the communication complexity is equal to $O(\sum\limits_{i=1}^m \|N_i\|)=O(m \cdot \overline{N})$ in terms of the number of distributed servers $m$ (network scale), where $\|N_i\|$ is the degree (or the cardinality of neighboring set $N_i$) of SP$_i$, and $\overline{N}$ is the average node degree of the network. Assuming that $\overline{N}=O(\log(m))$, the communication cost of both \textsf{DPCrowd} and \textsf{DPCrowd}+ is then $O(m \log (m))$, which is scalable in terms of the number of distributed servers $m$.

According to the adaptive sampling based intermittent communication in Section~\ref{sec: Adaptive Sampling}, each server only incurs message broadcasts at its sampling timestamps. In \textsf{DPCrowd}, suppose that the total timestamp length is $T$, then for the fixed sampling strategy, given the sampling interval $I$, the communication reduction ratio is $I/T$; while for the adaptive sampling strategy, given the maximal sampling timestamps of $T_s$, the communication reduction ratio is $T_s/T$.

\section{Performance Evaluation}\label{sec:evaluation}

We conducted extensive experiments on both synthetic and real-world datasets to demonstrate both the effectiveness and efficiency of our proposed algorithms \textsf{DPCrowd} and \textsf{DPCrowd}+.

\subsection{Simulation Setup}

\textbf{Datasets:} For single-dimensional data, we used one synthetic dataset and two real-world datasets as follows:.

\begin{itemize}
  \item \textsf{Linear} is a synthetic dataset consisting of $1000$ timestamps, which are generated according to the process model in Eq.~(\ref{eq:state}) 
  with the variance ${Q}$ as $10^5$. 
  \item \textsf{Flu}\footnote{\url{http://www.cdc.gov/flu}} is part of the weekly surveillance data of flu infection provided by the Influenza Division of the Center for Disease Control and Prevention. We extracted a time-series consists of $791$ timestamps for each weekly report.
\end{itemize}

For multi-dimensional data, we also used one synthetic dataset and one real-world dataset.

\begin{itemize}

  \item \textsf{Multi-Linear} is synthesized according to the process model of Eq.~(\ref{eq:multistate}). 
   It is a six-dimensional time-series with 1000 timestamps. The size of both the transition matrix $\bm{A}$ and covariance matrix $\bm{Q}$ is $6 \times 6$.

  \item \textsf{Multi-Flu} is a multi-dimensional version of \textsf{Flu} and contains the weekly outpatient death population of $51$ states in US for $441$ weeks between 2009 and 2017. We approximated the transition matrix $\bm{A}$ by frequency statistics and trained the optimal covariance matrix $\bm{Q}$ by genetic algorithm, in which the average relative error is used as the input of fitness function.

\end{itemize}

\textbf{Simulation Methodology:} We implemented all algorithms in Matlab for simulating the interactions among $m=50$ distributed servers in a fully decentralized network. The network topology is described by an evolving stochastic matrix $\mathcal{E}(t)$, which is randomly generated with various level of network density. Based on the network model in Section~\ref{sec: System Model}, the network density $\rho$ is defined as
\begin{align}\label{eq:density}
\rho=\frac{2*num_E}{m(m-1)},
\end{align}
where $m$ is the number of distributed servers and $num_E$ is the average number of edges in $\mathcal{E}(t)$.
The communication latency between any two servers is assigned as a random number follows a uniform distribution around $100$ms~\footnote{The runtime of core algorithms are much faster, e.g., Line $2 \sim 17$ in Algorithm 2 (\textsf{DPCrowd}) and Line $2 \sim 25$ in Algorithm 3 (\textsf{DPCrowd}+) consume less than $0.1$ms when executed on a real machine (Matlab R2018a, Win10, 8GB RAM, CPU i5-5200U).}. Besides, each distributed server is assigned a random observation coefficient of $H_i (t)$ sampled according to uniform distribution $U(0, 1)$. Finally, each server fuses the received data to correct its posterior estimation. The above processes are repeated until all timestamps of each dataset are touched.

\textbf{Comparison:} To show the effectiveness of our schemes, we also summarized, simulated and compared with the coral algorithms of relevant and typical benchmark schemes on differentially private streaming: \textsf{FAST}~\cite{fan2014adaptive}, \textsf{RescueDP}~\cite{rescuedp2016}, \textsf{BD/BA}~\cite{kellaris2014differentially}, and \textsf{PeGaSus}~\cite{chen2017pegasus}. We extended them to realize real-time decentralized statistical estimation by the straightforward whole-network broadcasting and averaging, discussed in Section~\ref{sec: alternatives}. For simplicity, we denote these extension schemes as \textsf{DFAST}, \textsf{DRescueDP}, \textsf{DBD/DBA}, and \textsf{DPeGaSus}, respectively. To fairly compare the utility, we also further extended \textsf{DBD/DBA}, and \textsf{DPeGaSus} to support multi-dimensional data streams and transformed \textsf{DPeGaSus} to meet the equivalent $w$-event level privacy level. For example, we extended $w$-event level \textsf{BD/BA} and event-level \textsf{PeGaSus} to meet the same $w$-event level privacy guarantee for \textsf{DPCrowd} or $w$-event level privacy for \textsf{DPCrowd}+. We also extended our \textsf{DPCrowd} into \textsf{DPCrowd$_w$} to compare the utility improvement of \textsf{DPCrowd+} by applying the basic \textsf{DPCrowd} independently on each dimension and each $w$-timestamp-long non-overlapping window of an infinite multi-dimensional stream. The detailed features of the main comparable schemes are listed in Table~\ref{table:schemes}.

\begin{table}[htbp]\fontsize{5pt}{\baselineskip}\selectfont
	\centering
    \begin{threeparttable}[b]
        \caption{Features of Main Comparable Schemes}\label{table:schemes}
	     \begin{tabular}{|l|c|c|c|c|} \hline
    		\textbf{Schemes} &\textbf{Architecture} &\textbf{Communication} &\textbf{Dimension Correlation} & \textbf{Privacy Level}    \\
    		\hline
    		\textsf{FAST}~\cite{fan2014adaptive} & Centralized & No communication  & Single dimension & user level   \\
    		\textsf{DFAST}~\cite{fan2014adaptive}\tnote{*} & Decentralized & Multi-hop$\times$Continuous  & Single dimension & user level   \\
    		\textsf{DRescueDP}~\cite{rescuedp2016}\tnote{*} & Decentralized & Multi-hop$\times$Continuous & Correlated dimensions & $w$-event level  \\
            \textsf{DBD/DBA}~\cite{kellaris2014differentially}\tnote{*} & Decentralized & Multi-hop$\times$Continuous & Correlated dimensions & $w$-event level  \\
            \textsf{DPeGaSus}~\cite{chen2017pegasus}\tnote{*} & Decentralized & Multi-hop$\times$Continuous & Independent dimensions & $w$-event level  \\
    		\textsf{DPCrowd} & Decentralized & One-hop$\times$Intermittent & Single dimension & user level  \\
    		\textsf{DPCrowd$_w$} & Decentralized & One-hop$\times$Intermittent &  Independent dimensions & $w$-event level  \\
    		\textsf{DPCrowd+} & Decentralized & One-hop$\times$Intermittent & Correlated dimensions & $w$-event level  \\
    		\hline
	   \end{tabular}
        \begin{tablenotes}
             \item[*] \textsf{DFAST}, \textsf{DRescueDP}, \textsf{DBD/DBA}, \textsf{DPeGaSus} are the decentralized extension (via broadcasting and averaging at each server) of centralized schemes \textsf{FAST}~\cite{fan2014adaptive}, \textsf{RescueDP}~\cite{rescuedp2016}, \textsf{BD/BA}~\cite{kellaris2014differentially}, and \textsf{PeGaSus}~\cite{chen2017pegasus}, respectively.
           \end{tablenotes}
    \end{threeparttable}
	\vspace{-3mm}
\end{table}

Moreover, we compared the performance of our schemes under different strategies and extensions mentioned before, such as the intermittent communication of \textsf{DPCrowd} framework under different sampling strategies (fixed-rate sampling vs. adaptive sampling).

\textbf{Metrics:} In terms of accuracy, we adopted the metric of average relative error (ARE) to measure the relative distance between the final estimation $\hat{x}_i(t)$ and the ground truth $r(t)$, $1\leq i \leq m,~ 1 \leq t \leq T$. The ARE is defined as
\begin{align}\label{eq:are}
ARE(\bm{\hat{x}},\bm{r})=\frac{1}{m}\frac{1}{T}\sum_{i=1}^m \sum_{t=1}^T \frac{|\hat{x}_i(t)-r(t)|}{\max(r(t),\delta)},
\end{align}
where $\delta$ is set as $1$ in case that $x(t)$ is $0$. As observed, smaller ARE means the estimation is more close to the ground truth and have better accuracy. With regard to consensus, we used the metric of average consensus error (ACE) to measure the closeness of estimations $\hat{x}_i (t)$ among distributed servers $1 \leq i \leq m$. The ACE is defined as
\begin{align}\label{eq:ace}
ACE=(\bm{\hat{x}})=\frac{1}{m}\frac{1}{T}\sum_{i=1}^m \sum_{t=1}^T {|\hat{x}_i(t)-\hat{x}_{average}(t)|},
\end{align}
where $\hat{x}_{average}(t)$ is average estimation of all distributed servers. Similarly, smaller ACE means better consensus among distributed servers. For a fair comparison, each set of experiments is run 50 times and the average result is reported.

\begin{table*}[htbp]\vspace{-0.3cm}
	\centering \caption{Parameters Setup for Different Datasets in Experiments}\label{table:parameter}
	\begin{tabular}{|c|c|c|c|c|c|}
		\hline
		\textbf{Symbol} & \textbf{Description} & \textsf{Linear} & \textsf{Flu} & \textsf{Multi-Linear} & \textsf{Multi-Flu} \\
		\hline
		$\varepsilon$ & Total Privacy Budget & $1$ & $1$ & $1$  & $1$\\
		$d$ & Data Dimensionality & $1$ & $1$  & $6$ & $51$ \\
		$(C_p, C_i, C_d)$ & PID Control Gains & $(0.9,0.1,0)$ & $(0.9,0.1,0)$  & $(0.9,0.1,0)$ & $(0.9,0.1,0)$ \\
		$T$ & Total Timestamps & $1000$ & $791$ & $1000$ & $441$ \\
		$T_i$ & Integral Time Window & $5$ & $5$  & $5$ & $5$\\
		$T_s$ & Maximal Sampling Number & $0.3 T$ & $0.4 T$  & N/A & N/A\\
		$(\theta, \xi)$ & Interval Adjustment Parameters & $(2.5, 0.05)$ & $(1, 0.6)$ & ($1$, N/A) & ($3$, N/A)\\
		\hline
	\end{tabular}
	\vspace{-3mm}
\end{table*}

\textbf{Parameters:} The default parameters and their descriptions, unless otherwise explained, are listed in Table~\ref{table:parameter}. In our simulations, we chose the optimal parameters by experimentally minimizing the posterior estimation error. For example, we first chose optimal model variance $Q$ for different datasets; the sampling parameters $M, \theta, \xi$ were chosen separately to minimize the final error; $R$ was approximated according to Eq.~(\ref{eq:approx}) and varied across datasets.


\subsection{Estimation Utility of \textsf{DPCrowd}} \label{subsec:performance-propdp}

\textbf{Convergence of Estimation:} Fig.~\ref{fig:UtilityVsTime} demonstrates the time-varying estimation error of \textsf{DPCrowd} among all distributed servers, in comparison with that of related schemes. In Fig.~\ref{fig:UTLinear}, the relative error of distributed servers in all three schemes gradually drops with time. This is because, distributed servers would initially produce rough prior estimates, which are then gradually corrected via observing new measurements. Nonetheless, without communication, servers in \textsf{FAST} can only observe their own measurements and perform independent estimation, thus converging slowly with much higher consensus error. \textsf{DFAST} simply collects and averages the independent estimations in the whole network. Despite the reduced relative error via averaging, it still has the same convergence speed as \textsf{FAST}, which is determined by the independent estimation. Besides, although \textsf{DFAST} can achieve absolute consensus, it would lead to huge communication cost in blind flooding. Differently, distributed servers in \textsf{DPCrowd} conduct estimation via information fusion with their one-hop neighbors at each time, therefore shows fast convergence. With the increase of time, the estimations of all distributed servers will be finally disseminated and fused according to the weights to achieve both accurate and approximate consensus. Unlike \textsf{Linear} better follows an approximately linear process, \textsf{Flu} has more periodic fluctuations. Nonetheless, in Fig.~\ref{fig:UTFlu}, \textsf{DPCrowd} still shows much better estimation convergence.

\begin{figure*}[htbp]\centering\vspace{-0.3cm}
	\subfloat[\textsf{Linear}\label{fig:UTLinear}]{
		\label{RTLinear} 
		\includegraphics[width=125pt]{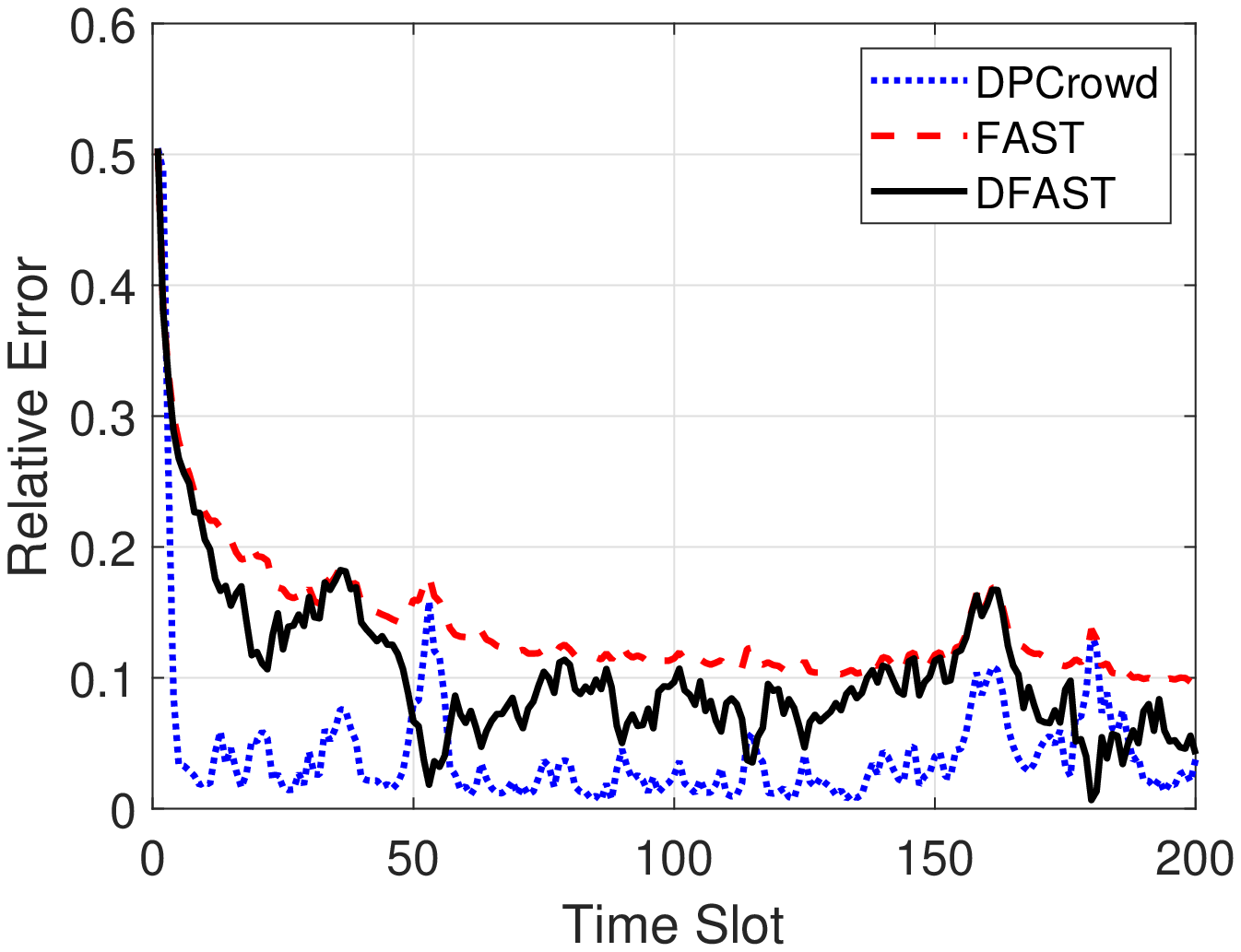}
		\label{CTLinear} 
		\includegraphics[width=125pt]{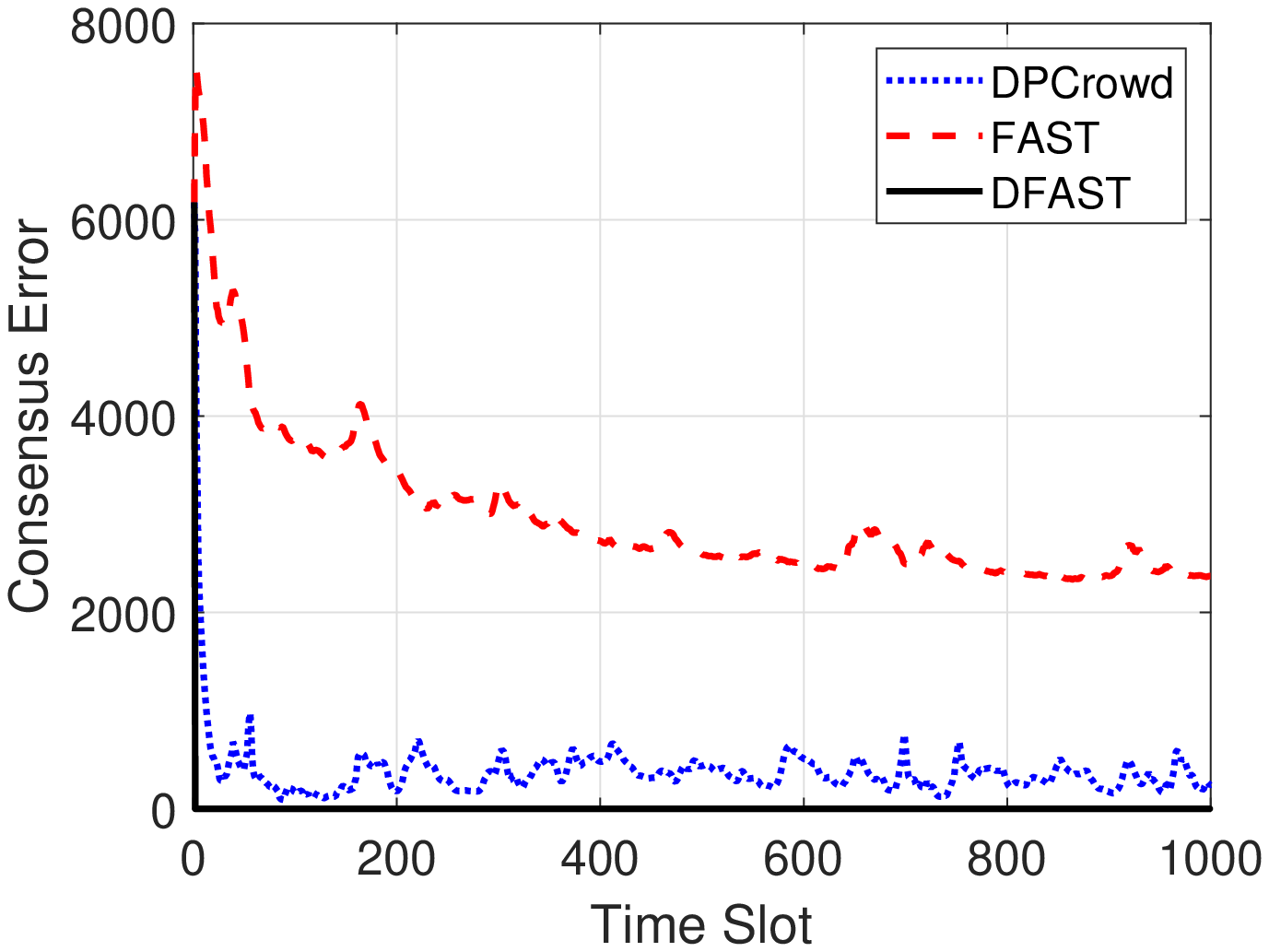}}
	\subfloat[\textsf{Flu}\label{fig:UTFlu}]{
		\label{RTLinear} 
		\includegraphics[width=125pt]{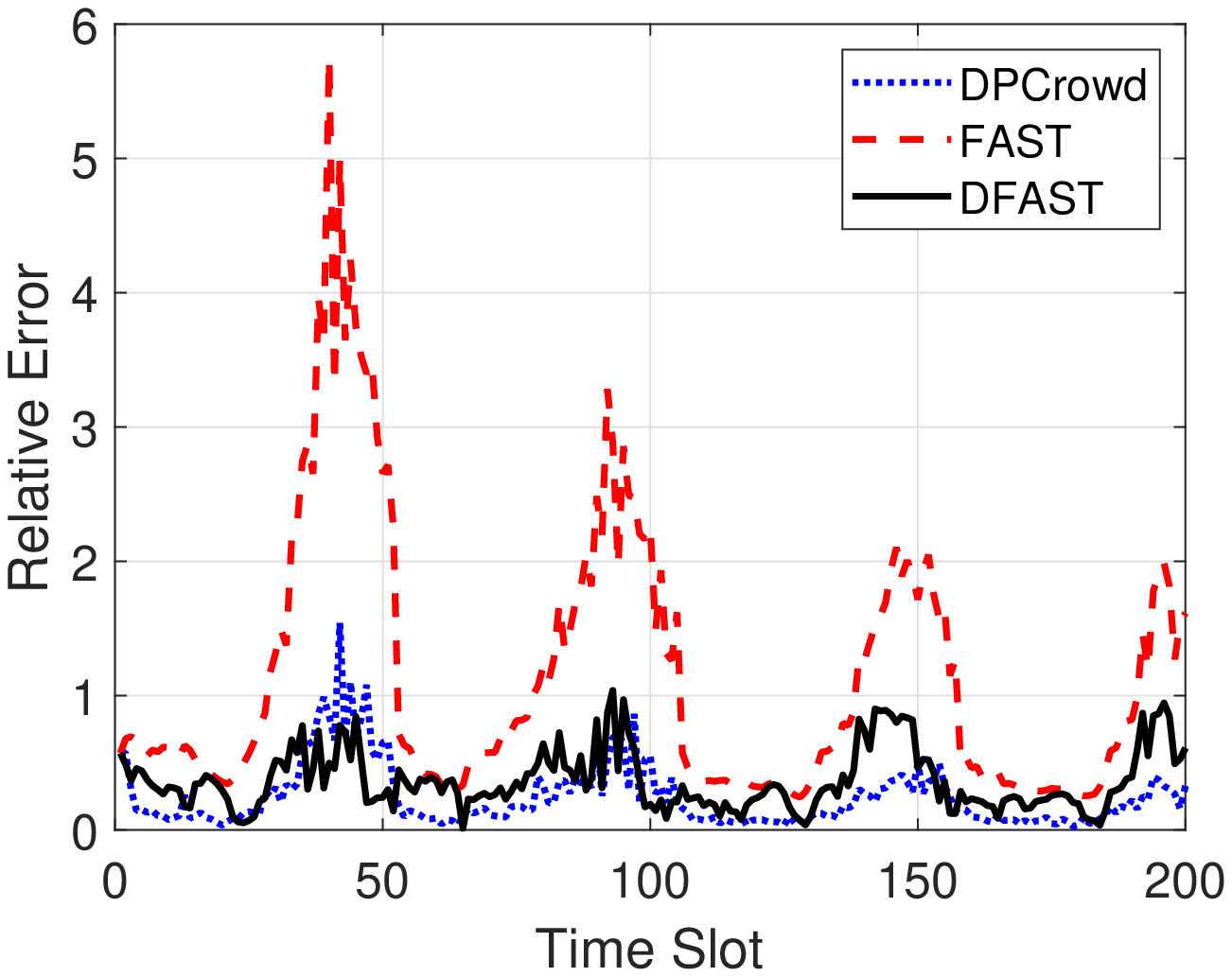}
		\label{CTLinear} 
		\includegraphics[width=125pt]{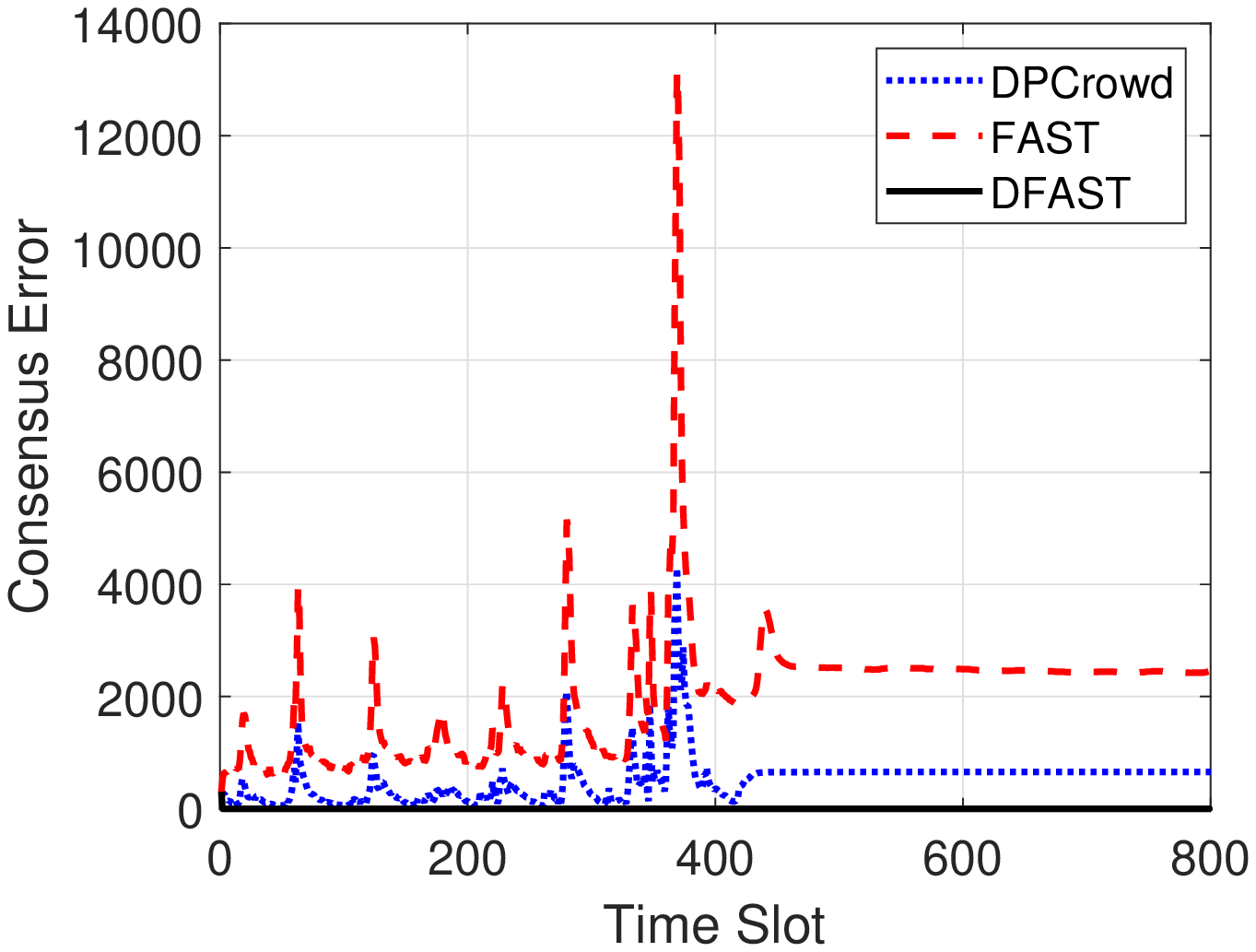}}
	\caption{Error Convergence with Time}\centering
	\label{fig:UtilityVsTime} 
	\vspace{-0.3cm}
\end{figure*}

\begin{figure*}\vspace{-0.3cm}
	\centering
	\subfloat[\textsf{Linear}]{
	\label{ARELinear} 
	\includegraphics[width=125pt]{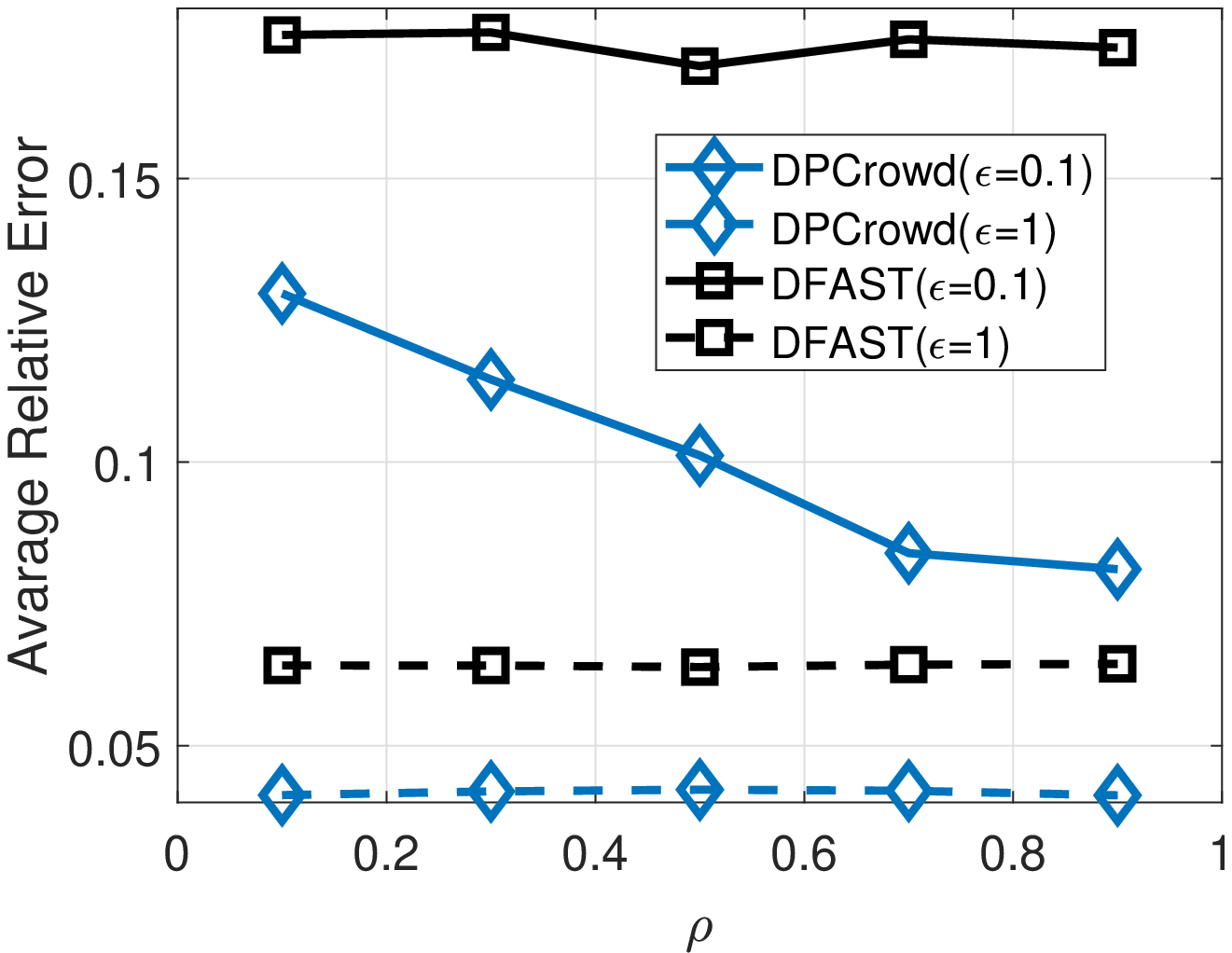}
	\label{ACELinear} 
	\includegraphics[width=125pt]{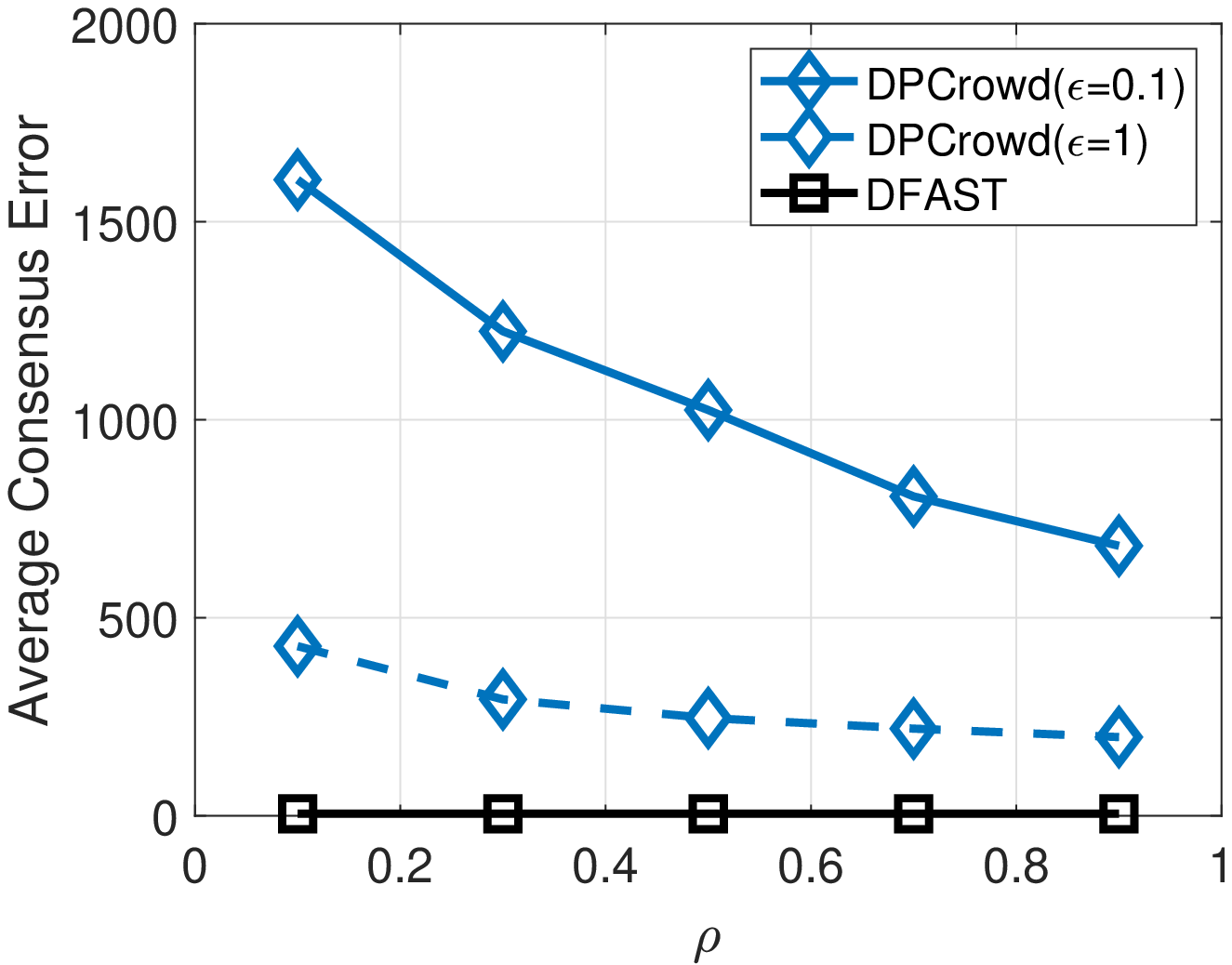}}
\subfloat[\textsf{Flu}]{
	\label{AREFlu} 
	\includegraphics[width=125pt]{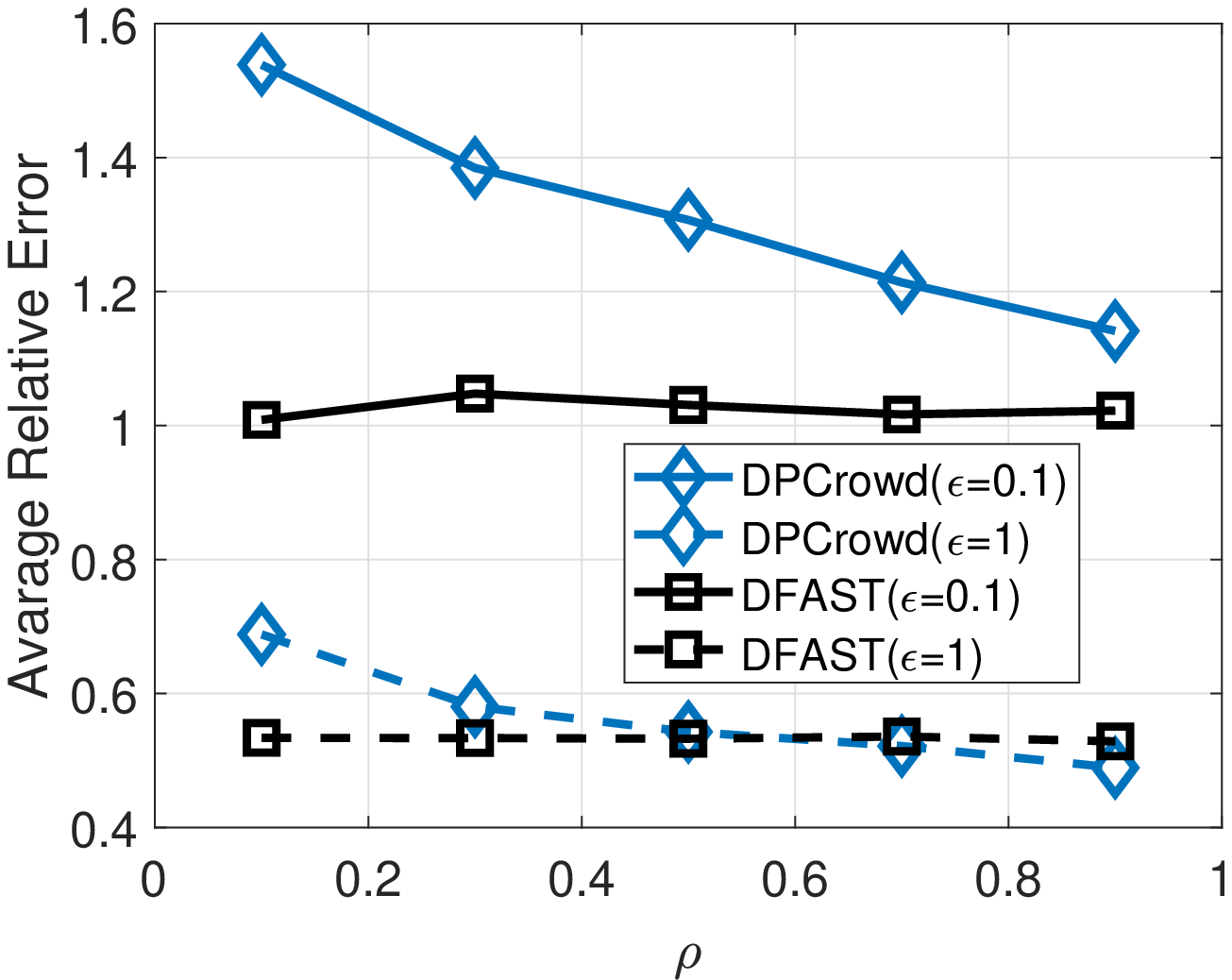}
	\label{ACEFlu} 
	\includegraphics[width=125pt]{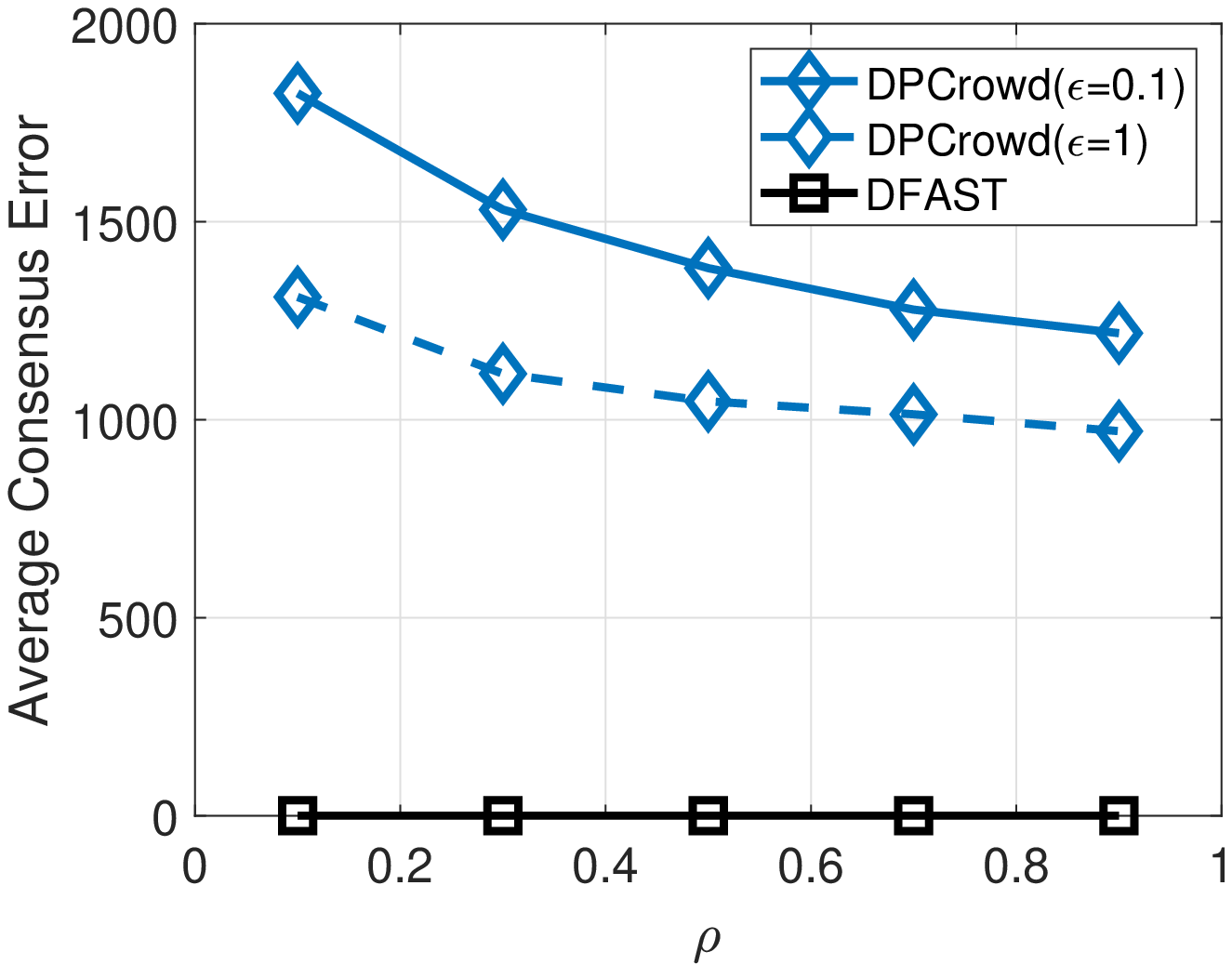}}
	\caption{Average Error vs. Network Density}\centering
	\label{fig:UtilityVsDensity} 
	\vspace{-0.3cm}
\end{figure*}

\textbf{Impact of Network Density:} Fig.~\ref{fig:UtilityVsDensity} presents the impact of network density $\rho$ on both ARE and ACE of \textsf{DPCrowd} compared with \textsf{DFAST}. Since estimation results are broadcast to all servers, the ARE of \textsf{DFAST} remains unchanged for different $\rho$ and its ACE is as small as zero. On both datasets, both the ARE and ACE of \textsf{DPCrowd} decrease with $\rho$ since a denser network can better guarantee the convergence via more extensive communications. Both errors in the stronger privacy regime ($\varepsilon=0.1$) are larger than those in the weaker privacy regime ($\varepsilon=1$), which reflects the utility-privacy tradeoff. When $\varepsilon=1$, both ARE and ACE are not sensitive to $\rho$. This is because, with less noise, distributed servers can easily achieve accurate and consensus estimation with fewer neighbors. As shown, ARE and ACE of \textsf{DPCrowd} also vary across datasets. As analyzed before, unlike \textsf{Linear}, the higher ARE and ACE of \textsf{DPCrowd} on \textsf{Flu} result from the large fluctuations of both dataset.

\begin{figure*}\vspace{-0.3cm}
	\centering
	\subfloat[\textsf{Linear}]{
		\label{ARELinear-Sampling} 
		\includegraphics[width=125pt]{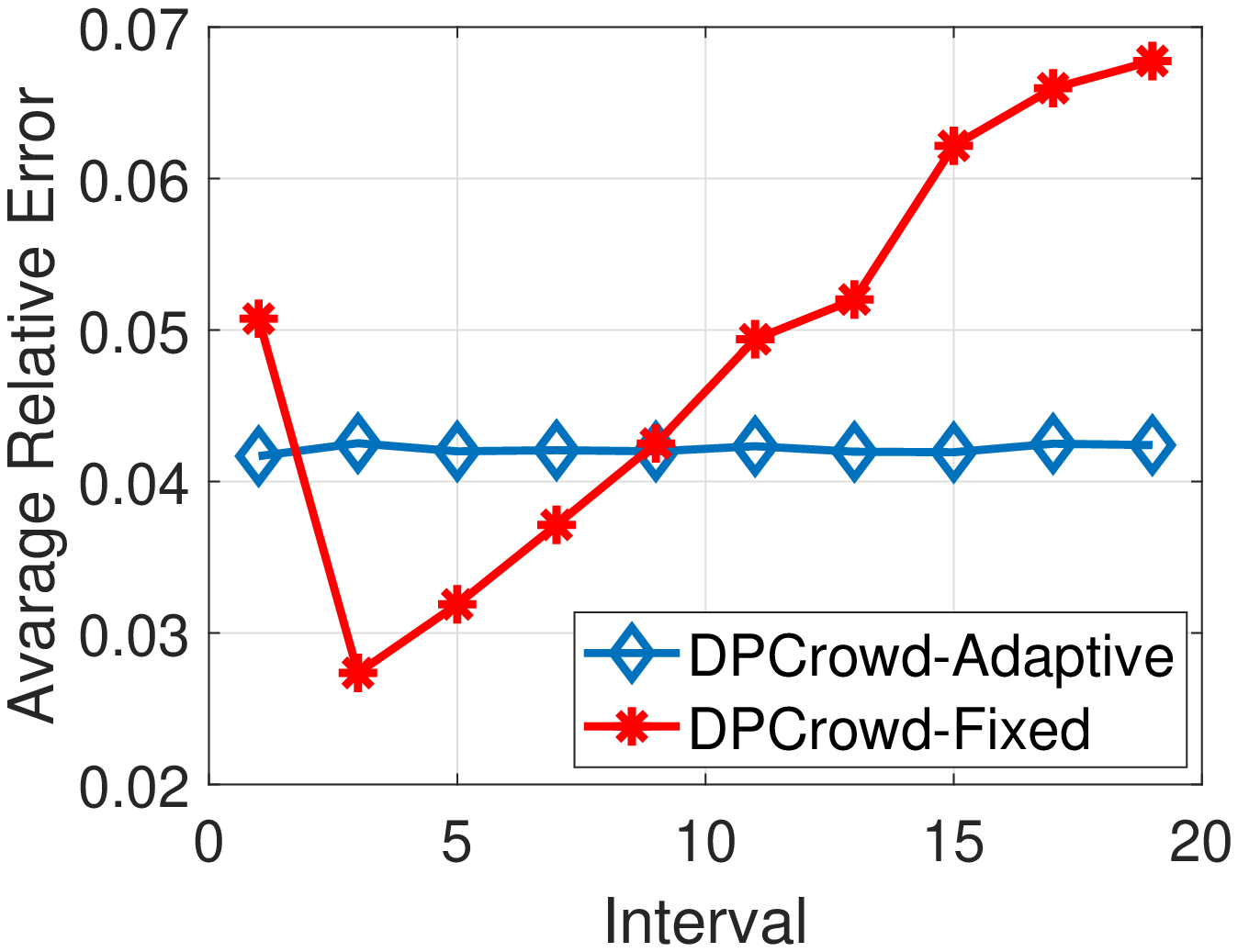}
		\label{ACELinear-Sampling} 
		\includegraphics[width=125pt]{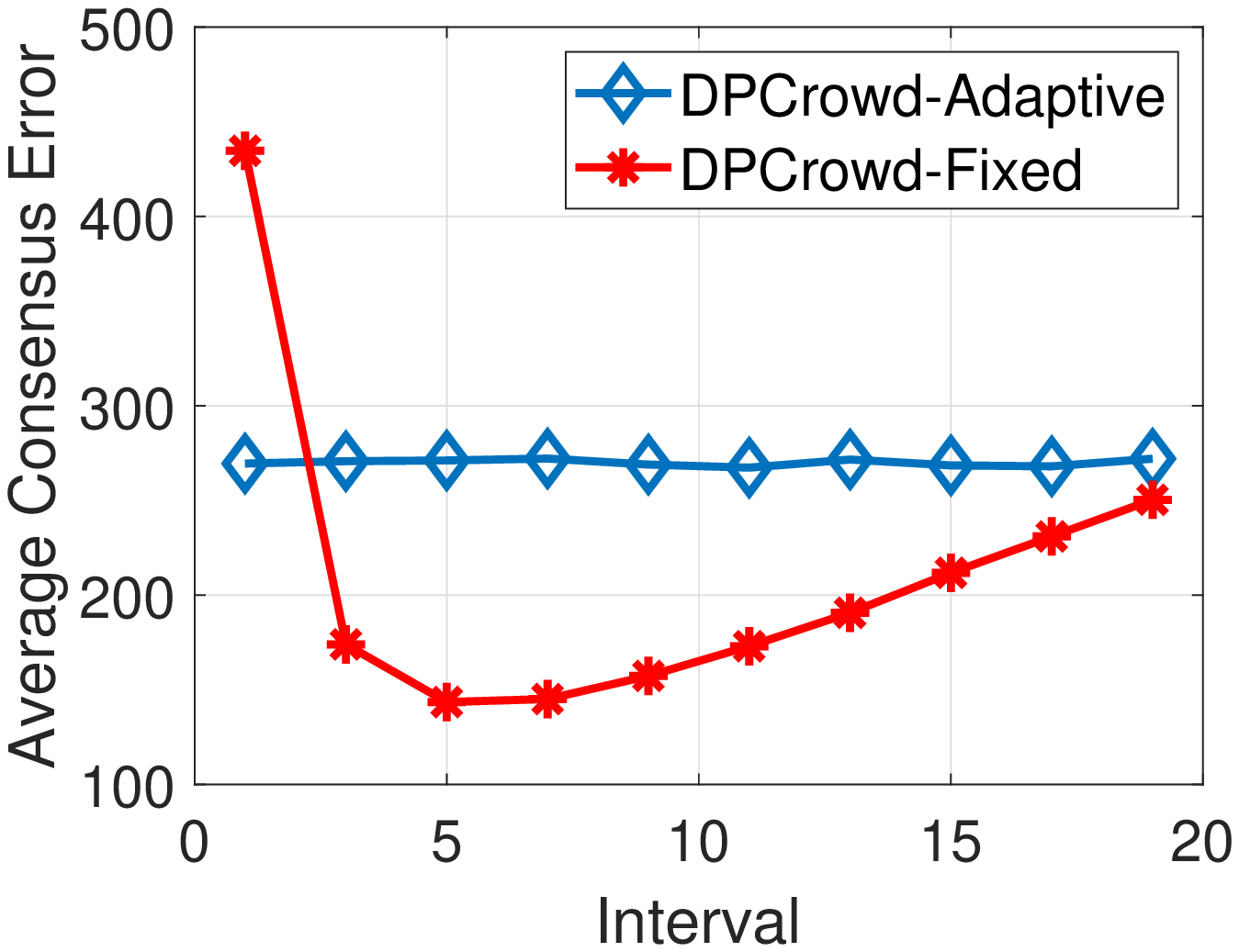}}
	\subfloat[\textsf{Flu}]{
		\label{AREFlu-Sampling} 
		\includegraphics[width=125pt]{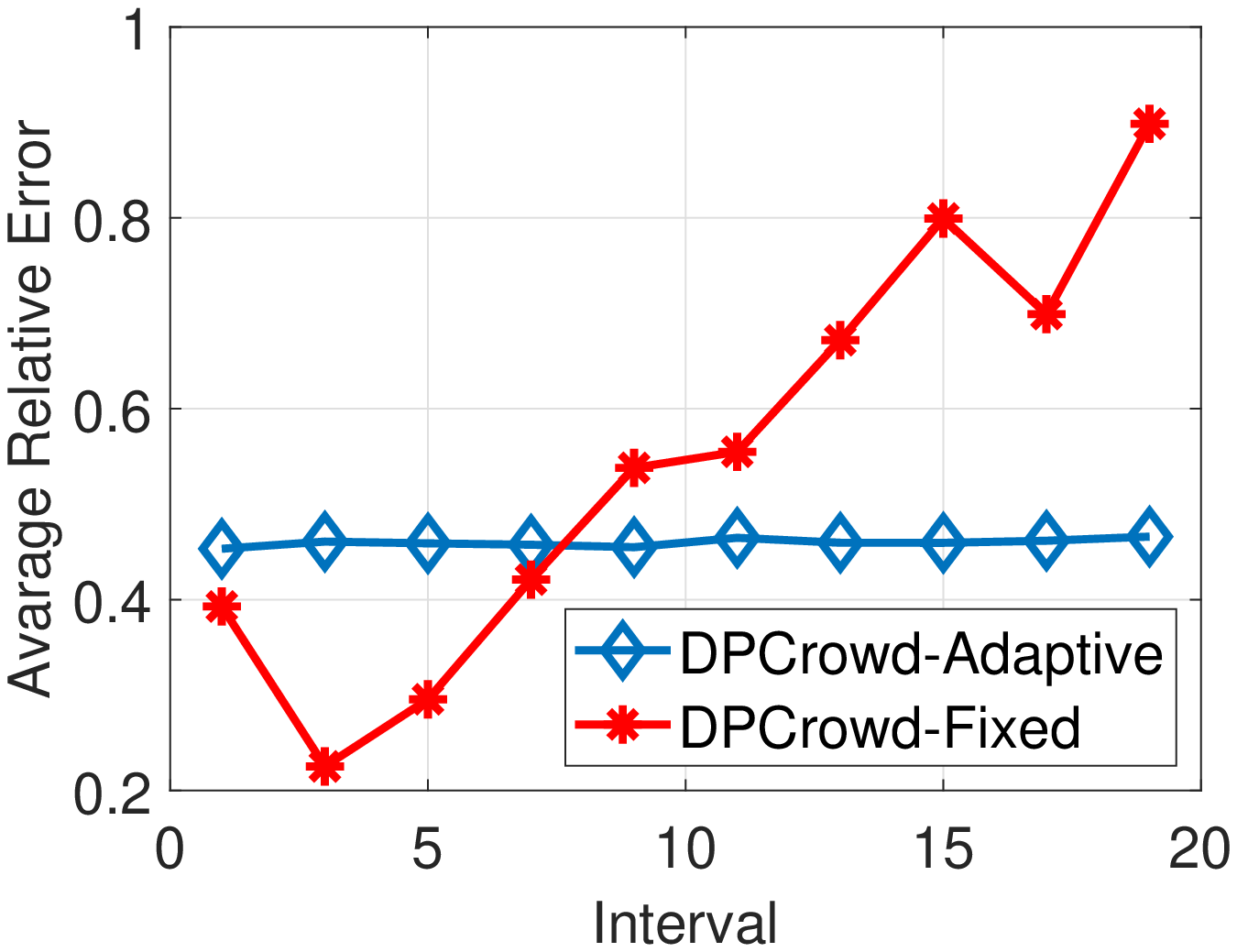}
		\label{ACEFlu-Sampling} 
		\includegraphics[width=125pt]{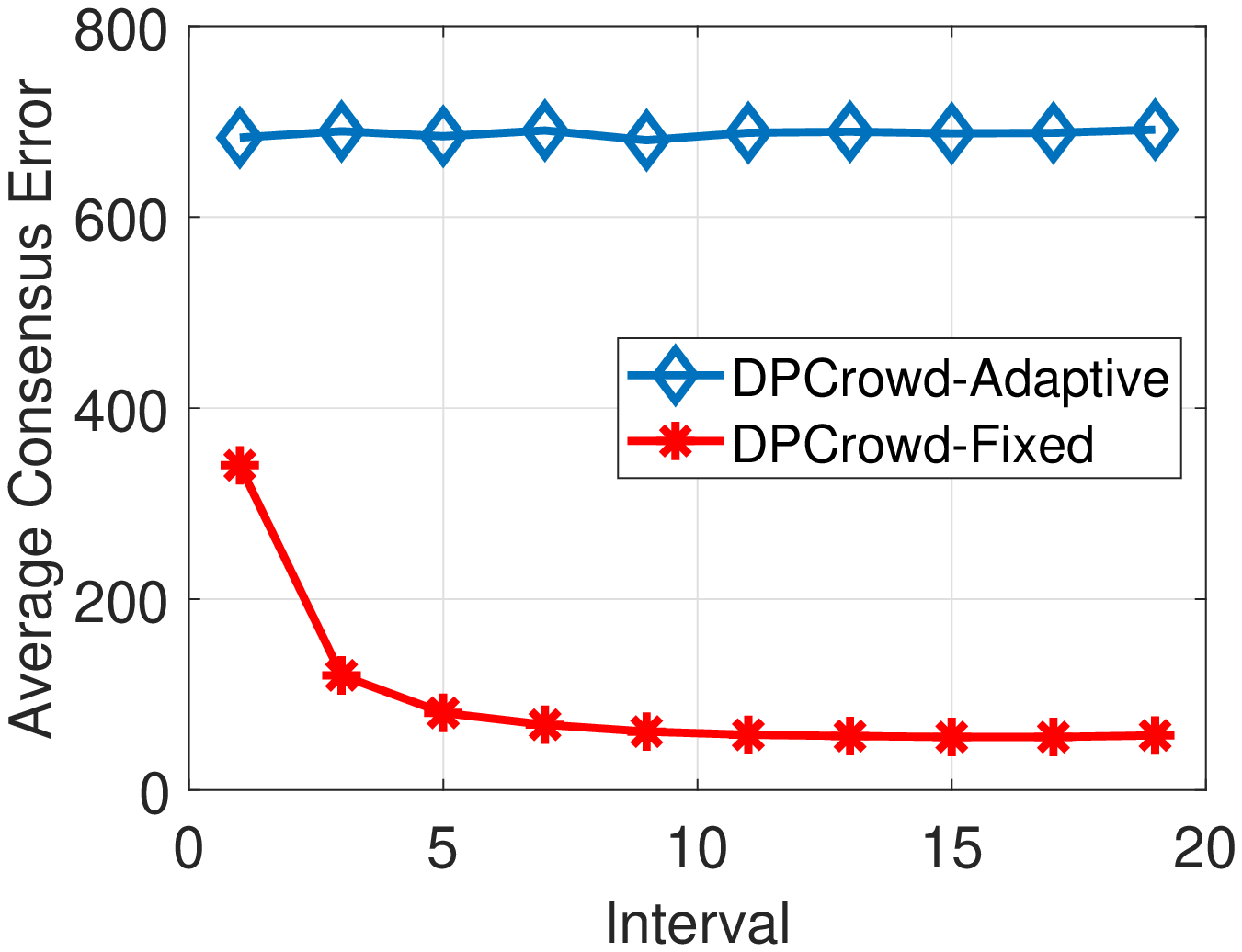}}
	\caption{Average Error vs. Sampling Interval}\centering
	\label{fig:UtilityVsSampling} 
	\vspace{-0.3cm}
\end{figure*}

\textbf{Impact of Sampling Strategy:} Fig.~\ref{fig:UtilityVsSampling} reports the ARE and ACE of \textsf{DPCrowd} under both fixed and adaptive sampling based intermittent communication strategies. Because of the adaptiveness, \textsf{DPCrowd-Adaptive} keeps nearly the same error. Nonetheless, \textsf{DPCrowd-fixed} varies greatly with different sampling intervals. When the interval is small (such as $1$), it incurs high perturbation error on both datasets because much noise is injected at nearly every timestamp. When the sampling interval increases slightly, it performs better since less noise is injected in a sampling manner. Nonetheless, with the further increase of intervals, the ARE of \textsf{DPCrowd-fixed} increases gradually and goes beyond that of \textsf{DPCrowd-Adaptive} since larger sampling interval will lead to larger prediction error in spite of smaller perturbation error. Similar trends can be observed in the ACE comparison. \textsf{DPCrowd-fixed} seems to show a smaller consensus error. The reason is that consensus error mainly comes from the perturbation error, which is much smaller when there are more non-sampling stamps. In other words, larger sampling interval means more non-sampling points and less dynamic changes, which naturally lead to better consensus, but less accuracy. Thus, ARE and ACE should be combined to analyze the performance of \textsf{DPCrowd}. Overall, \textsf{DPCrowd} under the adaptive sampling strategy is more robust to different datasets.

\textbf{Tradeoff between Utility and Privacy:} Fig.~\ref{fig:UtilityVsPrivacy} compares both ARE and ACE of \textsf{DPCrowd} with \textsf{FAST} and \textsf{DFAST} under different privacy $\varepsilon$. All AREs decrease with $\varepsilon$, which demonstrates the trade-off between utility and privacy. However, the ARE of \textsf{DPCrowd} and \textsf{DFAST} is consistently lower than that of \textsf{FAST} as both schemes can greatly improve the estimation via communications. Furthermore, \textsf{DPCrowd} incurs less ARE than \textsf{DFAST} in most cases since the estimation can be better corrected according to the weights of different servers (Eq.~\ref{eq:weight}). Compared with Fig.~\ref{AREFlu}, Fig.~\ref{ARELinear} has the lowest ARE as the synthetic \textsf{Linear} perfectly follows the known process model. Whereas \textsf{Flu} have more fluctuations. Due to whole-network broadcast at the expense of large overhead, \textsf{DFAST} can achieve almost the same estimation for all servers and therefore incurs no consensus error. Compared with the non-communication scheme \textsf{FAST}, \textsf{DPCrowd} has much smaller ACE on both datasets for all privacy levels. The reason is all distributed servers in \textsf{DPCrowd} can exchange and disseminate information iteratively until the convergence. In general, with the increase of $\varepsilon$, ACE for both \textsf{DPCrowd} and \textsf{FAST} drop slowly since fewer noises are added and the differences among servers become smaller.

\begin{figure*}\vspace{-0.3cm}
	\centering
	\subfloat[\textsf{Linear}]{
		\label{ARELinear} 
		\includegraphics[width=125pt]{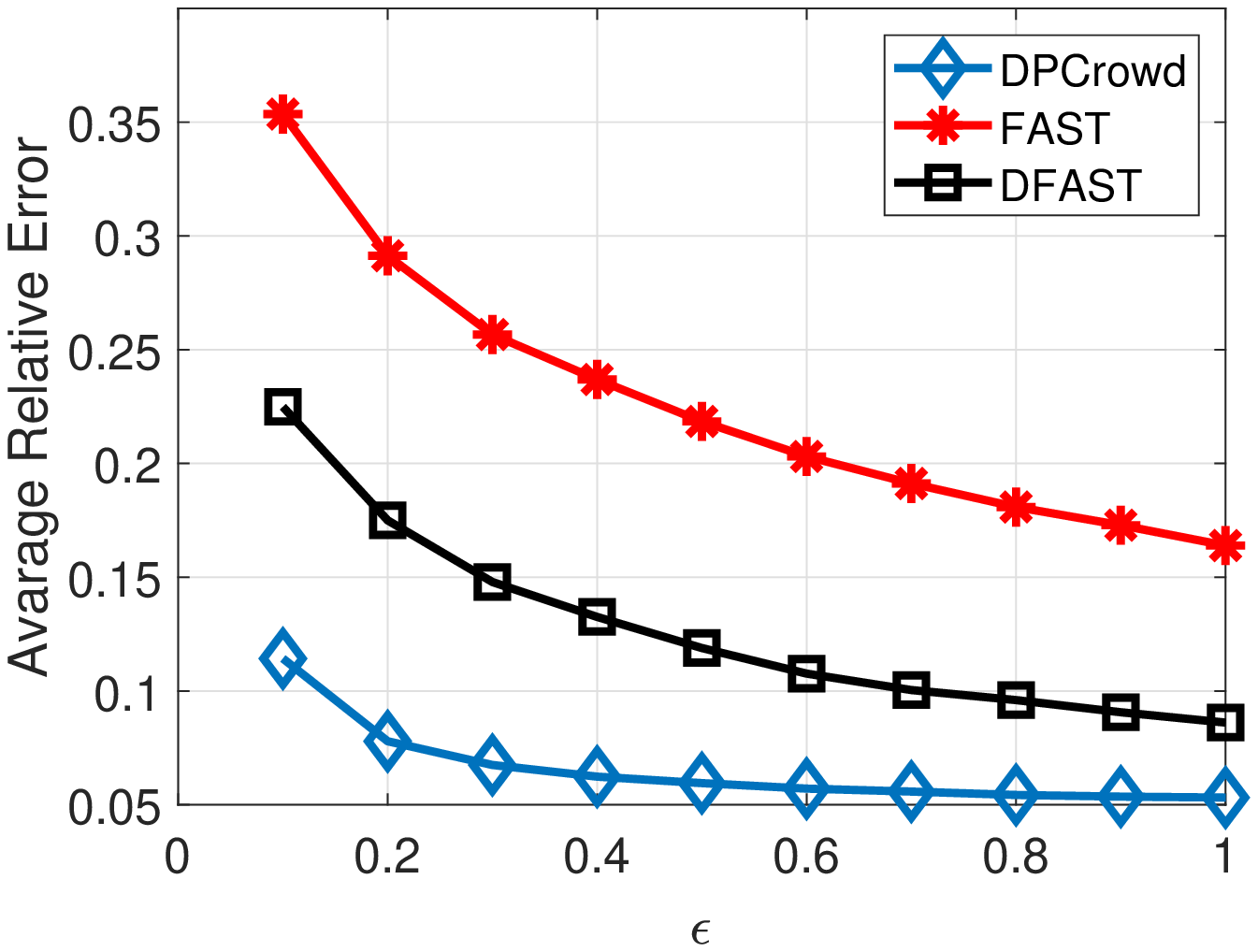}
		\label{ACELinear} 
		\includegraphics[width=125pt]{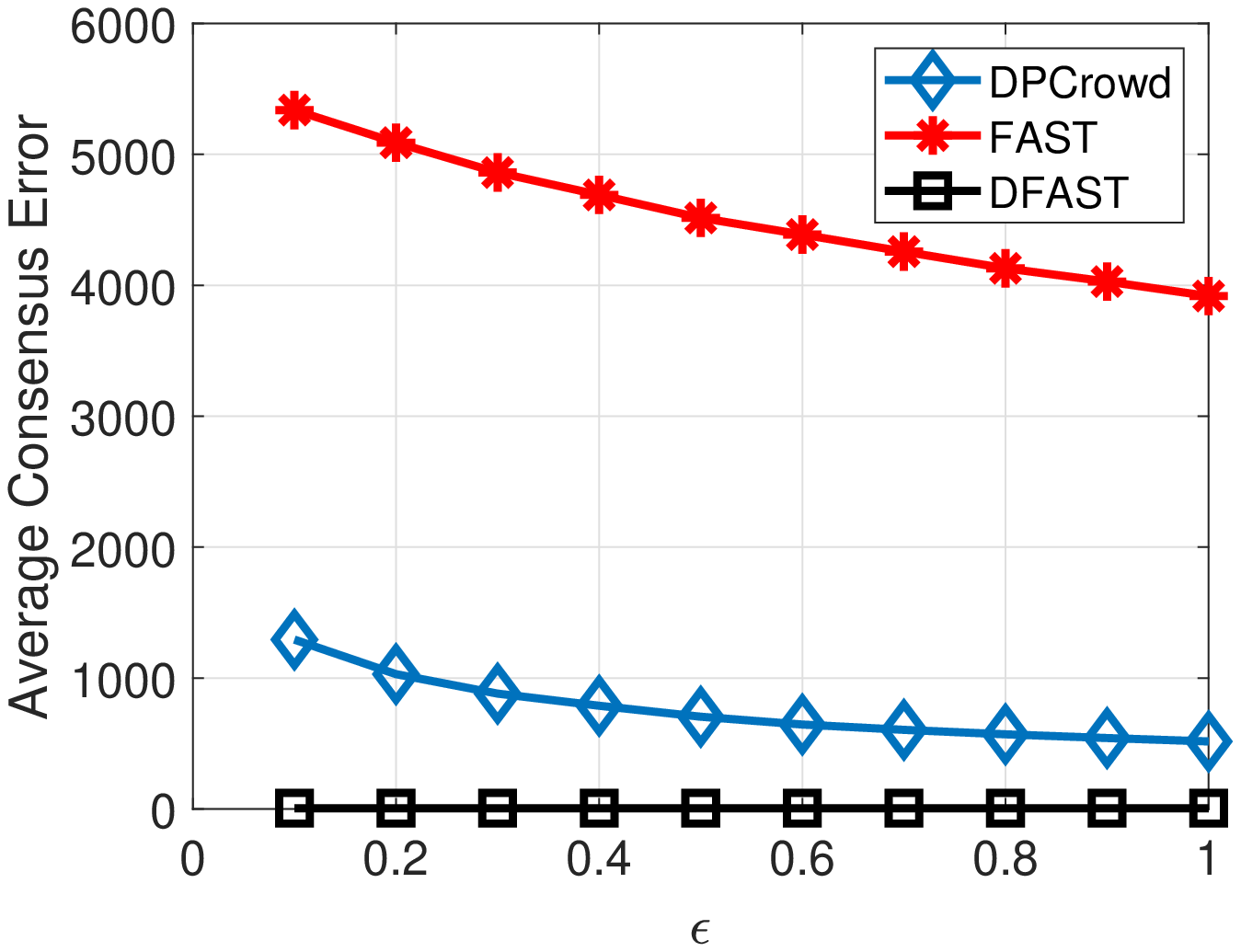}}
	\subfloat[\textsf{Flu}]{
		\label{AREFlu} 
		\includegraphics[width=125pt]{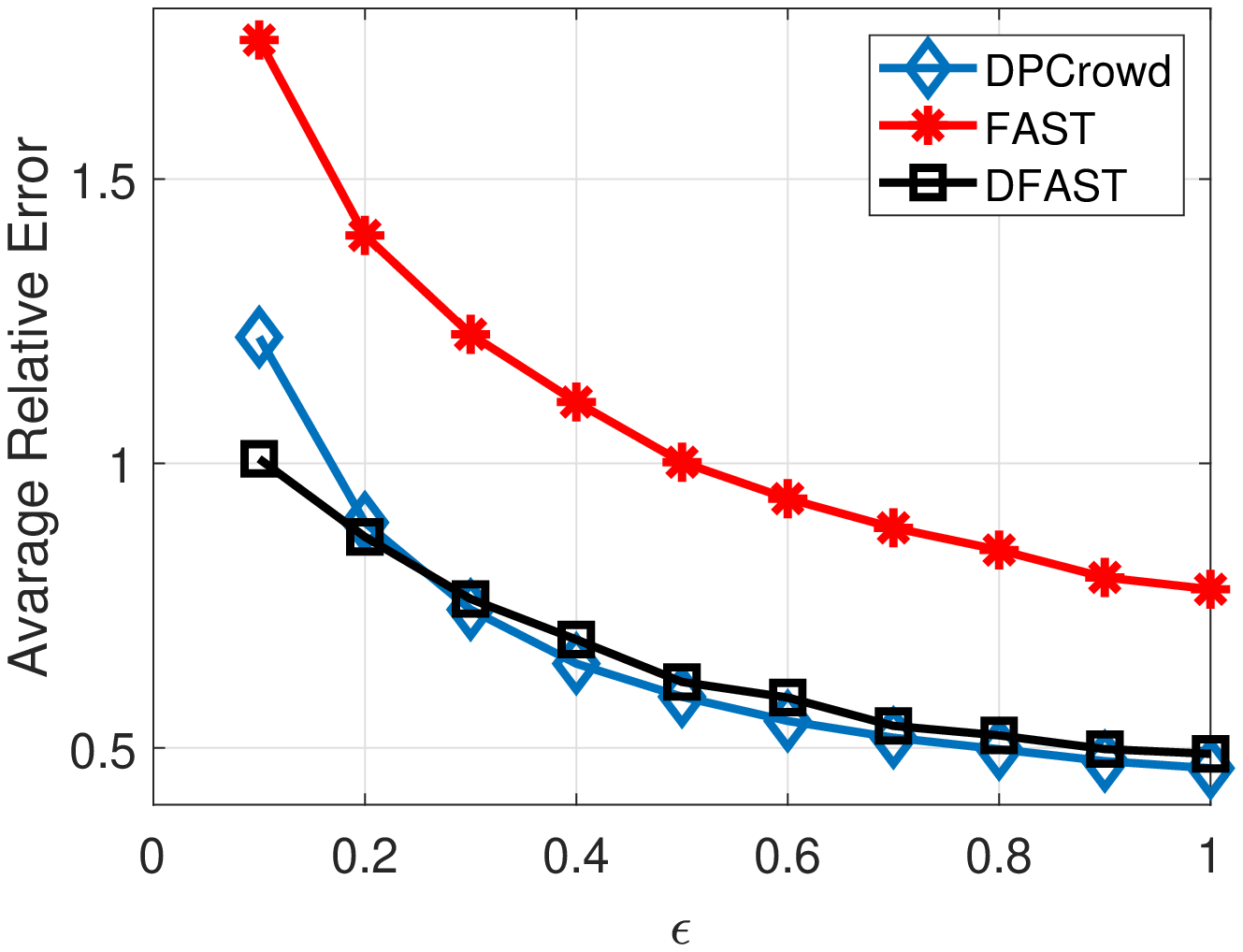}
		\label{ACEFlu} 
		\includegraphics[width=125pt]{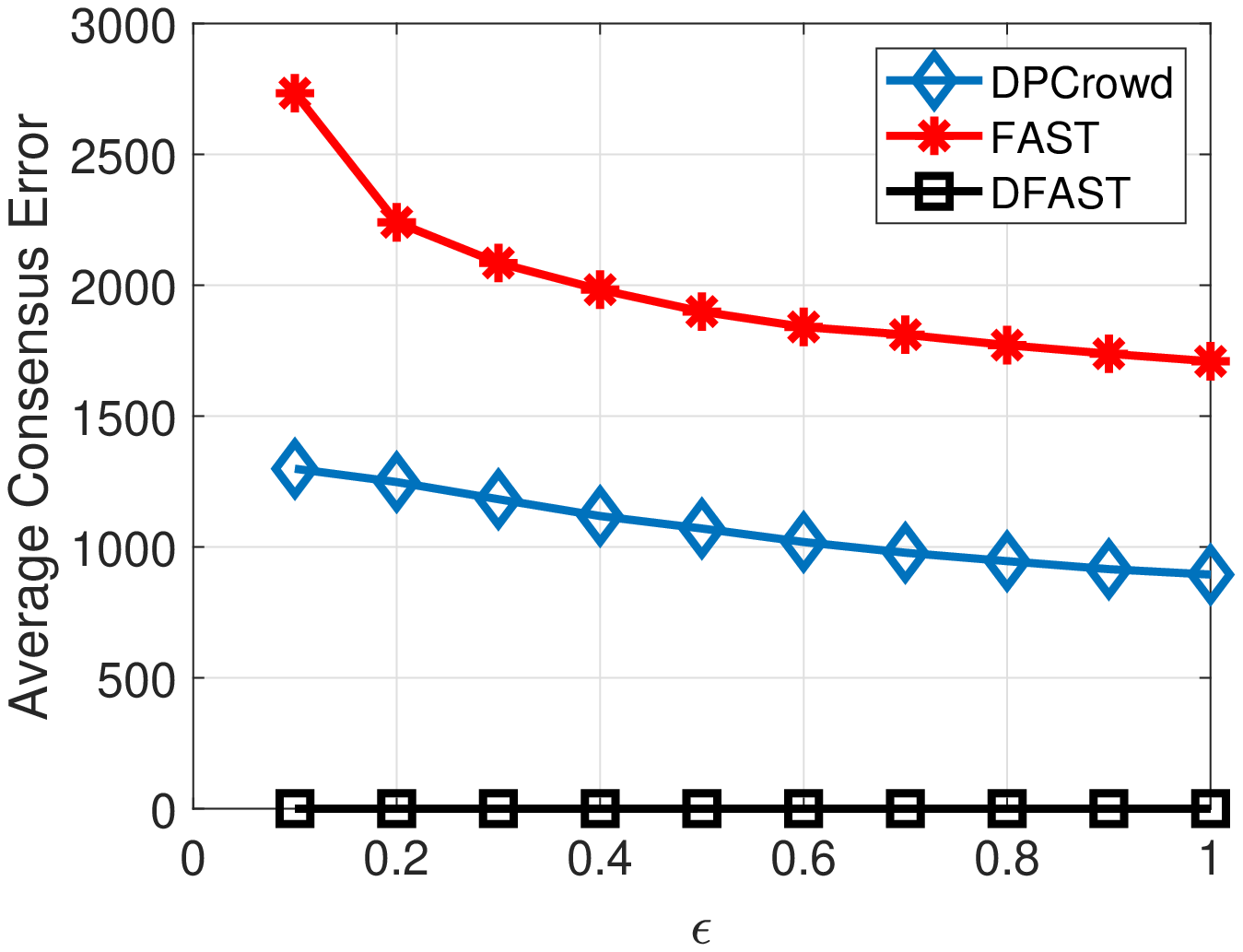}}
	\caption{Average Error vs. Privacy}\centering
	\label{fig:UtilityVsPrivacy} 
	\vspace{-0.3cm}
\end{figure*}

\subsection{Estimation Efficiency of DPCrowd}

The communication efficiency of \textsf{DPCrowd} results from two aspects: communication with only one-hop neighbors and communication frequency reduction via sampling based intermittency.

\begin{figure*}[htbp] \centering\vspace{-0.3cm}
	\subfloat[\textsf{Communication Time}]{
		\label{fig:runtime-density} 
		\includegraphics[width=125pt]{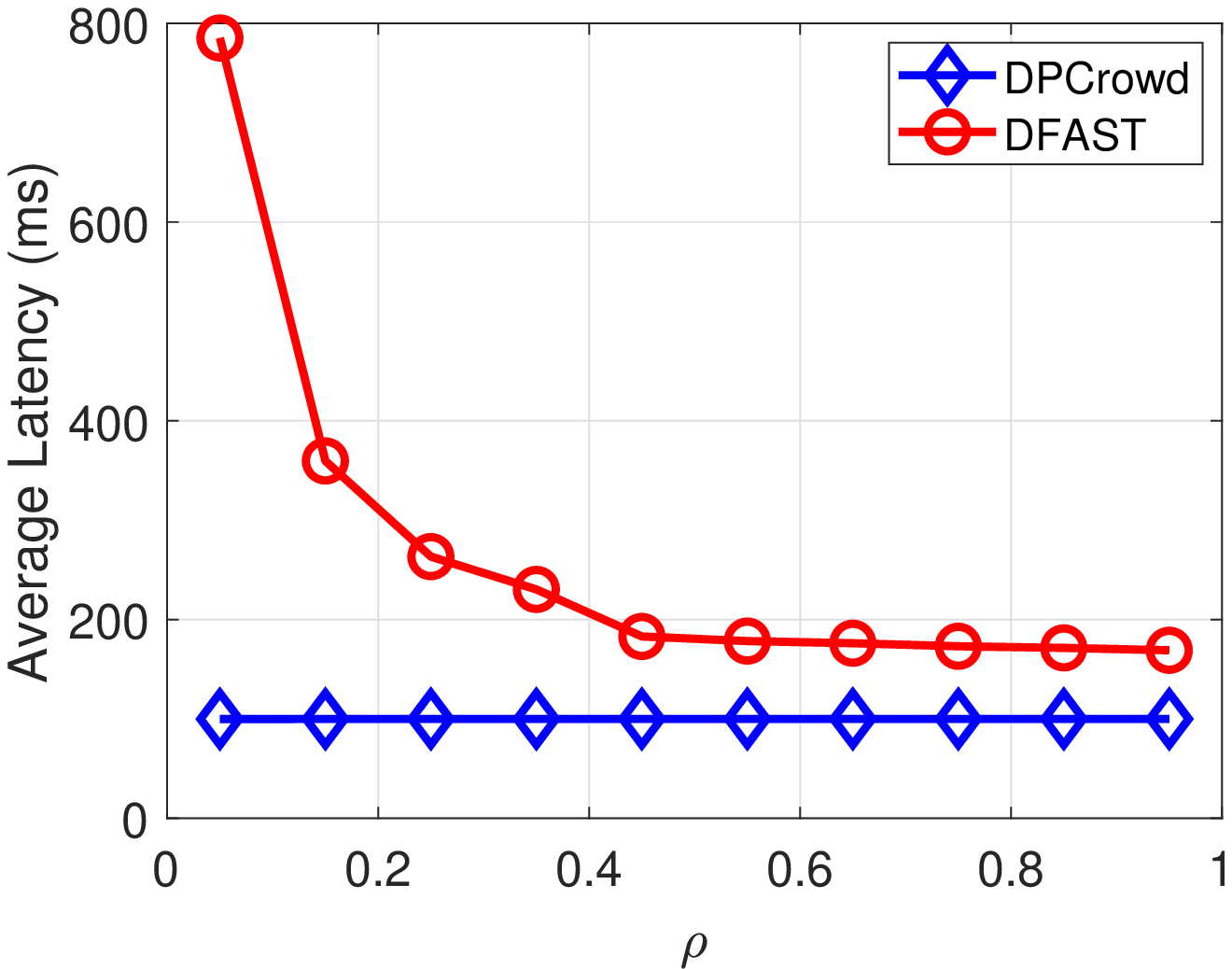}}
	\subfloat[\textsf{Communication Overhead}]{
		\label{fig:overhead-density} 
		\includegraphics[width=125pt]{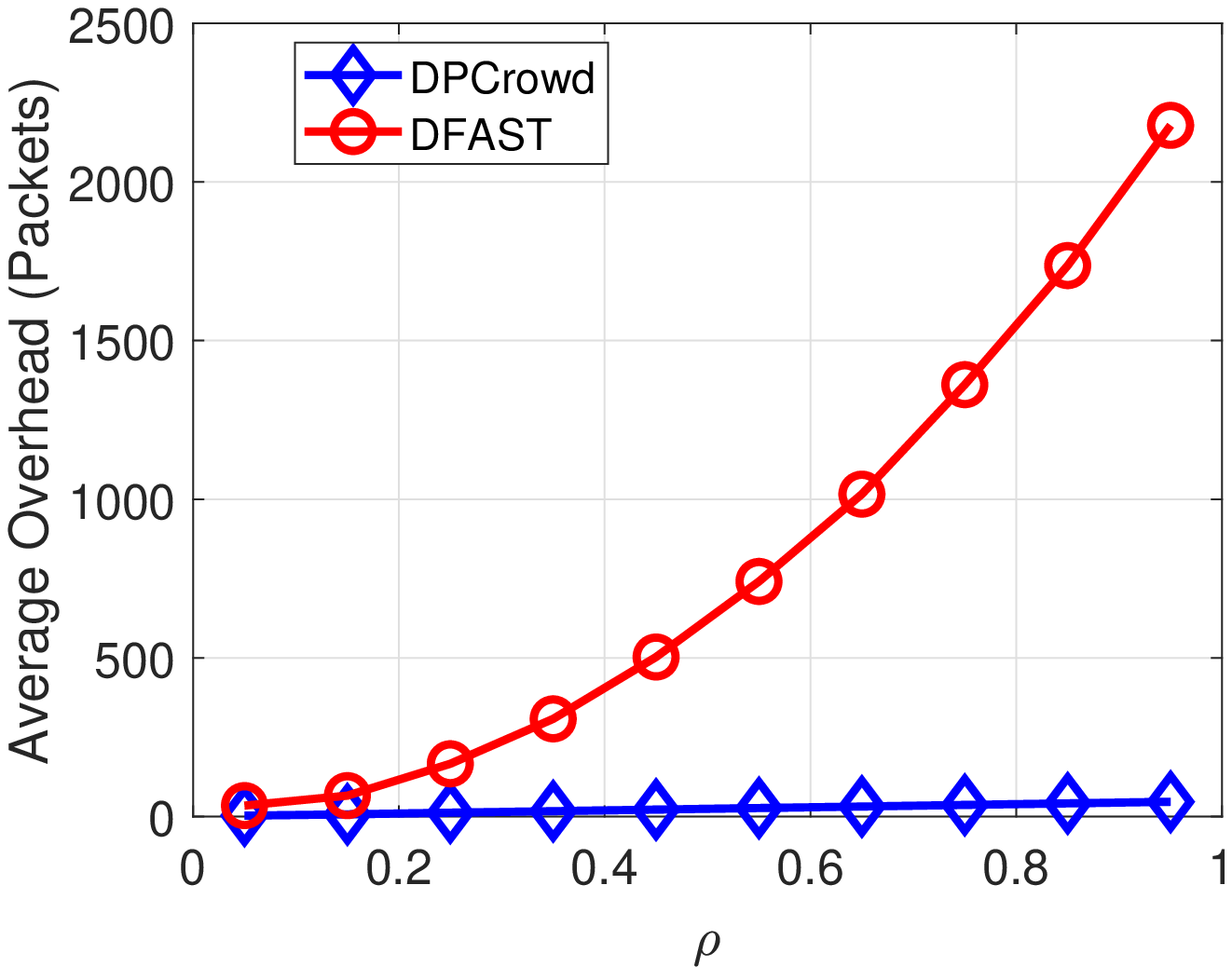}}
	\subfloat[\textsf{Frequency vs. Interval (Linear)}]{
		\label{fig:frequency(Liner)} 
		\includegraphics[width=125pt]{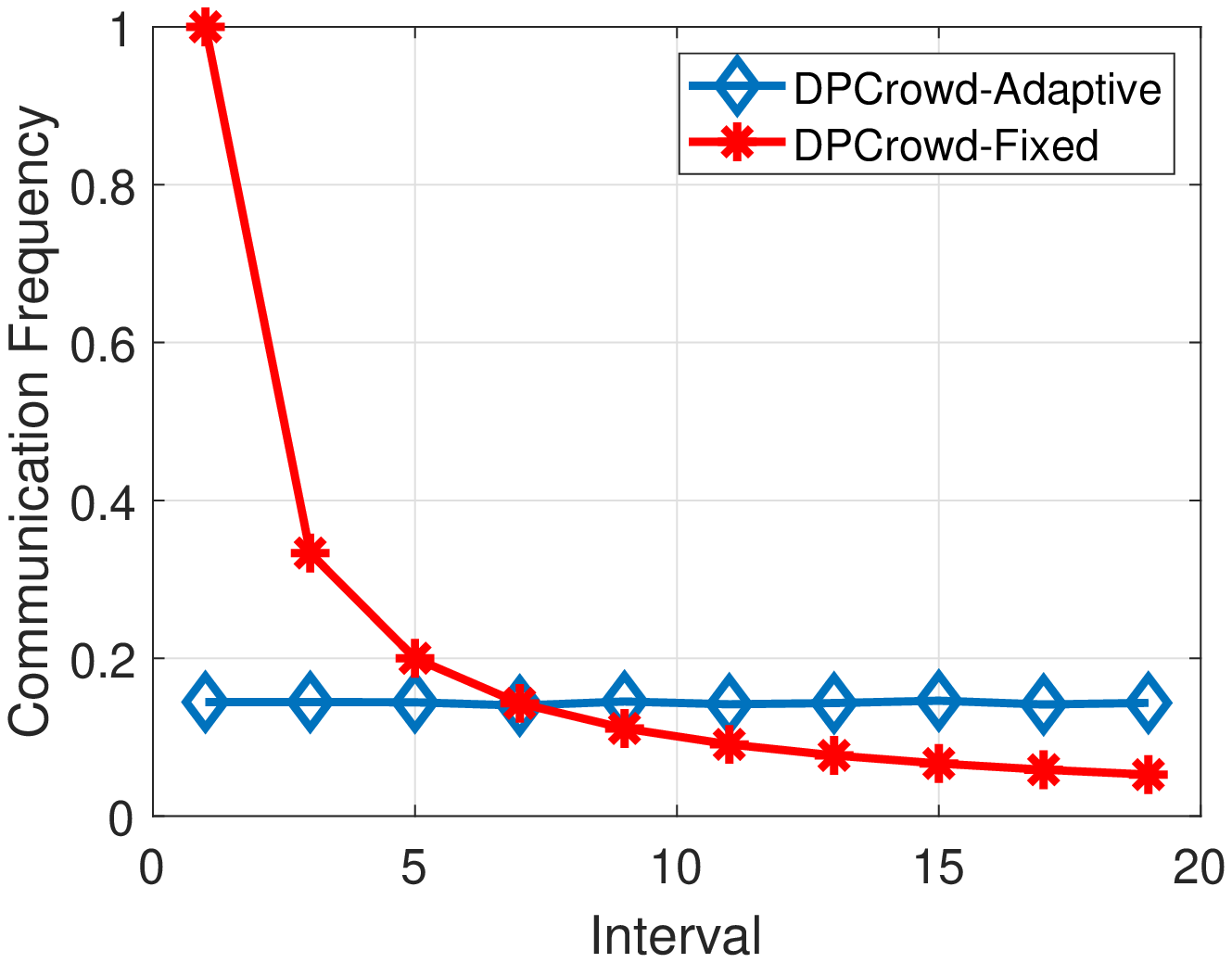}}
	\subfloat[\textsf{Frequency vs. Interval (Flu)}]{
		\label{fig:frequency(Flu)} 
		\includegraphics[width=125pt]{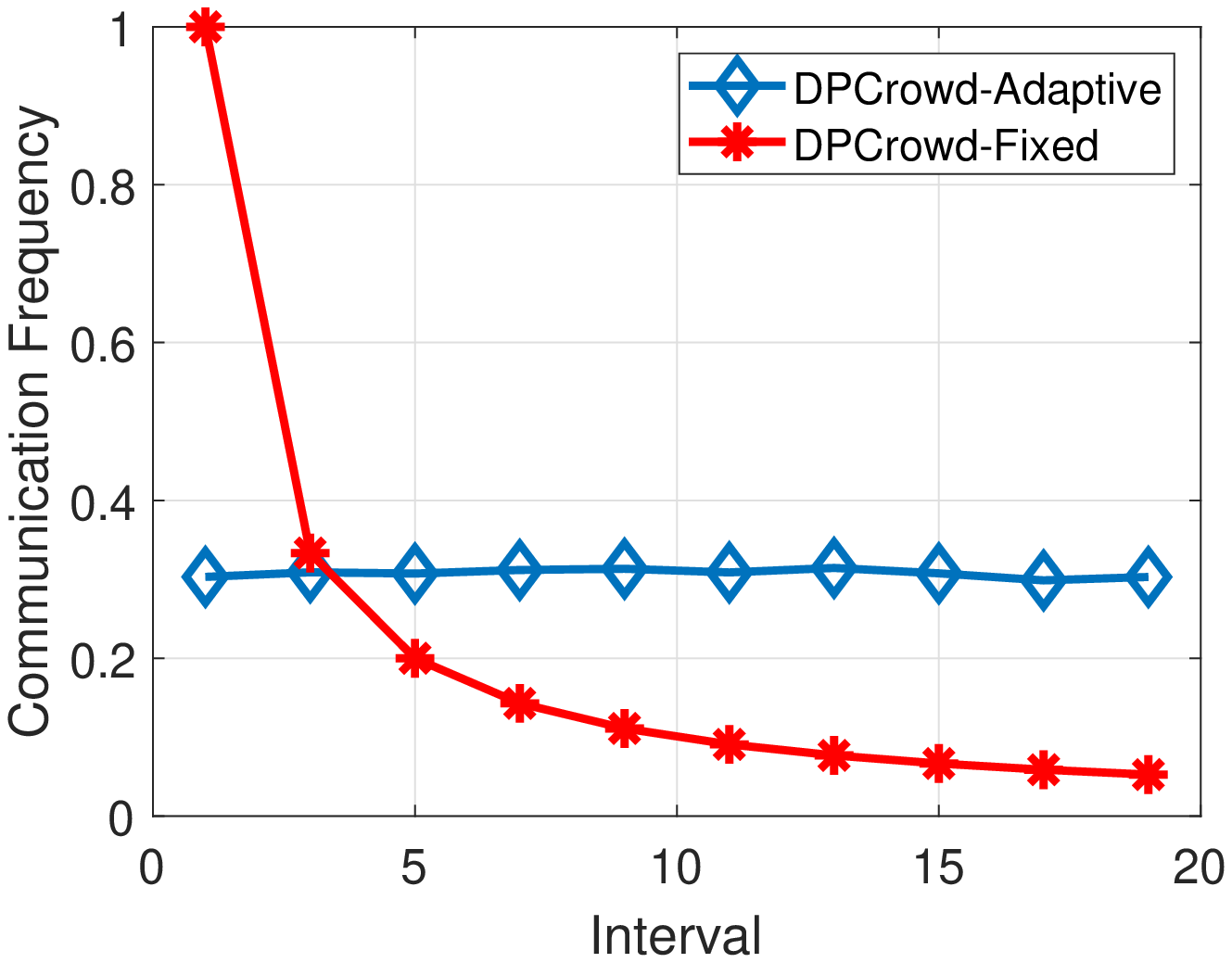}}
	\caption{Communication Efficiency of \textsf{DPCrowd}}\centering
	\label{fig:ComplexVsDensity} 
	\vspace{-0.3cm}
\end{figure*}

\textbf{Communication Latency and Overhead:} Figs.~\ref{fig:runtime-density} and \ref{fig:overhead-density} show the communication latency and overhead of \textsf{DPCrowd} in comparison with \textsf{DFAST} under different network density $\rho$. In Fig.~\ref{fig:runtime-density}, \textsf{DPCrowd} keeps much less latency since each server only exchanges messages with its one-hop neighbors. Nonetheless, the baseline scheme \textsf{DFAST} incurs much larger latency because the multi-hop broadcast requires much more time to ensure the full dissemination of information. When the network is sparser (smaller $\rho$), there are fewer viable routines among servers and cost more communication time. In Fig.~\ref{fig:overhead-density}, \textsf{DPCrowd} incurs much less communication packets in each sampling timestamp, which increases with $\rho$ slowly. However, the communication overhead of \textsf{DFAST} increases significantly with the density $\rho$. The reason is that the messages have to be broadcast to the whole network via hop by hop. With the increase of network density, more redundant messages will be forwarded and relayed.

\textbf{Communication Frequency Reduction:} Figs.~\ref{fig:frequency(Liner)} and \ref{fig:frequency(Flu)} depict the average communication frequency of \textsf{DPCrowd} under both the fixed-rate and adaptive rate sampling based intermittent communication strategies. The average communication frequency decreases as the sampling interval increases in the fixed rate strategy, but keeps a lower level for the adaptive sampling strategy given the maximal sampling points ($0.3 T$ and $0.4 T$ for \textsf{Linear} and \textsf{Flu}, respectively). Together with Fig.~\ref{fig:UtilityVsSampling}, we can say the sampling based intermittent communication can effectively reduce the communication frequency and better utilize the privacy budget. Especially, adaptive sampling can better find an optimal sampling interval for \textsf{DPCrowd} with higher efficiency in both communication and privacy preservation.

For conciseness, we mainly compared \textsf{DPCrowd} with \textsf{DFAST}. It should note that, the experimental conclusions of \textsf{DFAST} in terms of communication efficiency also apply to other extension schemes including \textsf{DBD/DBA}, \textsf{DRescueDP}, and \textsf{DPeGaSus}. Apparently, similar reduction in both communication latency and overhead can also be achieved by \textsf{DPCrowd}+ when compared to its counterpart \textsf{DRescueDP}, which is the decentralized extension of \textsf{RescueDP}~\cite{rescuedp2016}.


\subsection{Overall Performance of \textsf{DPCrowd}+} \label{subsec:performance-promisedp}

\textbf{Impact of Windows Size:} Fig.~\ref{fig:MUtilityVSw} shows the estimation utility of \textsf{DPCrowd}+ with the varying windows size $w$, in comparison with other comparable schemes. The AREs of all schemes increase with $w$ for both datasets. This is because given certain privacy budget $\varepsilon$ for a sliding window, larger $w$ means smaller privacy budget for each timestamp and higher perturbation error. The ARE of \textsf{RescueDP} increases with $w$ sharply and reaches the highest in both datasets due to the lack of collaborations among servers. While \textsf{DPCrowd}+ and \textsf{DPCrowd$_w$} show relatively steadily increasing trends. \textsf{DPCrowd}+, compared with \textsf{DPCrowd$_w$}, can not only utilize neighbors' knowledge, but also reduce the error via adapting dynamic grouping strategy on the dimensions with small values. Moreover, \textsf{DPCrowd}+ shows superior performance than \textsf{DPeGaSus} and \textsf{DBD/BA}, which is because of further consideration of estimation weights in Eq.~(\ref{eq:weight}).

Similarly, \textsf{DPeGaSus} and \textsf{DBD/BA} have almost no consensus error with the cost of whole-network broadcast; and the ACEs of all other schemes increase with $w$ due to less privacy budget allocated on each timestamp. \textsf{RescueDP} shows the largest ACE since there is no collaboration. Although collaborations in \textsf{DPCrowd$_w$} can help to reduce the ACE of \textsf{RescueDP}, higher fluctuations and dimensionality of \textsf{Multi-Flu} still lead to high sparsity and make \textsf{DPCrowd$_w$} less effective. In contrast, with the dynamic dimension reduction, \textsf{DPCrowd}+ can achieve better consensus by mitigating the sparsity issue in high-dimensional data.

\begin{figure*}[htbp]\vspace{-0.3cm}
	\centering
	\subfloat[\textsf{Multi-Linear}]{
		\label{AREMLinear-w} 
		\includegraphics[width=125pt]{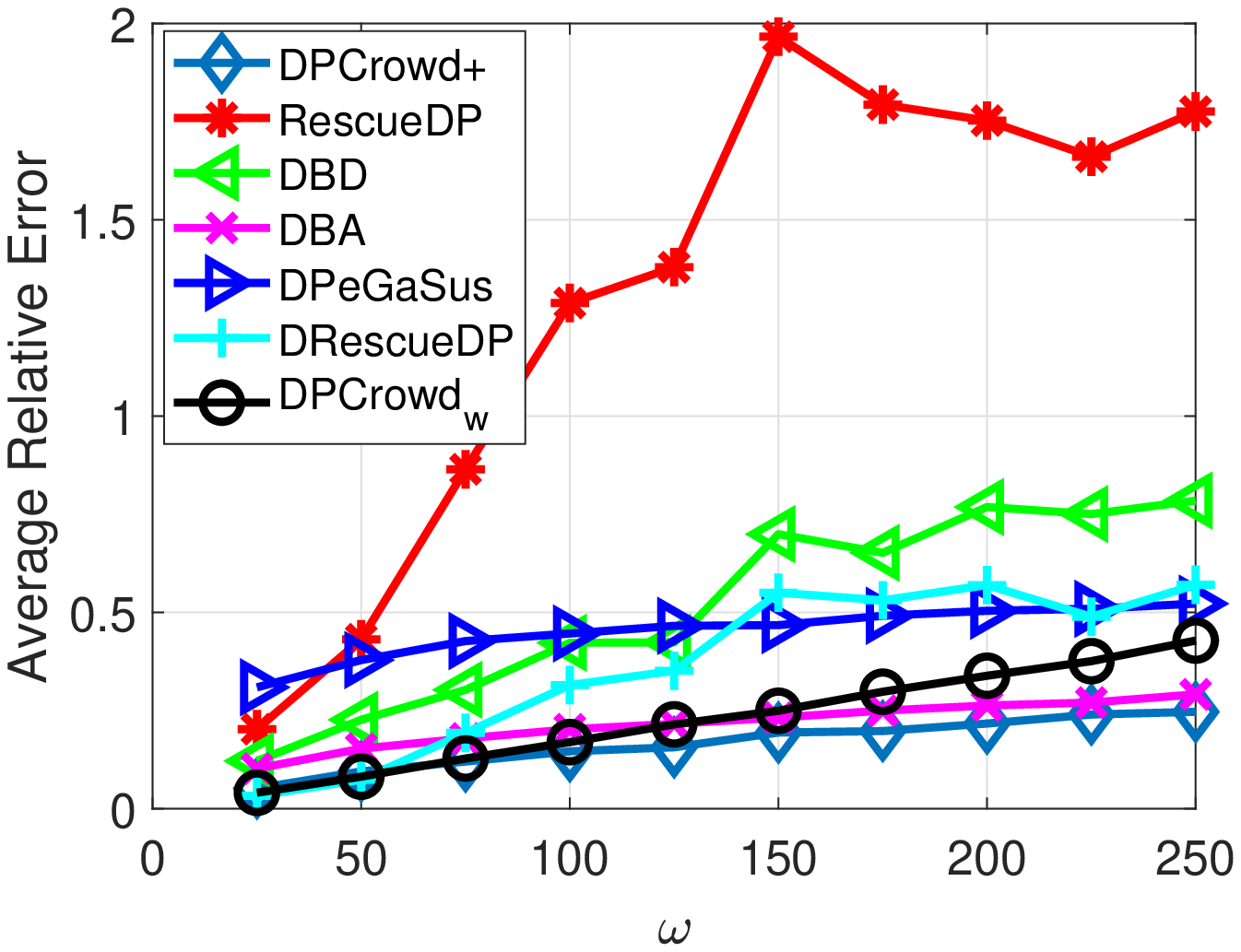}
		\label{ACEMLinear-w} 
		\includegraphics[width=125pt]{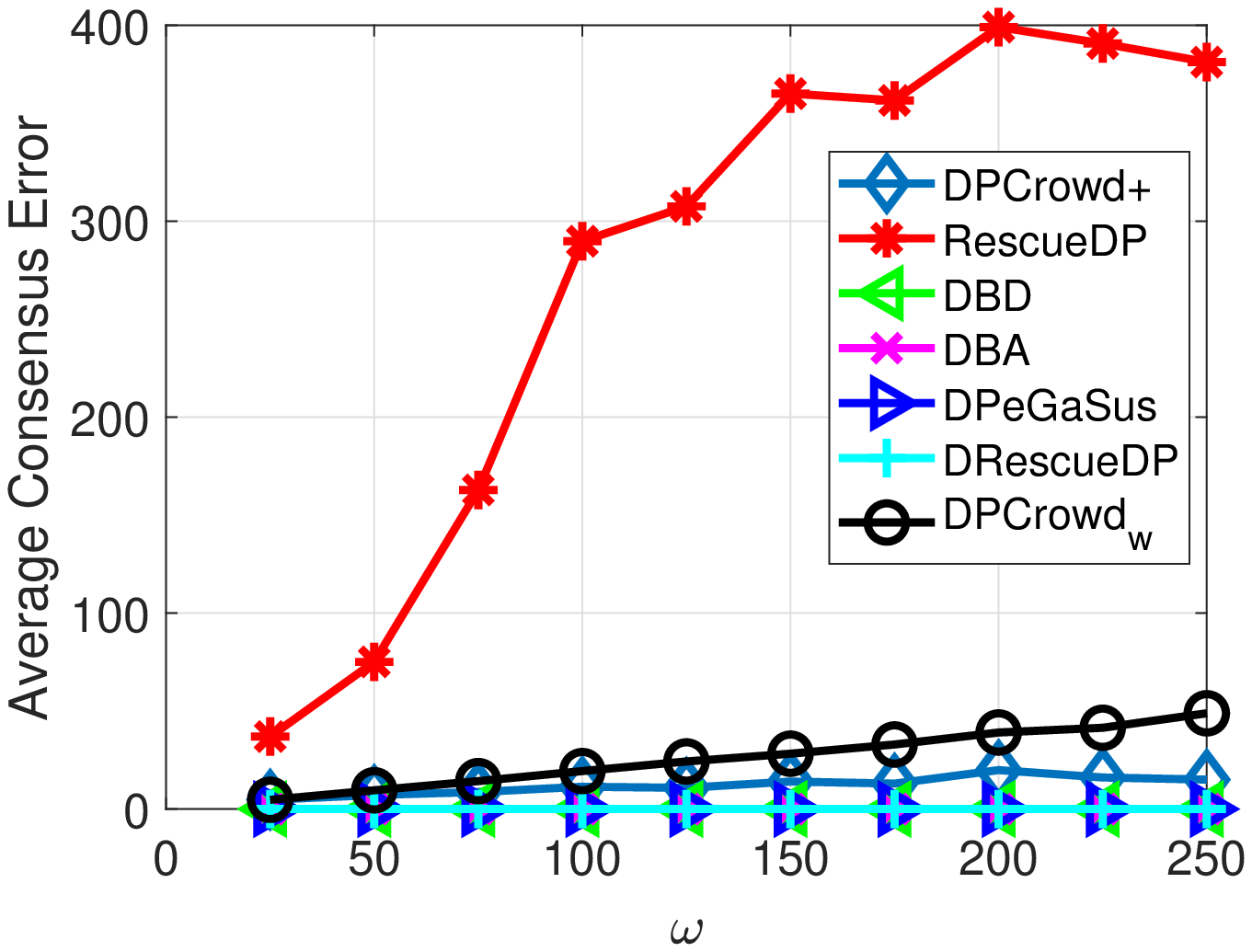}}
	\subfloat[\textsf{Multi-Flu}]{
		\label{AREMFlu-w} 
		\includegraphics[width=125pt]{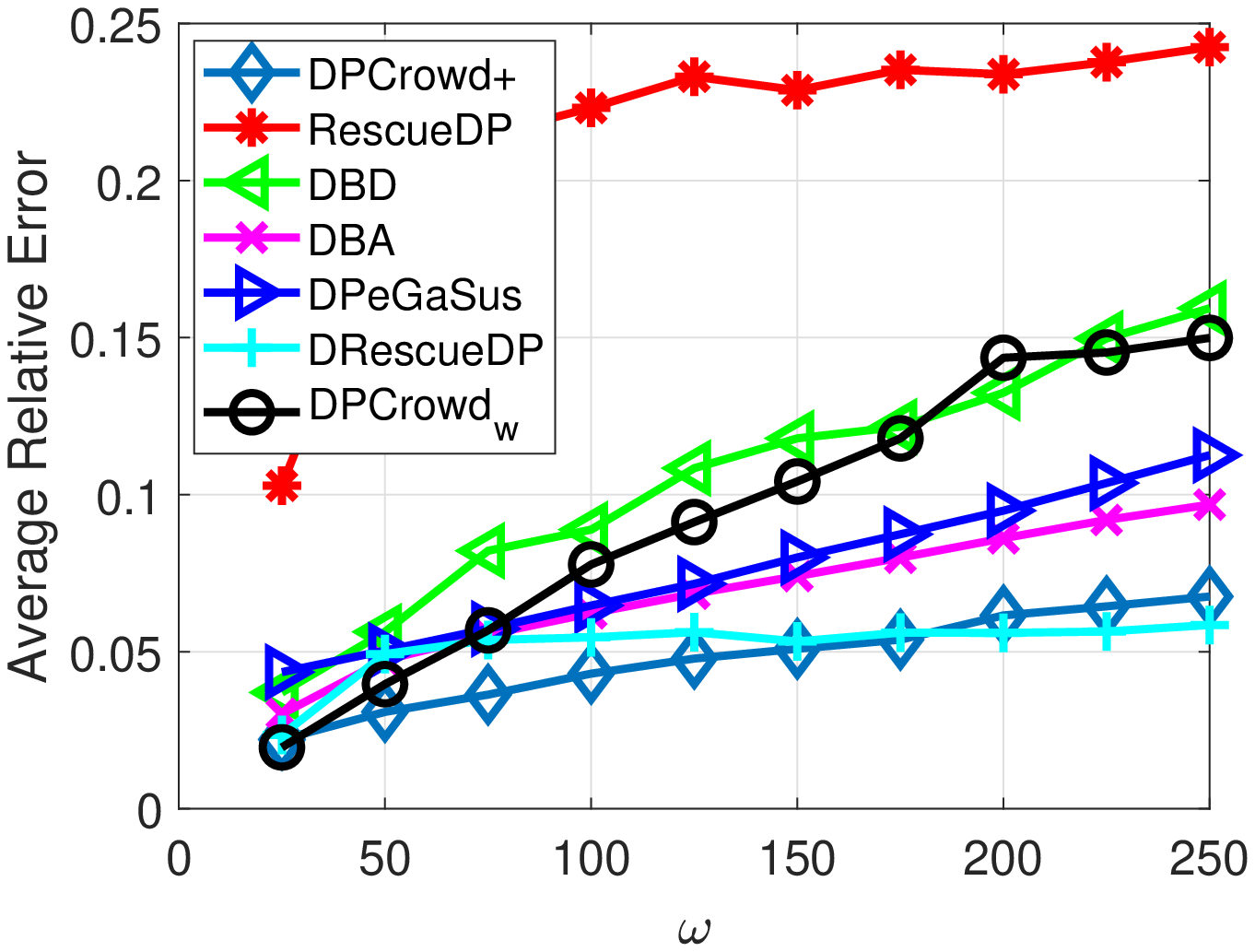}
		\label{ACEMFlu-w} 
		\includegraphics[width=125pt]{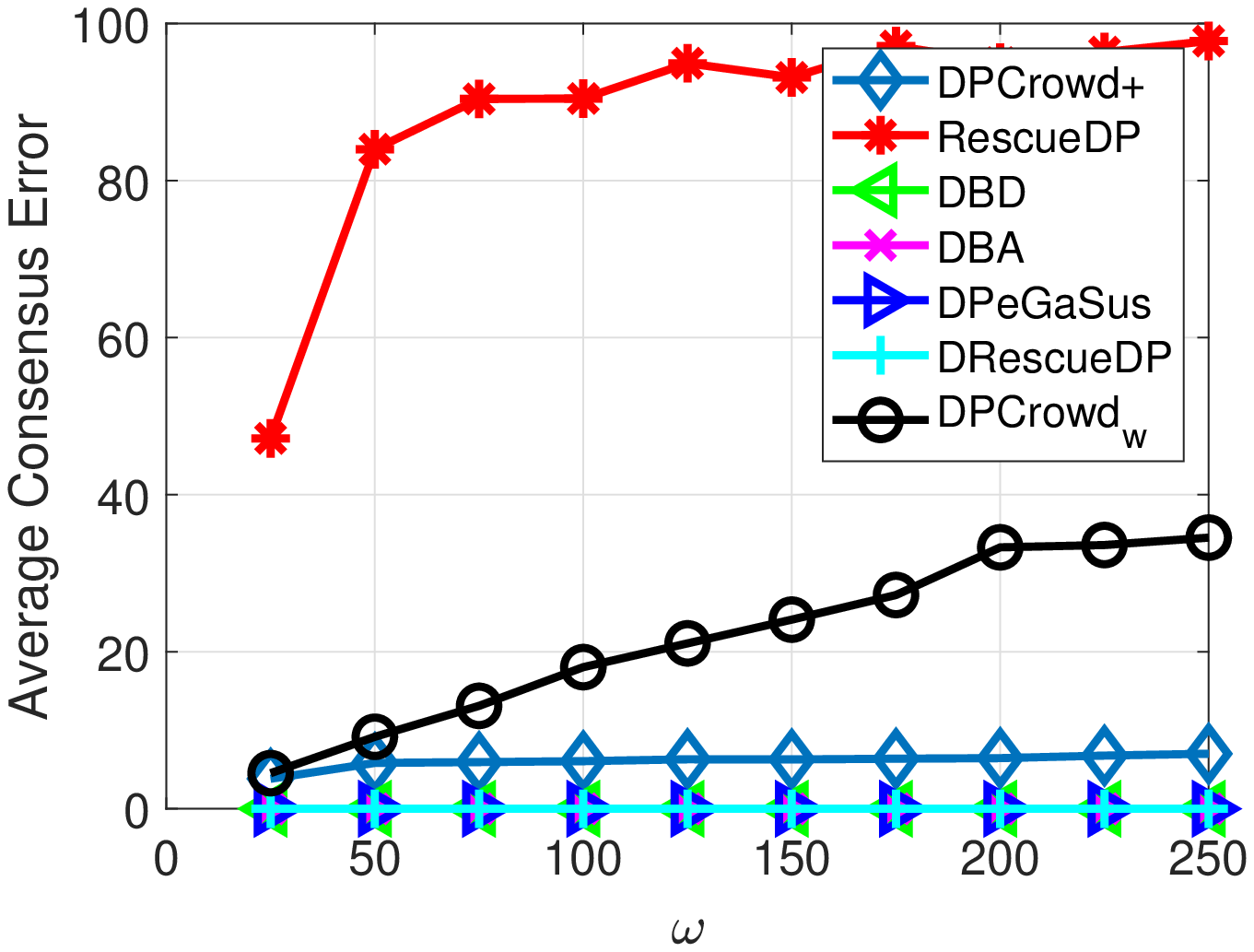}}
	\caption{Average Error vs. Window Size}\centering
	\label{fig:MUtilityVSw} 
	\vspace{-0.3cm}
\end{figure*}

\textbf{Tradeoff between Utility and Privacy:} Fig.~\ref{fig:MUtilityVsPrivacy} presents the estimation utility of all comparable schemes with respect to different privacy levels $\varepsilon$.
For various $\varepsilon$, the ARE of independent estimation scheme \textsf{RescueDP} is the largest due to no communications among distributed servers. 
Compared with the straightforward extension schemes that incurs great communication latency and overheads, \textsf{DPCrowd$_w$} enforces information exchanges of one-hop neighbors to collaboratively correct the estimation over the network with higher communication efficiency.
While \textsf{DPCrowd$_w$} can reduce the estimation error by implementing $w$-event privacy, which, however does not consider the sparsity in multi-dimensional streams. Instead, \textsf{DPCrowd}+ can further reduce the overdose noise on dimensions with small values by adopting the dynamic grouping strategy. Besides, compared with \textsf{DRescueDP} that directly average the whole-network estimations, \textsf{DPCrowd}+ can better fuse the neighboring estimations according to the estimation weights (Eq.~(\ref{eq:weight})). Thus, \textsf{DPCrowd}+ shows the smallest ARE among all aforementioned schemes, especially on \textsf{Multi-Flu}. This is because \textsf{Multi-Flu} has more dimensions and is much sparser than \textsf{Multi-Linear}.

In terms of consensus error, \textsf{DRescueDP}, \textsf{DBD/DBA}, and \textsf{DPeGaSus} have nearly no consensus error since expensive all-to-all communications are realized in the whole-network. The ACEs of \textsf{DPCrowd}+, \textsf{DPCrowd$_w$}, and the independent estimation scheme \textsf{RescueDP} drops with the increase of $\varepsilon$, which shows that it is easy to achieve consensus when less noise is added. \textsf{RescueDP} has the largest ACE because of no communication among servers. Instead, \textsf{DPCrowd$_w$} shows its superior since message exchange and collaborative correction is leveraged in the estimation. Furthermore, \textsf{DPCrowd}+ has much smaller consensus error as it combines the collaborative correction of in \textsf{DPCrowd} and dynamic grouping to enhance the utility for multi-dimensional data streams.

\begin{figure*}[htbp]\vspace{-0.3cm}
	\centering
	\subfloat[\textsf{Multi-Linear}]{
		\label{AREMLinear} 
		\includegraphics[width=125pt]{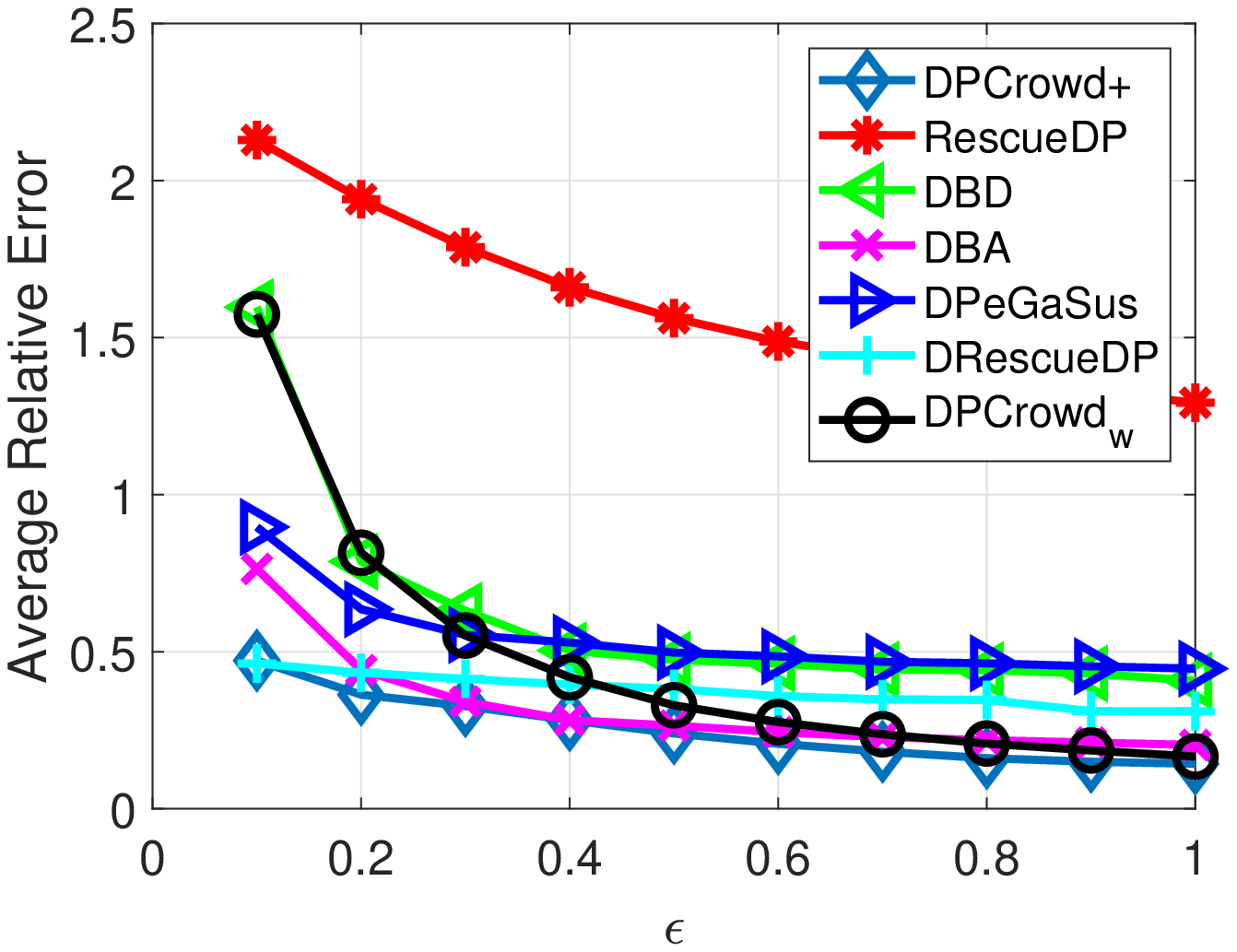}
		\label{ACEMLinear} 
		\includegraphics[width=125pt]{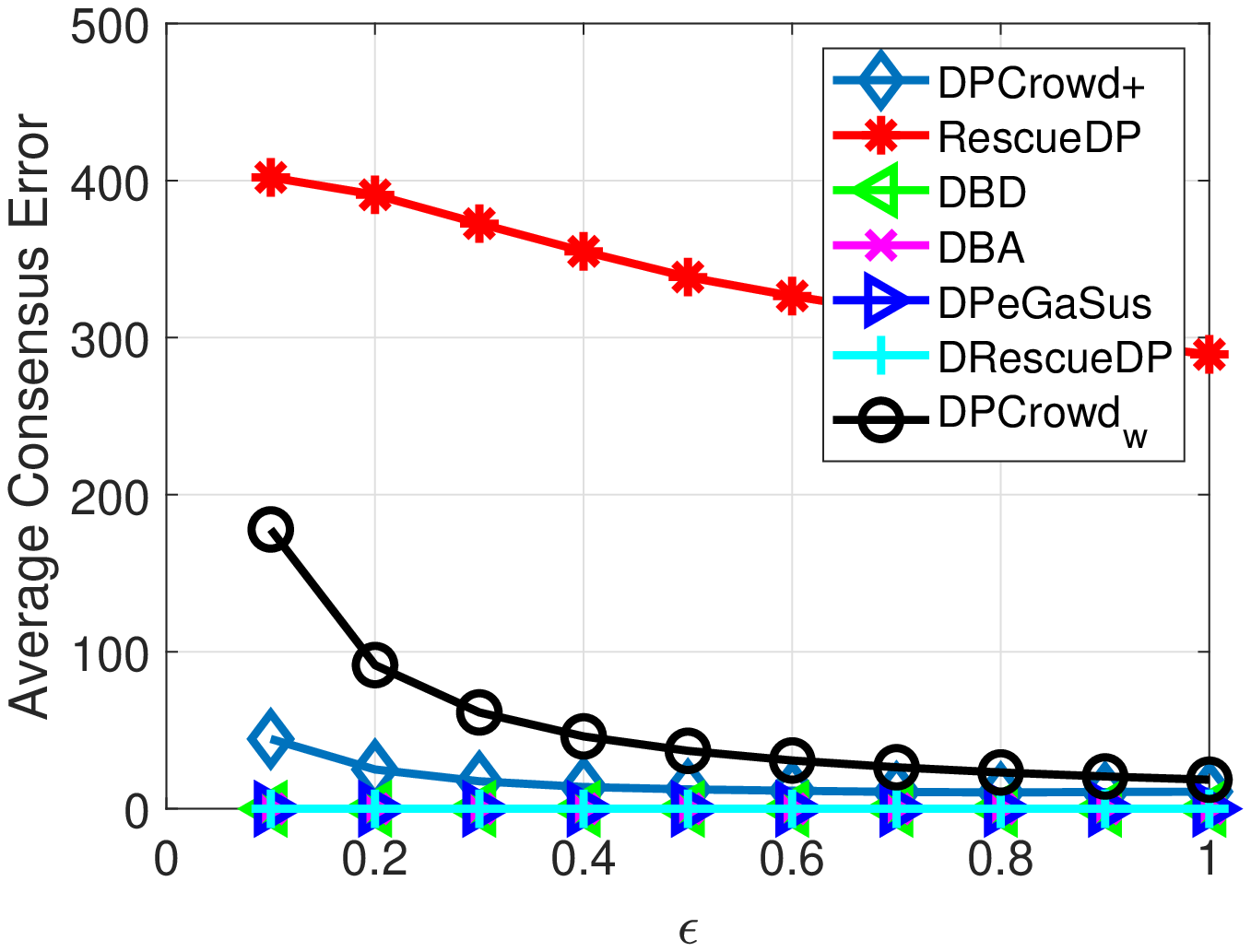}}
	\subfloat[\textsf{Multi-Flu}]{
		\label{AREMFlu} 
		\includegraphics[width=125pt]{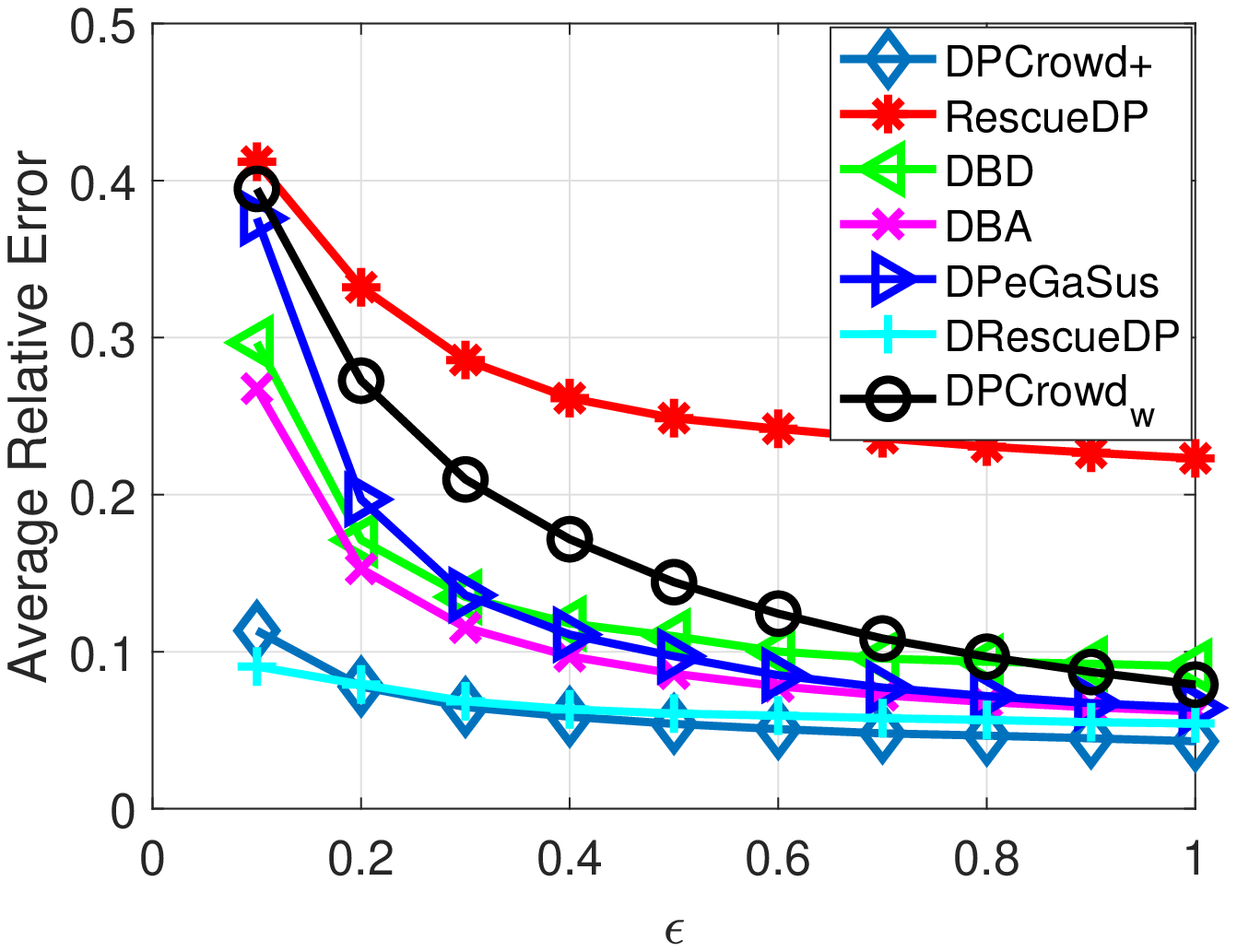}
		\label{ACEMFlu} 
		\includegraphics[width=125pt]{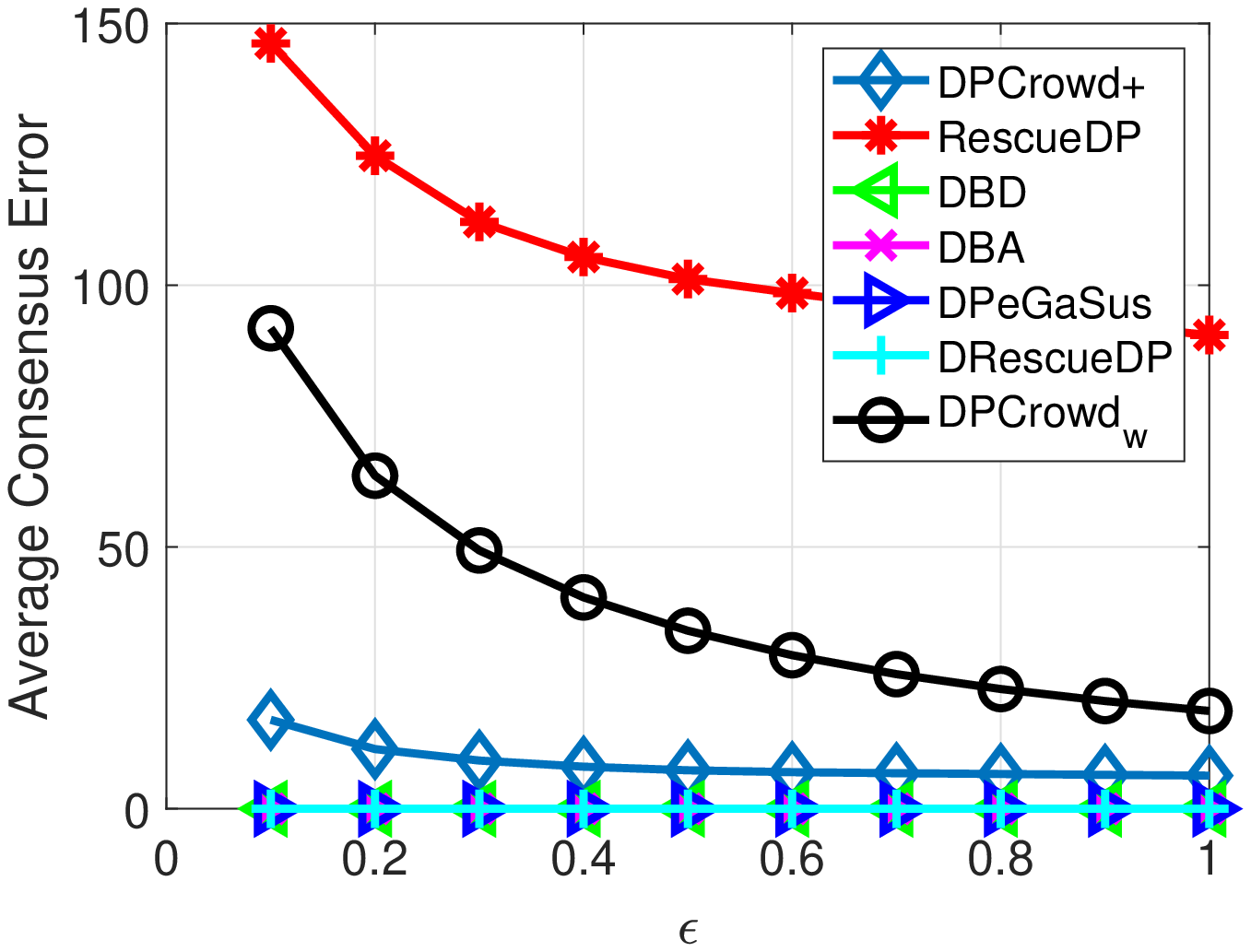}}
	\caption{Average Error vs. Privacy}\centering
	\label{fig:MUtilityVsPrivacy} 
	\vspace{-0.5cm}
\end{figure*}

\section{Final Remarks}\label{sec:conclusion}

In this paper, we have studied the framework of real-time statistical estimation for multiple distributed servers with crowd-sourced data in a decentralized setting, which enables data sharing and supports IoT-driven smart-world systems. Based on this framework, we first propose a novel scheme with both differential privacy preservation and communication efficiency, \textsf{DPCrowd}, for real-time decentralized statistical estimation on a finite crowd-sourced data stream. In specific, \textsf{DPCrowd} on distributed servers can achieve a consensus estimate of the true statistics by identifying the temporal correlations in data streams and exchanging the perturbed information intermittently with only one-hop neighbors. Additionally, as an extension for practical decentralized statistical estimation on infinite high-dimensional crowd-sourced data streams, we further propose \textsf{DPCrowd}+ to realize not only $w$-event DP, but also dimensional reduction by learning the sparse structure of multi-dimensional data. Extensive experimental results on real-world datasets show that our proposed schemes \textsf{DPCrowd} and \textsf{DPCrowd}+ are efficient and effective in obtaining accurate and consensus real-time statistical estimation for distributed servers on crowd-sourced data streams while guaranteeing sufficient DP for crowd-sourcing users.

\small
\bibliographystyle{IEEEtran}
\bibliography{refs}

\begin{thebibliography}{10}
\providecommand{\url}[1]{#1}
\csname url@samestyle\endcsname
\providecommand{\newblock}{\relax}
\providecommand{\bibinfo}[2]{#2}
\providecommand{\BIBentrySTDinterwordspacing}{\spaceskip=0pt\relax}
\providecommand{\BIBentryALTinterwordstretchfactor}{4}
\providecommand{\BIBentryALTinterwordspacing}{\spaceskip=\fontdimen2\font plus
\BIBentryALTinterwordstretchfactor\fontdimen3\font minus
  \fontdimen4\font\relax}
\providecommand{\BIBforeignlanguage}[2]{{%
\expandafter\ifx\csname l@#1\endcsname\relax
\typeout{** WARNING: IEEEtran.bst: No hyphenation pattern has been}%
\typeout{** loaded for the language `#1'. Using the pattern for}%
\typeout{** the default language instead.}%
\else
\language=\csname l@#1\endcsname
\fi
#2}}
\providecommand{\BIBdecl}{\relax}
\BIBdecl

\bibitem{zheng2020privacy}
X.~Zheng and Z.~Cai, ``Privacy-preserved data sharing towards multiple parties
  in industrial iots,'' \emph{IEEE J. Sel. Areas Commun.}, vol.~PP, pp. 1--12,
  2020.

\bibitem{8730298}
X.~Liu, C.~Qian, W.~G. Hatcher, H.~Xu, W.~Liao, and W.~Yu, ``Secure internet of
  things (iot)-based smart-world critical infrastructures: Survey, case study
  and research opportunities,'' \emph{IEEE Access}, vol.~7, pp.
  79\,523--79\,544, 2019.

\bibitem{7588230}
N.~Kumar, S.~Zeadally, and J.~J. P.~C. Rodrigues, ``Vehicular delay-tolerant
  networks for smart grid data management using mobile edge computing,''
  \emph{IEEE Commun. Mag.}, vol.~54, no.~10, pp. 60--66, 2016.

\bibitem{lin2017survey}
J.~Lin, W.~Yu, N.~Zhang, X.~Yang, H.~Zhang, and W.~Zhao, ``A survey on internet
  of things: Architecture, enabling technologies, security and privacy, and
  applications,'' \emph{IEEE Internet Things J.}, vol.~4, no.~5, pp.
  1125--1142, 2017.

\bibitem{huang2019practical}
P.~Huang, L.~Guo, M.~Li, and Y.~Fang, ``Practical privacy-preserving ecg-based
  authentication for iot-based healthcare,'' \emph{IEEE Internet Things J.},
  vol.~6, no.~5, pp. 9200--9210, 2019.

\bibitem{WH16}
H.~Wang, M.~S. M.~H. Fang, and C.~Wang, \emph{Wireless Health}, 2016.

\bibitem{braverman2016communication}
M.~Braverman, A.~Garg, T.~Ma, H.~L. Nguyen, and D.~P. Woodruff, ``Communication
  lower bounds for statistical estimation problems via a distributed data
  processing inequality,'' in \emph{Proc. ACM STOC}, 2016, pp. 1011--1020.

\bibitem{jordan2019communication}
M.~I. Jordan, J.~D. Lee, and Y.~Yang, ``Communication-efficient distributed
  statistical inference,'' \emph{J. Amer. Statist. Assoc.}, vol. 114, no. 526,
  pp. 668--681, 2019.

\bibitem{rescuedp2016}
Q.~Wang, Y.~Zhang, X.~Lu, Z.~Wang, Z.~Qin, and K.~Ren, ``Rescuedp: Real-time
  spatio-temporal crowd-sourced data publishing with differential privacy,'' in
  \emph{Proc. IEEE INFOCOM}, 2016, pp. 1--9.

\bibitem{wang2018privacy}
Z.~Wang, X.~Pang, Y.~Chen, H.~Shao, Q.~Wang, L.~Wu, H.~Chen, and H.~Qi,
  ``Privacy-preserving crowd-sourced statistical data publishing with an
  untrusted server,'' \emph{IEEE Trans. Mobile Comput.}, vol.~18, no.~6, pp.
  1356--1367, 2018.

\bibitem{zanella2014internet}
A.~Zanella, N.~Bui, A.~Castellani, L.~Vangelista, and M.~Zorzi, ``Internet of
  things for smart cities,'' \emph{IEEE Internet Things J.}, vol.~1, no.~1, pp.
  22--32, 2014.

\bibitem{li2018privacy}
X.~Li, S.~Liu, F.~Wu, S.~Kumari, and J.~J. Rodrigues, ``Privacy preserving data
  aggregation scheme for mobile edge computing assisted iot applications,''
  \emph{IEEE Internet Things J.}, vol.~6, no.~3, pp. 4755--4763, 2018.

\bibitem{sonehara2011isolation}
N.~Sonehara, I.~Echizen, and S.~Wohlgemuth, ``Isolation in cloud computing and
  privacy-enhancing technologies,'' \emph{Business \& information systems
  engineering}, vol.~3, no.~3, p. 155, 2011.

\bibitem{zhang2018reap}
Z.~Zhang, S.~He, J.~Chen, and J.~Zhang, ``Reap: An efficient incentive
  mechanism for reconciling aggregation accuracy and individual privacy in
  crowdsensing,'' \emph{IEEE Trans. Inf. Forensics Security}, vol.~13, no.~12,
  pp. 2995--3007, 2018.

\bibitem{asadi2014survey}
A.~Asadi, Q.~Wang, and V.~Mancuso, ``A survey on device-to-device communication
  in cellular networks,'' \emph{IEEE Commun. Surveys Tuts.}, vol.~16, no.~4,
  pp. 1801--1819, 2014.

\bibitem{stankovic2010decentralized}
S.~S. Stankovic, M.~S. Stankovic, and D.~M. Stipanovic, ``Decentralized
  parameter estimation by consensus based stochastic approximation,''
  \emph{IEEE Trans. Autom. Control}, vol.~56, no.~3, pp. 531--543, 2010.

\bibitem{kar2012distributed}
S.~Kar, J.~M. Moura, and K.~Ramanan, ``Distributed parameter estimation in
  sensor networks: Nonlinear observation models and imperfect communication,''
  \emph{IEEE Trans. Inf. Theory}, vol.~58, no.~6, pp. 3575--3605, 2012.

\bibitem{DworkRoth-77}
C.~Dwork and A.~Roth, ``The algorithmic foundations of differential privacy,''
  \emph{Foundations \& Trends\circledR in Theoretical Computer Science},
  vol.~9, no.~3, 2013.

\bibitem{hassan2019differential}
M.~U. Hassan, M.~H. Rehmani, and J.~Chen, ``Differential privacy techniques for
  cyber physical systems: a survey,'' \emph{IEEE Commun. Surveys Tuts}, 2019.

\bibitem{su2016differentially}
S.~Su, P.~Tang, X.~Cheng, R.~Chen, and Z.~Wu, ``Differentially private
  multi-party high-dimensional data publishing,'' in \emph{Proc. IEEE ICDE},
  2016, pp. 205--216.

\bibitem{yang2017survey}
X.~Yang, T.~Wang, X.~Ren, and W.~Yu, ``Survey on improving data utility in
  differentially private sequential data publishing,'' \emph{IEEE Trans. Big
  Data}, pp. 1--17, 2017.

\bibitem{ren2018textsf}
X.~Ren, C.-M. Yu, W.~Yu, S.~Yang, X.~Yang, J.~A. McCann, and S.~Y. Philip,
  ``Lopub: High-dimensional crowdsourced data publication with local
  differential privacy,'' \emph{IEEE Trans. Inf. Forensics Security}, vol.~13,
  no.~9, pp. 2151--2166, 2018.

\bibitem{wang2019locally}
T.~Wang, X.~Yang, X.~Ren, W.~Yu, and S.~Yang, ``Locally private
  high-dimensional crowdsourced data release based on copula functions,''
  \emph{IEEE Trans. Services Comput.}, pp. 1--16, 2019.

\bibitem{fan2014adaptive}
L.~Fan and L.~Xiong, ``An adaptive approach to real-time aggregate monitoring
  with differential privacy,'' \emph{IEEE Trans. Knowl. Data Eng}, vol.~26,
  no.~9, pp. 2094--2106, 2014.

\bibitem{WangZYLYRS19}
T.~Wang, J.~Zhao, H.~Yu, J.~Liu, X.~Yang, X.~Ren, and S.~Shi,
  ``Privacy-preserving crowd-guided {AI} decision-making in ethical dilemmas,''
  in \emph{Proc. ACM CIKM}, 2019, pp. 1311--1320.

\bibitem{HuangMitra-1130}
Z.~Huang, S.~Mitra, and G.~Dullerud, ``Differentially private iterative
  synchronous consensus,'' in \emph{Proc. ACM WPES}, 2012, pp. 81--90.

\bibitem{huang2015differentially}
Z.~Huang, S.~Mitra, and N.~Vaidya, ``Differentially private distributed
  optimization,'' in \emph{Proc. ACM ICDCN}, 2015, pp. 1--10.

\bibitem{8260919}
C.~Li, P.~Zhou, L.~Xiong, Q.~Wang, and T.~Wang, ``Differentially private
  distributed online learning,'' \emph{IEEE Trans. Knowl. Data Eng.}, vol.~30,
  no.~8, pp. 1440--1453, 2018.

\bibitem{Mcsherry-39}
F.~D. McSherry, ``Privacy integrated queries: an extensible platform for
  privacy-preserving data analysis,'' in \emph{Proc. ACM SIGMOD}, 2009, pp.
  19--30.

\bibitem{dwork2010differentialco}
C.~Dwork, M.~Naor, T.~Pitassi, and G.~N. Rothblum, ``Differential privacy under
  continual observation,'' in \emph{Proc. ACM STOC}, 2010, pp. 715--724.

\bibitem{Dwork-405}
C.~Dwork, ``Differential privacy,'' in \emph{Proc. ICALP}, 2006, pp. 1--12.

\bibitem{dwork2010differential}
------, ``Differential privacy in new settings,'' in \emph{Proc. ACM-SIAM
  SODA}, 2010, pp. 174--183.

\bibitem{mir2011pan}
D.~Mir, S.~Muthukrishnan, A.~Nikolov, and R.~N. Wright, ``Pan-private
  algorithms via statistics on sketches,'' in \emph{Proc. ACM PODS}, 2011, pp.
  37--48.

\bibitem{chan2012differentially}
T.-H.~H. Chan, M.~Li, E.~Shi, and W.~Xu, ``Differentially private continual
  monitoring of heavy hitters from distributed streams,'' in \emph{Proc.
  PETS}.\hskip 1em plus 0.5em minus 0.4em\relax Springer, 2012, pp. 140--159.

\bibitem{chen2017pegasus}
Y.~Chen, A.~Machanavajjhala, M.~Hay, and G.~Miklau, ``Pegasus: Data-adaptive
  differentially private stream processing,'' in \emph{Proc. ACM CCS}, 2017,
  pp. 1375--1388.

\bibitem{kellaris2014differentially}
G.~Kellaris, S.~Papadopoulos, X.~Xiao, and D.~Papadias, ``Differentially
  private event sequences over infinite streams,'' \emph{Proceedings of the
  VLDB Endowment}, vol.~7, no.~12, pp. 1155--1166, 2014.

\bibitem{wang2016tdsc}
Q.~Wang, Y.~Zhang, X.~Lu, Z.~Wang, Z.~Qin, and K.~Ren, ``Real-time and
  spatio-temporal crowd-sourced social network data publishing with
  differential privacy,'' \emph{IEEE Trans. Dependable Secure Comput.},
  vol.~15, no.~4, pp. 591--606, 2016.

\bibitem{acs2011have}
G.~{\'A}cs and C.~Castelluccia, ``I have a dream!(differentially private smart
  metering),'' in \emph{International Workshop on Information Hiding}.\hskip
  1em plus 0.5em minus 0.4em\relax Springer, 2011, pp. 118--132.

\bibitem{7286780}
S.~Goryczka and L.~Xiong, ``A comprehensive comparison of multiparty secure
  additions with differential privacy,'' \emph{IEEE Trans. Dependable Secure
  Comput.}, vol.~14, no.~5, pp. 463--477, 2017.

\bibitem{AlhadidiMohammed-1153}
D.~Alhadidi, N.~Mohammed, B.~C. Fung, and M.~Debbabi, ``Secure distributed
  framework for achieving $\varepsilon$-differential privacy,'' in \emph{Proc.
  PETS}, 2012, pp. 120--139.

\bibitem{hong2015collaborative}
Y.~Hong, J.~Vaidya, H.~Lu, P.~Karras, and S.~Goel, ``Collaborative search log
  sanitization: Toward differential privacy and boosted utility,'' \emph{IEEE
  Trans. Dependable Secure Comput.}, vol.~12, no.~5, pp. 504--518, 2015.

\bibitem{Erlingsson-2014}
{\'U}.~Erlingsson, V.~Pihur, and A.~Korolova, ``Rappor: Randomized aggregatable
  privacy-preserving ordinal response,'' in \emph{Proc. ACM CCS}, 2014, pp.
  1054--1067.

\bibitem{BelmegaSankar-1182}
E.~V. Belmega, L.~Sankar, and H.~V. Poor, ``Enabling data exchange in two-agent
  interactive systems under privacy constraints,'' \emph{IEEE J. Sel. Topics
  Signal Process.}, vol.~9, no.~7, pp. 1285--1297, 2015.

\bibitem{truex2019hybrid}
S.~Truex, N.~Baracaldo, A.~Anwar, T.~Steinke, H.~Ludwig, R.~Zhang, and Y.~Zhou,
  ``A hybrid approach to privacy-preserving federated learning,'' in
  \emph{Proc. ACM AISec@CCS}, 2019, pp. 1--11.

\bibitem{zhao2019privacy}
L.~Zhao, Q.~Wang, Q.~Zou, Y.~Zhang, and Y.~Chen, ``Privacy-preserving
  collaborative deep learning with unreliable participants,'' \emph{IEEE Trans.
  Inf. Forensics Security}, vol.~15, pp. 1486--1500, 2019.

\bibitem{geyer2017differentially}
R.~C. Geyer, T.~Klein, and M.~Nabi, ``Differentially private federated
  learning: A client level perspective,'' \emph{arXiv preprint
  arXiv:1712.07557}, 2017.

\bibitem{ZhangZhu-1222}
T.~Zhang and Q.~Zhu, ``Dynamic differential privacy for admm-based distributed
  classification learning,'' \emph{IEEE Trans. Inf. Forensics Security},
  vol.~12, no.~1, pp. 172--187, 2017.

\bibitem{wang2019beyond}
Z.~Wang, M.~Song, Z.~Zhang, Y.~Song, Q.~Wang, and H.~Qi, ``Beyond inferring
  class representatives: User-level privacy leakage from federated learning,''
  in \emph{Proc. IEEE INFOCOM}.\hskip 1em plus 0.5em minus 0.4em\relax IEEE,
  2019, pp. 2512--2520.

\bibitem{li2016differential}
N.~Li, M.~Lyu, D.~Su, and W.~Yang, ``Differential privacy: From theory to
  practice,'' \emph{Synthesis Lectures on Information Security, Privacy, \&
  Trust}, vol.~8, no.~4, pp. 1--138, 2016.

\bibitem{cao2018quantifying}
Y.~Cao, M.~Yoshikawa, Y.~Xiao, and L.~Xiong, ``Quantifying differential privacy
  in continuous data release under temporal correlations,'' \emph{IEEE Trans.
  Knowl. Data Eng.}, vol.~31, no.~7, pp. 1281--1295, 2018.

\bibitem{mcsherrypost1}
F.~McSherry, ``Differential privacy and correlated data,''
  \url{https://github.com/frankmcsherry/blog/blob/master/posts/2016-08-29.md}.

\bibitem{song2017pufferfish}
S.~Song, Y.~Wang, and K.~Chaudhuri, ``Pufferfish privacy mechanisms for
  correlated data,'' in \emph{Proc. ACM CIKM}, 2017, pp. 1291--1306.

\bibitem{kang2019incentive}
J.~Kang, Z.~Xiong, D.~Niyato, S.~Xie, and J.~Zhang, ``Incentive mechanism for
  reliable federated learning: A joint optimization approach to combining
  reputation and contract theory,'' \emph{IEEE Internet Things J.}, vol.~6,
  no.~6, pp. 10\,700--10\,714, 2019.

\bibitem{olfati2009kalman}
R.~Olfati-Saber, ``Kalman-consensus filter: Optimality, stability, and
  performance,'' in \emph{Proc. IEEE CDC}, 2009, pp. 7036--7042.

\bibitem{fan2014monitoring}
L.~Fan, L.~Bonomi, L.~Xiong, and V.~Sunderam, ``Monitoring web browsing
  behavior with differential privacy,'' in \emph{Proc. ACM WWW}, 2014, pp.
  177--188.

\end{thebibliography}

\end{document}